\pdfoutput=1
\documentclass[11pt]{article}
\usepackage{style}
\usepackage{shortcuts}

\begin{document}

\title{Min-Plus Convolution Lower Bounds\\via a Higher-Order BSG Theorem}
\author{%
    Nick Fischer\footnote{Max Planck Institute for Informatics. E-Mail: \url{nfischer@mpi-inf.mpg.de}} \and
    Ce Jin\footnote{University of California, Berkeley. E-Mail: \url{cejin@berkeley.edu}} \and
    Yinzhan Xu\footnote{University of California, San Diego. E-Mail: \url{xyzhan@ucsd.edu}}}
\date{}
\maketitle

\begin{abstract}
    \noindent
\emph{Min-Plus Convolution} is a central problem in fine-grained complexity, and the associated Min-Plus Convolution Hypothesis forms the basis for a wide range of conditional lower bounds for fundamental problems. It is closely connected to the APSP and 3SUM Hypotheses, and in fact implies both, making it a unifying hypothesis for two of the main pillars of the area. In this work we establish several strong results related to Min-Plus Convolution. We design a \emph{universe reduction}, showing, under a plausible additive combinatorics assumption, that the Min-Plus Convolution Hypothesis is \emph{equivalent} to the \emph{Strong} Min-Plus Convolution Hypothesis. We also obtain tight conditional lower bounds for multiple long-standing problems, including Min-Max Convolution and Bounded Monotone Min-Plus Convolution.

Our approach is inspired by Fischer's recent equivalence between several variants of APSP [STOC '26], but extending that technique to the arithmetic setting requires overcoming deep obstacles. To this end, we develop a novel additive structure theorem that can be viewed as a higher-order substitute of the Balog--Szemerédi--Gowers (BSG) theorem, allowing us to extract strong additive structure even from weakly structured sets. Building on this structural result, we show that certain structured 3SUM instances (namely, sets with low \emph{rank}) can be solved in truly subquadratic time. This algorithm forms the main algorithmic ingredient in our reductions. Besides, it generalizes \emph{all} previously known truly subquadratic-time special cases of 3SUM, and is therefore of independent interest.

\end{abstract}
\thispagestyle{empty}
\clearpage

\tableofcontents
\setcounter{page}{0}
\thispagestyle{empty}
\clearpage

\section{Introduction} \label{sec:intro}
The \emph{Min-Plus Convolution} problem is one of the central problems in fine-grained complexity. Given two integer sequences $A$ and $B$ of length $n$, the task is to compute the sequence $A * B$ where $(A * B)[k] = \min_i (A[i] + B[k - i])$. Naively this takes time $O(n^2)$, and the fastest known algorithm achieves only mild improvements, $n^2 / 2^{\Omega(\sqrt{\log n})}$~\cite{Williams18}. This leads to the \emph{Min-Plus Convolution Hypothesis}~\cite{CyganMWW19,KunnemannPS17} postulating that truly subquadratic time $O(n^{2-\epsilon})$, for some constant $\epsilon > 0$, is impossible.\footnote{Unless otherwise specified, the entry bounds of all problems we consider are polynomially bounded in the input size. The postulated lower bounds hold in the Word-RAM model with $\Theta(\log n)$-bit words, and hold against potentially randomized algorithms.} This hypothesis has been used to establish conditional lower bounds for problems including Knapsack~\cite{CyganMWW19,KunnemannPS17}, Unbounded Subset Sum and the Frobenius problem~\cite{Klein22}, Integer Programming~\cite{JansenR23}, Tree Sparsity~\cite{BackursIS17}, hop-bounded shortest paths~\cite{KociumakaP23,VassilevskaXXZ25}, and an approximate version of 3SUM~\cite{BringmannGK24}.

Beyond its many applications, the Min-Plus Convolution Hypothesis plays a central conceptual role in the landscape of fine-grained complexity. It is closely connected to the famous \emph{3SUM} and \emph{APSP} Hypotheses, and in fact implies them both~\cite{BremnerCDEHILPT14,BackursIS17}. Thus, it can be viewed as a unified hypothesis of two of the main pillars of the area. In particular, the Min-Plus Convolution Hypothesis is one of the strongest well-established hypotheses, implying more conditional lower bounds than most other fine-grained assumptions. This makes studying Min-Plus Convolution a fundamental problem in its own right.

In this work, we establish several strong results that significantly sharpen our understanding of Min-Plus Convolution. We design a \emph{universe reduction}: assuming a plausible hypothesis from additive combinatorics, we show that the Min-Plus Convolution Hypothesis is equivalent to its \emph{strong} variant, which asserts hardness even for instances with integer weights bounded by $O(n)$. As a consequence, we obtain tight conditional lower bounds for a variety of long-standing convolution problems with ``intermediate'' complexity, including Min-Max Convolution and Bounded Monotone Min-Plus Convolution.

\subsection{Starting Point}
Our starting point is a recent breakthrough by Fischer~\cite{Fischer26}, which establishes analogous results for the APSP problem. Specifically, Fischer shows that the APSP Hypothesis and the Strong APSP Hypothesis are equivalent, assuming a plausible hypothesis from additive combinatorics and that $\omega = 2$. He further shows that the Directed Unweighted APSP Hypothesis, stating that APSP in unweighted directed graphs requires $n^{2.5-o(1)}$ time, is in this equivalence class. This greatly sharpens our understanding of the role of the universe size in APSP.

Given that Min-Plus Convolution and APSP are similar in many ways, one might wonder whether Fischer's technique transfers. In particular, can it be applied to show the equivalence between the Min-Plus Convolution Hypothesis and the Strong Min-Plus Convolution Hypothesis?

At a high level, Fischer~\cite{Fischer26} follows a structure-versus-random approach that could, in principle, apply to any problem $P$, and that has already been successful in other fine-grained reductions~\cite{AbboudBKZ22,AbboudBF23,JinX23,ChanVX23}. It can be summarized in three generic steps:
\begin{enumerate}[label=(\arabic*)]
    \item Identify a suitable measure of \emph{structure} for $P$.
    \item Show that $P$ can be solved fast on \emph{structured} instances. 
    \item Then, suppose that we aim to design a fine-grained reduction from $P$ to some other problem~$Q$. Given any instance of $P$, we first repeatedly extract structured subinstances, and solve these subinstances with the fast algorithm from step 2. Once there is no more structured subinstance, we regard the remaining instance as \emph{pseudo-random}, and exploit the pseudo-randomness in the reduction to $Q$.
\end{enumerate}
Of course, successfully implementing this approach hinges crucially on the choice of the structure notion. It has to be weak enough so that structured instances can be solved efficiently, but strong enough so that the corresponding pseudo-random instances can be exploited in the reduction. 

In Fischer's setup, $P$ is the APSP problem (or rather, for technical reasons, the equivalent \emph{Min-Plus Product} problem), and the relevant notion of structure is a matrix rank parameter known as \emph{select-plus rank}\footnote{The \emph{select-plus rank} of an $n \times n$ integer matrix $A$ is the smallest number $r$ such that there exist matrices $U, V$ of sizes $n \times r$ and $r \times n$, respectively, so that $A[i, j] \in \bigcup_{k=1}^r \set{U[i, k] + V[k, j]}$.}. A major technical contribution is a fast algorithm for instances of low rank, from which the universe reduction and further lower bounds follow via the general framework.

Let us return to Min-Plus Convolution. As the analogy to Min-Plus Product is strong, a priori one may hope that Fischer's techniques can be extended to Min-Plus Convolution without too much difficulty. However, our situation is quite the exception. Obtaining the analogous results for Min-Plus Convolution poses several substantially new obstacles that are not present in the matrix world. First, it is entirely unclear what the appropriate notion of structure should be for sequences or sets. This naturally falls into the realm of additive combinatorics; however, any such notion must capture \emph{extremely weak} additive structure, far from the commonly used strong structure measures. This leads to a second major challenge: the existing tools from additive combinatorics are not well suited for this regime, and we need to invent our own structural results. All in all, our route involves three steps, each of which is interesting in its own right:
\begin{itemize}
    \item We develop a new additive structure theorem that can be viewed as a ``higher-order'' substitute of the Balog--Szemerédi--Gowers theorem (\cref{sec:intro:additive-comb}).
    \item We then identify the appropriate structure notion, called \emph{rank} in analogy to~\cite{Fischer26}, and show that low-rank instances of Min-Plus Convolution can be solved in subquadratic time, based on the new structure theorem. In fact, we directly work with the harder 3SUM problem (\cref{sec:intro:low-rank-3sum}). 
    \item Equipped with this new algorithm, we derive the aforementioned universe reduction and strong lower bounds from Min-Plus Convolution in analogy to Fischer's proof (\cref{sec:intro:sec:univ-reduct,sec:intro:sec:intermediate}).
\end{itemize}

\subsection{A New Additive Structure Theorem} \label{sec:intro:additive-comb}
Using additive combinatorics in fine-grained algorithm design and complexity is an exciting recent trend~\cite{ChanL15,BringmannW21,BringmannN20,BringmannN21,AbboudBF23,JinX23,ChanVX23,Jin24,Bringmann24,ChenLMZ24a,BringmannFN25,FischerJX25,AbboudFJVX25,ChenMZ26,Fischer26}. This was pioneered by Chan and Lewenstein~\cite{ChanL15}, who introduced the celebrated Balog--Szemerédi--Gowers (BSG) Theorem to our community, and derived various strong algorithmic consequences of it. The BSG Theorem has since found many more applications~\cite{ChanL15,AbboudBF23,JinX23,ChanVX23,FischerJX25,AbboudFJVX25}, notably also in Fischer's work~\cite{Fischer26}, and became an important part of the fine-grained toolkit. To understand both the strengths and the limitations of the BSG Theorem---and how we overcome them---let us take a short excursus into additive combinatorics.

Intuitively, the BSG theorem states that any set $A$ with a lot of ``local structure'', e.g., in the form of many solutions to the equation $a_1 + a_2 = a_3 + a_4$ for $a_1, a_2, a_3, a_4 \in A$, must admit a large subset $H$ that is ``globally structured'', in the sense that the \emph{sumset} $H + H = \set{a + b : a, b \in H}$ is small. Such sets $H$ are also called \emph{low-doubling} sets since their \emph{doubling} \smash{$\frac{|H+H|}{|H|}$} is small. To state the BSG Theorem quantitatively, let us introduce a piece of notation that will turn out convenient later:
\begin{equation*}
    \Pop^k(A; p) = \set*{ x : \abs{\set{(a_1, \dots, a_k) \in A^k : x = a_1 + \dots + a_k}} \geq p},
\end{equation*}
i.e., the set of all sums $a_1 + \dots + a_k$ that can be expressed in at least $p$ distinct ways. We refer to these as the \emph{popular $k$-sums}. Then the BSG theorem can be equivalently stated in the following (non-standard) form:\footnote{To relate \cref{thm:bsg-reformulated} to the usual BSG Theorem (e.g., \cref{thm:bsg-symm}), note that $A$ has a popular 4-sum, $\Pop^4(A; |A|^3 / K) \neq \emptyset$, if and only if $A$ has large \emph{additive energy}, $E(A) \geq |A|^3 / K'$, where $K$ and $K'$ depend polynomially on each other. The ``if'' direction is immediate by Cauchy--Schwarz. The ``only if'' direction follows from some basic Fourier analysis and Hölder's inequality (see e.g., \cite[Lemma 2.11]{JinX23}).}

\begin{theorem}[BSG Theorem, Reformulated] \label{thm:bsg-reformulated}
Let $A \subseteq \Int$ be a set with $\Pop^4(A; |A|^3 / K) \neq \emptyset$. Then there is a set $H \subseteq \Int$ satisfying that:
\begin{itemize}
    \item $|H| \leq |A|$ (in fact, we can take $H$ as a subset of $A$),
    \item $|H + H| \leq K^{O(1)} |H|$,
    \item $|A \cap H| \geq |A| / K^{O(1)}$.
\end{itemize}
\end{theorem}

In other words, the existence of a popular 4-sum certifies the presence of a large low-doubling subset.

Unfortunately, the BSG Theorem is insufficient to deal with sets exhibiting much weaker additive structure. Consider the following example. Let~$A$ be a random size-$n$ set in a large universe~$\set{0, \dots, u}$, say for $u = n^{10}$. In this case, with good probability, each 4-sum is unique. In particular, there are no popular 4-sums. This is consistent with the fact that $A$ contains no large low-doubling subset. Nevertheless, $A$ is not entirely unstructured: it is contained in a much larger low-doubling set, namely the entire universe $\set{0, \dots, u}$. The challenge is that this weak global structure is invisible to the 4-sums. Our key insight is that such weak structure can be detected by moving to larger $k$-sums. Specifically, we show that for sufficiently large $k$, the presence of a popular $k$-sum reveals the desired structure. This leads to the following higher-order analogue of the BSG Theorem.

\begin{restatable}[Popular $k$-Sums Imply Low-Doubling Structure]{theorem}{thmksumpop} \label{thm:k-sum-pop}
For every $c > 0$ and $\epsilon > 0$, there is some $k = k(c, \epsilon)$ such that the following holds. Let $A \subseteq \Int$ and assume that $\Pop^k(A; p) \neq \emptyset$ for some $p = \Omega(|A|^{k-1-c})$. Then there exists a set $H \subseteq \Int$ and some $K = O_{c, \epsilon}(|A|^\epsilon)$ such that:
\begin{itemize}
    \item $|H| \leq K |A|^{1+c}$,
    \item $|H + H| \leq K |H|$,
    \item $|A \cap H| \geq |A| / K$.
\end{itemize}
\end{restatable}

To illustrate how to apply the theorem, consider again the example from before. We apply \cref{thm:k-sum-pop} with $c = 9$, and for some arbitrarily small constant $\epsilon$. The theorem then chooses some constant $k = k(c, \epsilon)$. There is indeed some $k$-sum in $A$ that is $\Omega(|A|^{k-1-c})$-popular (as there are $|A|^k$ $k$-sums taking at most $k |A|^{10}$ distinct values), satisfying the premise of the theorem. \cref{thm:k-sum-pop} then implies that a large part of $A$ is contained in some set $H$ of size at most $K |A|^{10}$ and doubling $|H + H| \leq K |H|$. Here, the doubling $K = O_{c, \epsilon}(|A|^\epsilon)$ is controlled by the initial choice of $\epsilon$. I.e., we can force $H$ to be arbitrarily well-structured (as long as $K$ is polynomially large), at the only cost of increasing the constant $k$. \cref{thm:k-sum-pop} turns out to be a surprisingly powerful tool that is likely to find more applications beyond our specific interest related to the Min-Plus Convolution problem.

The proof of \cref{thm:k-sum-pop} builds on the \emph{asymmetric} BSG Theorem of Tao and Vu~\cite{TaoV06}, a significantly more intricate extension of BSG to pairs of sets of different sizes. In fact, our result is the first in fine-grained complexity, and even theoretical computer science as a whole, that exploits the asymmetric BSG Theorem, to our knowledge.

At a very high level, the proof of \cref{thm:k-sum-pop} is by increasing $k$ gradually and tracking the sets of popular $\ell$-sums for $\ell = 4, 5, \dots, k$. By a pigeonhole argument, at some point, the popular $\ell$-sums and $(\ell+1)$-sums correlate (both in size and structurally), at which point the asymmetric BSG theorem can be applied to extract the desired low-doubling set. We provide a detailed proof in \cref{sec:pop-sums}. We also face the extra challenge that for our algorithmic applications, it does not suffice to know that $H$ exists, but we must be able to \emph{compute} $H$. Therefore, along the way we also devise an efficient randomized algorithm to find $H$ as in \cref{thm:k-sum-pop} (see \cref{thm:k-sum-pop-algo}). Notably, to this end we also had to design an efficient algorithm for the asymmetric BSG Theorem (see \cref{sec:asym-bsg}). This black box will certainly be handy for other algorithms in the future.

Finally, let us remark that there are some other higher-order variants of BSG that have been studied in additive combinatorics~\cite{SudakovSV05,BorensteinC11,Shkredov14,Lee25}. However, unfortunately none of these existing results applies in the weakly structured regime required for our applications (though some of the same ideas, e.g.\ as in~\cite{BatemanK11,Shkredov14}, turn out useful for the proof of \cref{thm:k-sum-pop} as well). This concludes our additive combinatorics excursus.

\subsection{Low-Rank 3SUM} \label{sec:intro:low-rank-3sum}
Equipped with \cref{thm:k-sum-pop} we are now in a position to start implementing the structure-versus-randomness approach following Fischer~\cite{Fischer26}. Recall from before that the foundation of Fischer's proof is to identify a certain \emph{rank} parameter, and to design a truly subcubic-time algorithm for computing the min-plus product of \emph{low-rank} matrices. An immediate challenge of extending Fischer's proof to Min-Plus Convolution is the lack of such a rank notion. It is unclear what should be the rank of a sequence or a set with respect to normal convolutions, let alone Min-Plus Convolution. 

A main contribution of our paper is to identify a suitable notion of rank for sequences, which is strong enough (so that unstructured instances are sufficiently pseudo-random to complete our reductions) but not overpowering (so we can still design a fast algorithm for Min-Plus Convolution on low-rank sequences).

In fact, we will directly work with \emph{3SUM}, a harder problem than Min-Plus Convolution. Given three sets $A, B, C$, it asks to decide if there is a solution to $a + b = c$ for $a \in A, b \in B, c \in C$. 3SUM can be solved naively in time $O(n^2)$, and the \emph{3SUM Hypothesis} speculates that it cannot be solved in truly subquadratic time. This hypothesis is among the ``three big'' fine-grained assumptions. As mentioned before, Min-Plus Convolution reduces tightly to 3SUM, and for this reason we prove a stronger statement by showing that low-rank 3SUM is in subquadratic time. Another advantage is that the structure measure can be defined for sets rather than sequences, which turns out cleaner:

\begin{definition}[Rank] \label{defn:low-rank}
The \emph{rank} of a set $A \subseteq \Int$, denoted $r(A)$, is defined as the smallest integer $r \geq 0$ such that there is a set $X \subseteq \Int$ of size $r$ satisfying $A \subseteq X + X$.
\end{definition}

This notion has been considered before in the additive combinatorics literature, see, e.g.,~\cite{Erdos77,Bukh25}.\footnote{In~\cite{Erdos77,Bukh25}, no explicit name is given for the measure in \cref{defn:low-rank}. Our name \emph{rank} is inspired both from the connection to~\cite{Fischer26}, and since, coincidentally, the name \emph{$[2]$-sumset rank} has already been used in some previous work~\cite{BulteauFRV15}.} Observe that the rank of $A$ ranges from $\ceil{\sqrt{|A|}}$ to $|A|+1$. Sets with low rank are considered structured; for instance intervals or arithmetic progressions have rank close to $\sqrt{|A|}$. At the other extreme, random-like sets have rank close to $|A|$. We show that for sets with rank up to~$|A|^{1-\epsilon}$, 3SUM can be solved in subquadratic time:

\begin{theorem}[Informal \cref{thm:3sum-low-rank}] \label{thm:informal-low-rank-3sum}
Let $\epsilon > 0$. Any 3SUM instance $(A, B, C)$ where one of the sets has rank at most $O(n^{1-\epsilon})$ can be solved in time $O(n^{2-\epsilon'})$ for some $\epsilon' > 0$.
\end{theorem}

Interestingly, \cref{thm:informal-low-rank-3sum} generalizes and unifies \emph{all} previously known special cases of 3SUM that are known to be in truly subquadratic time:
\begin{itemize}
    \item \emph{Small-Universe 3SUM:} When $A, B, C \subseteq [u]$, 3SUM can be solved in time $\tO(u)$ by the Fast Fourier Transform. This is truly subquadratic when $u \leq n^{2-\epsilon}$. In this case, all three sets have rank $O(\sqrt u) \leq n^{1-\epsilon/2}$ (witnessed by $X = \{0, 1, \ldots, \lceil \sqrt{u}\rceil -1\} \cup \{\lceil \sqrt{u}\rceil, 2\lceil \sqrt{u}\rceil, \ldots, \lceil \sqrt{u}\rceil^2\}$). 
    \item \emph{Small-Doubling 3SUM:} As a vast generalization to Small-Universe 3SUM, Abboud, Bringmann, and Fischer \cite{AbboudBF23} and Jin and Xu \cite{JinX23} considered 3SUM on inputs with $|A+A|\le K|A|$, and gave truly subquadratic time algorithm provided $K\le n^\delta$ where $\delta>0$ is sufficiently small.\footnote{A related problem, though significantly simpler, is to solve 3SUM when $|A + B|$ or $|A - B|$ is small; see~\cite{BringmannFN22,Fischer24}.} For example, \cite{JinX23}'s algorithm for this case runs in $\tO(Kn^{1.5})$ time, which is truly subquadratic provided $K\le n^{0.5-\epsilon}$.
    As an improvement to their algorithms, our \cref{thm:informal-low-rank-3sum} implies a truly subquadratic 3SUM algorithm whenever $K\le n^{1-\epsilon}$.
    
    This implication follows from the following simple fact: if $|A+A|\le |A|^{2-\delta}$, then there exists a set $X$ of size $|X|\le \tO(|A|^{1-\delta/2})$ such that $A\subseteq X+X$. Such $X$ exists by a simple probabilistic argument: Pick random subsets $Y\subseteq -A$ and $Z\subseteq A+A$ with sizes $|Y|,|Z| = \Theta(|A|^{1-\delta/2}\sqrt{\log |A|})$ and let $X=Y\cup Z$. Since each $a\in A$ is represented as $a=y+z$ by at least $|A|$ distinct pairs $(y,z)\in (-A)\times (A+A)$, with high probability at least one pair $(y,z)$ survives in $Y\cup Z$ for every $a\in A$.
    \item \emph{Clustered 3SUM}. Chan and Lewenstein \cite{ChanL15} studied \emph{clustered} 3SUM, where a set is clustered if it can be covered by $n^{1-\delta}$ intervals of length $n$ for some $\delta
     > 0$. They showed a truly subquadratic-time algorithm for this case. It is simple to show that if $A$ is clustered, then a low-rank decomposition of $A$ can be found efficiently: Let $S$ be the set of the smallest numbers in each of the intervals. Then notice that $A \subseteq X + X$ where $X = (S + \{0, 1, \ldots, n^{\delta / 2}\}) \cup \{n^{\delta / 2}, 2 n^{\delta / 2}, \ldots, n\}$, which has size $O(n^{1-\delta/2})$.\footnote{Chan and Lewenstein also studied a generalization of clustered 3SUM where the inputs are from $\Z^d$ for constant~$d$. 3SUM with inputs from $\Z^d$ is essentially the same as 3SUM with integer inputs: Given a 3SUM instance with inputs from $\Z^d$, we can map each input vector to an integer by viewing the vector as the $d$-digit representation of the integer in some large enough base. In the clustered case, when $A \subseteq S + H$, where $|S| \le n^{1-\delta}$ and $H$ is a hypercube of volume at most $n$, their algorithm also runs in truly subquadratic time. Similar as the integer case, we can write $H$ as the sum of $Y + Z$ where $Y$ has size $O(n^\delta)$ and $Z$ has $O(n^{1-\delta})$ (this is possible because $H$ is a hypercube). Hence, letting $X = (S + Y) \cup Z$ guarantees that $A \subseteq X + X$ and $|X| \le O(n^{1-\delta / 2})$.}
\end{itemize}

\subsection{Universe Reduction for Min-Plus Convolution} \label{sec:intro:sec:univ-reduct}
Having stated our two main technical innovations, let us finally turn to our results involving Min-Plus Convolution. We start with our universe reduction. Here, by \emph{universe} $[u]$ we refer to the range of the entries in the sequences $A$ and $B$. Recall that the Min-Plus Convolution Hypothesis postulates that the min-plus convolution $A * B$ of two given sequences $A, B$ cannot be computed in subquadratic time. In this hypothesis, as is common in fine-grained complexity, the assumption is that the universe is polynomially bounded, $u = n^{O(1)}$. A stronger hypothesis, the \emph{Strong} Min-Plus Convolution Hypothesis, postulates that subquadratic time is impossible even for sequences over the universe $[n]$:

\begin{hypothesis}[Strong Min-Plus Convolution]
Min-Plus Convolution over the universe $[n]$ cannot be solved in time $O(n^{2-\epsilon})$ for any constant $\epsilon > 0$.
\end{hypothesis}

For even smaller universes, however, there are faster algebraic algorithms. Specifically, the Min-Plus Convolution problem can be reduced to computing the sumset of two integer sets in $\set{0, \dots, n u}$, leading to an alternative $\tO(n u)$-time algorithm using the FFT. This algorithm differs substantially from the naive brute-force algorithm, suggesting that in this small-universe regime ($u \ll n$) the Min-Plus Convolution problem might behave quite differently than in the large-universe regime ($u \gg n$). Thus, a natural question is whether this $\tO(n u)$ could be improved, perhaps by some more elaborate algebraic technique? Investigating this question is also highly relevant for the Min-Plus Convolution problem as a source of hardness, as many lower bounds under the Min-Plus Convolution Hypothesis could be strengthened if we could assume hardness for small universes. We successfully resolve this question, proving that time $\tO(n u)$ is best-possible across all ranges of~$u \leq n$:

\begin{restatable}[Small-Universe Reduction]{theorem}{smalluniverseminplus}
\label{thm:small-universe-min-plus}
    Let $\gamma \in [0, 1]$. Min-Plus Convolution for length-$n$ input sequences over a universe of size $u = n^{\gamma}$ cannot be solved in time $O(n^{1+\gamma - \eps})$ (for any constant~$\eps > 0$), unless the Strong Min-Plus Convolution Hypothesis fails.
\end{restatable}

Even more ambitiously, can one also show a reduction from the large-universe regime? Is there a universe reduction from, say, $u = n^{10}$ to $u = n$? We \emph{almost} answer this question in the affirmative. Unfortunately, in analogy with~\cite{Fischer26}, there remains one key step in our proof for which the state-of-the-art additive combinatorics tools are not strong enough yet. For this reason, Fischer encapsulated the missing structural statement as the following plausible hypothesis. Let us call a function $h : X \to \Int$ \emph{sum-order-preserving} if, for every $x_1, x_2, y_1, y_2 \in X$, $x_1 + x_2 < y_1 + y_2$ implies $h(x_1) + h(x_2) < h(y_1) + h(y_2)$.

\begin{hypothesis}[\cite{Fischer26}]
\label{hypo:sum-preserving}
    There is a constant $c$ such that the following holds. For every set of integers $X$ with doubling $|X + X| \le K|X|$, there is a subset $Y \subseteq X$ of size $|Y| \ge |X| / K^c$ and a sum-order-preserving function $h: Y \rightarrow \{0, 1, \ldots, |X|\}$. Moreover, given $X$, one can compute $Y$ and~$h$ in time $|X|^{1+o(1)} K^{O(1)}$. 
\end{hypothesis}

Intuitively, \cref{hypo:sum-preserving} speculates that every low-doubling integer set can be compressed to an interval in such a way that the order of pairwise sums is preserved. This statement is closely related to the \emph{Polynomial Freiman--Ruzsa} (PFR) Conjecture for integers, which, if true, would almost resolve \cref{hypo:sum-preserving}. Currently, the state-of-the-art bounds~\cite{Raghavan25} together with the work of~\cite{AmirkhanyanBC18} imply that one can take $\exp((\log K)^{1+o(1)})$ instead of $K^c$ in \cref{hypo:sum-preserving}. We therefore hope that \cref{hypo:sum-preserving} will eventually be resolved alongside the PFR Conjecture. We refer to~\cite{Fischer26} for a more detailed discussion of this hypothesis and its plausibility. 

Under \cref{hypo:sum-preserving} we do obtain the full universe reduction:

\begin{restatable}[Large-Universe Reduction]{theorem}{thmstrongminplus}
\label{thm:strong-min-plus-conv-equal}
    Under \cref{hypo:sum-preserving}, the Strong Min-Plus Convolution Hypothesis and the Min-Plus Convolution Hypothesis are equivalent. 
\end{restatable}

In summary, under \cref{hypo:sum-preserving}, \cref{thm:small-universe-min-plus,thm:strong-min-plus-conv-equal} fully characterize the complexity of Min-Plus Convolution as $\min\set{n^2, n u}^{1\pm o(1)}$ for all universe sizes $u = n^{\Theta(1)}$. If this algorithm can be improved in any regime, then the Min-Plus Convolution Hypothesis fails.

\subsection{Application to Intermediate Convolution Problems} \label{sec:intro:sec:intermediate}
Fischer~\cite{Fischer26} also showed a list of APSP-based tight conditional lower bounds for ``intermediate'' problems, a well-studied class of matrix problems whose running times are about $n^{2.5}$ when $\omega = 2$. Likewise, our techniques also allow us to show matching conditional lower bounds for a range of intermediate sequence problems, whose best known running times are $\tO(n^{1.5})$.

\paragraph{Min-Max Convolution. } In the Min-Max Convolution problem, we are given sequences $A, B$ of length $n$, and need to compute a sequence $C$ where $C[k] = \min_{i + j = k} \max\{A[i], B[j]\}$. Bringmann, K{\"{u}}nnemann and Węgrzycki \cite{BringmannKW19} showed that this problem is equivalent to approximate Min-Plus Convolution.  An $\tO(n^{1.5})$ time algorithm for computing Min-Max Convolution was known as early as 1989 \cite{Kosaraju89}. However, no further improvements were known since then, and there is also no strong evidence why this running time might be optimal. We show a tight conditional lower bound for Min-Max Convolution. 

\begin{restatable}[Min-Max Convolution]{theorem}{thmminmaxconv} \label{thm:min-max-conv}
Min-Max Convolution cannot be solved in time $O(n^{3/2-\epsilon})$ (for any constant $\epsilon > 0$), unless the Strong Min-Plus Convolution Hypothesis fails.
\end{restatable}

\paragraph{Min-Equality Convolution. } In the Min-Equality Convolution problem, we are given sequences $A, B$ of length $n$, and need to compute a sequence $C$ where
\begin{equation*}
    C[k] = \min_{\substack{i + j = k \\ A[i] = B[j]}} A[i].
\end{equation*}
Previously, it was studied by Chan, Vassilevska Williams and Xu \cite{ChanVX23}, who showed that Min-Equality Convolution requires $n^{12/11-o(1)}$ time, under the Strong Min-Plus Convolution Hypothesis (they also showed an $n^{7/6-o(1)}$ lower bound based on the Strong APSP Hypothesis and an $n^{5/4-o(1)}$ lower bound based on the unweighted directed APSP Hypothesis, but neither matches the $\tO(n^{1.5})$ upper bound). In our work, we significantly improve their result, by showing a tight lower bound based on the Strong Min-Plus Convolution Hypothesis. 

\begin{restatable}[Min-Equality Convolution]{theorem}{thmminequalconv} \label{thm:min-equal-conv}
Min-Equality Convolution cannot be solved in time $O(n^{3/2-\epsilon})$ (for any constant $\epsilon > 0$), unless the Strong Min-Plus Convolution Hypothesis fails.
\end{restatable}

\paragraph{Bounded Monotone Min-Plus Convolution. } Bounded Monotone Min-Plus Convolution is a version of Min-Plus Convolution with structured inputs. A length-$n$ sequence is called bounded monotone if all entries are from $[O(n)]$, and the sequence is monotonically nondecreasing. Chan and Lewenstein \cite{ChanL15} were the first to study Bounded Monotone Min-Plus Convolution: they showed that when both input sequences are bounded monotone, Min-Plus Convolution can be computed in $O(n^{1.859})$ time. This running time has been improved to $\tO(n^{1.5})$ by Chi, Duan, Xie and Zhang~\cite{ChiDXZ22}, and Bringmann, D{\"{u}}rr and Polak \cite{BringmannDP24} generalized their algorithm to the case where only one of the sequences needs to be bounded monotone. Bounded-Monotone Min-Plus Convolution has found a variety of applications, including Jumbled Indexing~\cite{ChanL15} and several variants of knapsack~\cite{BringmannC22,BringmannDP24,ChenLMZ25}. We show that \cite{BringmannDP24}'s algorithm is conditionally optimal:

\begin{restatable}[Bounded Monotone Min-Plus Convolution]{theorem}{thmminplusconvmono} \label{thm:min-plus-conv-mono}
Min-Plus Convolution between an arbitrary sequence and a monotone sequence bounded by $[n]$ cannot be solved in time $O(n^{3/2-\epsilon})$ (for any constant $\epsilon > 0$), unless the Strong Min-Plus Convolution Hypothesis fails.
\end{restatable}

\paragraph{Min-Witness Convolution. }  Given two length-$n$ Boolean sequences $A, B$, their Min-Witness Convolution $C$ is defined such as $C[k]$ equals the minimum index $i$ where $A[i]$ and $B[k - i]$ are both true. This problem has been listed in papers studying fine-grained complexity of intermediate problems~\cite{LincolnPV20, ChanVX23} as a natural target for proving conditional lower bound for, but there has been no nontrivial lower bound so far for this problem. We prove the first super-linear (albeit non-tight) conditional lower bound for Min-Witness Convolution:

\begin{restatable}[Min-Witness Convolution]{theorem}{thmminwitnessconv} \label{thm:min-witness-conv}
Min-Witness Convolution cannot be solved in time $O(n^{4/3-\epsilon})$ (for any constant $\epsilon > 0$), unless the Strong Min-Plus Convolution Hypothesis fails.
\end{restatable}

A careful reader might have noticed that our lower bounds for intermediate sequence problems are based on the Strong Min-Plus Convolution Hypothesis, while Fischer \cite{Fischer26}'s lower bounds for intermediate matrix problems are based on the APSP Hypothesis itself rather than its strong variant. Without relying on \cref{hypo:sum-preserving}, Fischer first reduced APSP to a version of Min-Plus Product where the universe has low-doubling (rather than simply an interval $[u]$), and then obtained APSP-based lower bounds for intermediate matrix problems from this first step. We also manage to prove an analogous version of the first step in Min-Plus Convolution (\cref{lem:doubling-reduct}); unfortunately, it seems hard to obtain tight reductions for intermediate sequence problems from it. Intuitively, one reason why it is harder to design reductions to intermediate sequence problems is the following: In Fischer's universe reduction for APSP, both the inner dimension  and the universe of the Min-Plus Product are reduced, which make it easier to further reduce to intermediate matrix problems. In contrast, for sequence problems, there is no analogue of an ``inner dimension'', so there is less room to work with in downstream reductions. 

\subsection{Proof Ideas}
In this section we give a quick high-level overview for the results we have not sketched before.

\subsubsection{Low-Rank 3SUM in Subquadratic Time}
We start with the ideas behind our subquadratic-time algorithm for low-rank 3SUM. Let $(A, B, C)$ denote a given 3SUM instance, and assume that $C \subseteq X + X$ for some small set, $|X| \leq n^{1-\epsilon}$. At the heart of our algorithm lies a curious win-win paradigm: Either we can solve the instance by a clever brute-force algorithm, or we can \emph{amplify} structure.

Specifically, to find a 3SUM solution $a + b = c$ we could try to enumerate \emph{all} solutions
\begin{equation} \label{eq:meet-in-the-middle}
    a + b = x_1 + x_2, \qquad \text{for $a \in A, b \in B, x_1, x_2 \in X$.}
\end{equation}
For each such solution we can test if $a + b = x_1 + x_2$ is actually an element in $C$; only then have we found a solution. This might seem like overkill---after all we are listing many more solutions than necessary. But the crucial advantage is that solutions to this equation can be listed very efficiently. Indeed, rewriting each solution as
\begin{equation*}
    a - x_1 = x_2 - b,
\end{equation*}
a natural meet-in-the-middle algorithm is to preprocess the sets $A - X$ and $X - B$ in time $O(|A| \, |X| + |B| \, |X|) = O(n^{2-\epsilon})$. Afterwards, we can list all solutions in time $O(N)$, where $N$ denotes the number of solutions. If $N \leq n^{2-\delta}$, say, then this algorithm runs in subquadratic time. But of course, $N$ might be much larger, possibly up to $|X|^2 n = n^{3-2\epsilon}$, and in that case the meet-in-the-middle strategy is doomed.

Our key insight is that, in this situation, the sets $A$ and $B$ admit amplified structure. Concretely, focus on $B$. For simplicity, assume that all elements $b \in B$ participate in roughly the same number of solutions to \cref{eq:meet-in-the-middle}, i.e., in about $N / |B|$ solutions. Then, in other words, each $b \in B$ is an $N / |B| \geq n^{1-\delta}$-popular sum in the set $X + X - A$. Note that we started from the structural information that $C \subseteq X + X$, but ended up with the information that $B$ is covered by the $n^{1-\delta}$-popular sums in $X + X - A$. It turns out that this argument can be iterated. By arguing as above, but exchanging the roles of $B$ and $C$, we can either solve the instance, or further derive that $C$ is covered by the $n^{2-2\delta}$-popular sums in $X + X - A + A$. And iterating this argument $k$ times, we eventually infer that $C$ is covered by the $n^{2k-2k\delta}$-popular sums in $X + X - kA + kA$. 

At this point, $C$ witnesses an $n^{2k-2k\delta}$-popular sum in $X + X - kA + kA$. This is where our structure theorem (\cref{thm:k-sum-pop}) comes into play. When carefully applied, it implies that there is some low-doubling set $H$ of size at about $|X|^2 n^{2k\delta} \leq n^{2-2\epsilon+2k\delta}$ that covers a large fraction of~$A$. Here, intuitively think of $\delta$ as an extremely small constant, so that the $2k \delta$ term effectively becomes~$0$. That is, intuitively, we have showed that $A$ is covered by a low-doubling set of size $|X|^2$. All in all, it is quite intriguing \emph{how} the structure amplifies: We start from a weakly structured set~$C$, then infer some more structure for $B$, then again for $C$, and so on. And eventually, we conclude that $A$ admits a small low-doubling container. Besides, it is interesting that for a 3SUM algorithm the relevant structure lies in the number of $k$-sums.

Of course, we are not completely done. What remains is a 3SUM instance $(A, B, C)$ equipped with a container set $H \supseteq A$ of size $|H| \leq n^{2-\epsilon}$ and with arbitrarily low doubling. But this instance can be dealt with following ideas that have already appeared in earlier work (notably,~\cite{JinX23}), and for this reason we defer further details to \cref{sec:3sum-low-rank}.

\subsubsection{Universe Reduction for Min-Plus Convolution}
Finally, we give a brief glimpse into our new universe reduction. Only at this point will it become clear why the rank of a set is a reasonable notion in our context.

For concreteness, focus on the universe reduction from $[n]$ to $[\sqrt n]$. That is, we argue that if Min-Plus Convolution over the universe $[\sqrt n]$ is in time $O(n^{3/2-\epsilon})$, then Min-Plus Convolution over the universe $[n]$ is in time $O(n^{2-\epsilon'})$ (contradicting the Strong Min-Plus Convolution Hypothesis). Let $A, B$ denote the given length-$n$ sequences over the universe $[n]$. We rewrite the sequences into their base-$\sqrt n$ representation as
\begin{align*}
    A &= A_1 \cdot q + A_0, \\
    B &= B_1 \cdot q + B_0,
\end{align*}
where the four sequences $A_0, B_0, A_1, B_1$ are each over the universe $[\sqrt n]$. Let us call an index $i$ a \emph{witness} of $k$ if $(A * B)[k] = A[i] + B[k - i]$, and a \emph{pseudo-witness} if $(A_1 * B_1)[k] = A_1[i] + B_1[k - i]$. Equivalently, the pseudo-witnesses are those indices that would become indices when rounding $A, B$ to multiples of $\sqrt n$. By standard approaches, we can assume $A_0$ and $B_0$ have entries less than $\sqrt{n} / 2$, so that a witness must be a pseudo-witness. 

We will now execute the structure-versus-randomness paradigm that we sketched in the beginning. The notion of \emph{pseudo-randomness} turns out to concern the number of pseudo-witnesses: in a pseudo-random instance, all indices $k$ have few pseudo-witnesses, say at most $n^{1/2+\epsilon/2}$. This is indeed a good case for our reduction, as we can afford to list \emph{all} pseudo-witnesses. Notably, this list also contains the true witnesses. To obtain the list of pseudo-witnesses, we apply a simple sampling trick (\cref{lem:min-plus-conv-listing}), reducing the listing problem to $\tO(n^{1/2+\epsilon/2})$ plain min-plus convolutions of restrictions of $A_1$ and $B_1$. These are sequences over the universe $[\sqrt n]$, and can thus be convolved in time $O(n^{3/2-\epsilon})$ by our initial assumption. The total time is $O(n^{2-\epsilon/2})$.

Of course, we first need to transform the initial instance $(A, B)$ into a pseudo-random one, and this is where the rank becomes relevant. To simplify in this overview, suppose that \emph{all} indices $k$ have many (i.e., more than $n^{1/2+\epsilon/2}$) pseudo-witnesses.\footnote{In fact, this idealized situation is impossible. For instance, $k = 0$ always has exactly one (pseudo\=/)witness, namely $i = 0$. But this is a technicality that shall not concern us in this overview.} In this case, so we claim, the whole instance has low rank and can directly be solved in subquadratic time. To see this, we first apply the standard reduction from Min-Plus Convolution to 3SUM: We approximate $C \approx A * B$ (by some standard trick, with coordinate-wise constant additive error), then replace the sequences $A, B, C$ by their \emph{indexed sets}
\begin{align*}
    \mathcal A &= \set{(i, A[i]) : 0 \leq i \leq n}, \\
    \mathcal B &= \set{(j, B[j]) : 0 \leq j \leq n}, \\
    \mathcal C &= \set{(k, C[k]) : 0 \leq k \leq 2n},
\end{align*}
and solve the remaining 3SUM instance $(\mathcal A, \mathcal B, \mathcal C)$. These are 2-dimensional sets that can be flattened to plain integers via the usual embedding. We argue that $\mathcal C$ has low rank, $r(\mathcal C) \leq \tilde O(n^{1-\epsilon/4})$. This can be shown in two steps:
\begin{itemize}
    \item First, consider the set
    \begin{equation*}
        \mathcal C_1 = \set{(k, (A_1 * B_1)[k] \cdot \sqrt{n}) : 0 \leq k \leq 2n}.
    \end{equation*}
    Assuming that each $k$ has many pseudo-witnesses, we can sample a set of $\tilde O(n^{3/4-\epsilon/4})$ indices $X \subseteq [0, n]$. With good probability, for each $k$ we include some witness $i \in X$ and simultaneously also $k - i \in X$. In this case, the set
    \begin{equation*}
        \mathcal X_1 = \set{(i, A_1[i] \cdot \sqrt{n}) : i \in X} \cup \set{(j, B_1[j] \cdot \sqrt{n}) : j \in X}
    \end{equation*}
    satisfies $\mathcal C_1 \subseteq \mathcal X_1 + \mathcal X_1$, witnessing that $r(\mathcal C_1) \leq |\mathcal X_1| \leq O(n^{3/4-\epsilon/4})$. 
    \item Second, note that $\mathcal C_1$ is an indexed set that approximates $\mathcal C$ well. Indeed, for any $c \in \mathcal C$, there exists $c_1 \in \mathcal C_1$ such that $c - c_1 \in \mathcal I$, where $\mathcal I = \set{(0, 0), \dots, (0, \sqrt{n})}$. Hence, we can expand $\mathcal X_1$ to $\mathcal{X}$ so that $\mathcal X_1 + \mathcal X_1 + \mathcal I \subseteq \mathcal{X} + \mathcal{X}$, which will then imply $\mathcal C \subseteq \mathcal X + \mathcal X$. To achieve this, we set \[\mathcal X = \mathcal X_1 \cup \left\{(0, 0), (0, 1), \ldots, (0, n^{1/4})\right\} \cup \left\{(0, n^{1/4}), (0, 2\cdot n^{1/4}), \ldots, (0, \sqrt{n})\right\}.\] Finally, $r(\mathcal C) \leq |\mathcal X| \leq O(n^{1-\epsilon/4})$. 
\end{itemize}
In summary, our low-rank 3SUM algorithm (\cref{thm:3sum-low-rank}) can solve the 3SUM instance $(\mathcal A, \mathcal B, \mathcal C)$ in subquadratic time $O(n^{2-\epsilon'})$.
\section{Preliminaries} \label{sec:prelims}

We use $[n]$ to denote $\{1,2,\dots,n\}$ for any integer $n \ge 0$. 
\subsection{Additive Combinatorics}
For two sets $A, B \subseteq \Int$, define their \emph{sumset} $A+B = \{a+b : a\in A, b\in B\}$ and \emph{difference set} $A-B = \{a-b : a\in A, b\in B\}$. Let $-A= \{-a: a\in A\}$.
The $k$-fold sumset $kA$ is recursively defined by $kA=(k-1)A + A$ with $1A=A$. Denote $r_{A+B}(z) = |\{(x,y)\in A\times B: x+y=z\}|$. We often conveniently refer to $r_{A+B}(z)$ as the \emph{multiplicity} of $z$ in the sumset $A+B$. Similarly, the \emph{multiplicity} of $z$ in the $k$-fold sumset $kA$ means $|\{(x_1,\dots,x_k)\in A^k: x_1+\dots + x_k = z\}|$.

\begin{lemma}[Pl\"unnecke--Ruzsa Inequality] \label{lem:plunnecke-ruzsa}
    For two sets $A, B$, if $|A + B| \le K |A|$, then for any integers $n, m \ge 0$, $|mB - nB| \le K^{m+n} |A|$. 
\end{lemma}

\begin{lemma}[Sparse Convolution, \cite{BringmannFN22}] \label{lem:sparse-conv}
There is a deterministic algorithm that, given $A, B \subseteq [-u, u]$, computes the sumset $A + B$ along with the multiplicities $r_{A+B}$ in time $|A + B| \cdot (\log u)^{O(1)}$.
\end{lemma}

The \emph{additive energy} between two sets $A,B\subseteq \Z$ is
\begin{equation*}
    E(A,B) = \abs{\set{(a_1,a_2,b_1,b_2)\in A^2 \times B^2: a_1+b_1=a_2+b_2}} = \sum_{z}r_{A+B}(z)^2.
\end{equation*}
The additive energy of a set $A\subseteq \Z$ is $E(A)=E(A,A)$. We also define the following standard higher-order analogue of the additive energy, for any integer $k \geq 1$:
\begin{equation*}
    T_k(A) = \abs{\set{(a_1, \dots, a_k, b_1, \dots, b_k) \in A^{2k} : a_1 + \dots + a_k = b_1 + \dots + b_k}}.
\end{equation*}

The Balog--Szemer\'edi--Gowers (BSG) Theorem was originally proved in \cite{BalogS94,Gowers01}.
We use the following algorithmic version of the BSG theorem appearing in \cite[Theorem 3.3]{AbboudBF23} (which is an adaptation of \cite{ChanL15}).
\begin{theorem}[BSG Theorem]
\label{thm:bsg-symm}
\label{thm:bsg}
If set $A$ has additive energy $E(A)\ge |A|^3/K$, then there exists a subset $A'\subseteq A$ such that
   \begin{itemize}
       \item $|A'|\ge \Omega(K^{-2}|A|)$, and
       \item $|A'+A'|\le O(K^{24}|A|)$.
   \end{itemize}
   Moreover, given $A$ and $K$, we can compute $A'$ in time $\tO(K^{12}|A|)$ by a randomized algorithm.
\end{theorem}

The Asymmetric BSG Theorem was proved in the textbook by Tao and Vu \cite[Section 2.6]{TaoV06}.  We need an algorithmic version of this theorem stated below, which we prove in \cref{sec:asym-bsg}.
\begin{restatable}[Asymmetric BSG Theorem]{theorem}{thmasymbsg} \label{thm:asym-bsg}
Let $\epsilon > 0$ and $0 < \alpha \leq 1$. Let $A, B \subseteq \Int$ be sets with energy $E(A, B) \geq \alpha |A|^2 |B|$. Then there are sets $H, \Lambda \subseteq \Int$ and some $K \leq O_\epsilon(\alpha^{-O_\epsilon(1)} (1 + |B| / |A|)^\epsilon)$ such that:
\begin{itemize}
    \item $|H + H| \leq K |H|$,
    \item $|H| \, |\Lambda| \leq K |B|$,
    \item $|A \cap H| \geq |A| / K$ and $|B \cap (H + \Lambda)| \geq |B| / K$.
\end{itemize}
Moreover, there is a randomized algorithm to compute $H, \Lambda$ in time $\tO((|A| + |B|) K)$.
\end{restatable}

Related to these tools, we will also need to compute approximations of sumset multiplicities. 

\begin{lemma}[Approximating Sumset Multiplicities] \label{lem:approx-mult}
There is a randomized algorithm that, given sets $A, B, C$ and $\epsilon > 0$, computes approximations $\tilde r(x)$, for all $x \in C$, with additive error
\begin{equation*}
    |r_{A+B}(x) - \tilde r(x)| \leq \epsilon|A|.
\end{equation*}
The algorithm runs in time $\tO(|A| + |B| / \epsilon + |C|)$.
\end{lemma}

A proof of  \cref{lem:approx-mult} can be found in \cite[Footnote 6]{FischerJX25}.
See \cite[Corollary~3.1]{FischerJX25} for a derandomization of this algorithm. 

\subsection{Problem Definitions}
As is common in fine-grained complexity, hardness hypotheses and reductions implicitly allow Monte-Carlo randomization (though there has been some interest in derandomizing 3SUM-related reductions~\cite{ChanH20,FischerKP24}).

\begin{definition}[3SUM]
    Given three sets of integers $A, B, C$, determine whether there exists $a \in A, b \in B, c \in C$ such that $a + b = c$. In its All-Integer version we need to determine for every number in the input, whether it participates in a solution $a + b = c$. 
\end{definition}

We use $\ThreeSUM(n)$ to denote the running time of 3SUM instances with size-$n$ inputs. We also use $\ThreeSUM(n \mid r \le R)$ to denote the running time of 3SUM instances where the rank (\cref{defn:low-rank}) of one of the sets is at most $R$, and such a rank decomposition is given. A subinstance of a 3SUM instance with input $(A, B, C)$ is another 3SUM instance with input $(A', B', C')$ where $A' \subseteq A, B' \subseteq B, C'\subseteq C$.

\begin{definition}[Min-Plus Convolution]
    Given two sequences $A, B \in \Z^{[0, n]}$, compute a sequence $C = A * B \in \Z^{[0, 2n]}$, where $C[k] = \min_{i + j = k} \{A[i] + B[j]\}$. 
\end{definition}

We use $\MinPlusConv(n)$ to denote the running time of Min-Plus Convolution instances with size-$n$ inputs. We also use $\MinPlusConv(n \mid u \le U)$ to denote the running time of Min-Plus Convolution instances where the entries are from $[U]$.

\begin{definition}[Min-Plus Convolution Listing]
\label{def:mpconvlist}
    Given two sequences $A, B \in \Z^{[0, n]}$, and an integer threshold $t$, list $t$ witnesses $(i, j)$ where $A[i] + B[j] = (A * B)[k]$ for every $k$. If there are fewer than $t$ witnesses for certain indices $k$, list all the witnesses for these indices. 
\end{definition}

Similar as before, we use $\MinPlusConvListing(n, t)$ and $\MinPlusConvListing(n, t \mid u \le U)$ to denote the running times. A subinstance of Min-Plus Convolution (and all sequence problems considered in this paper) can be obtained by setting some entries to $\perp$, a value that cannot be involved in any solution. 

\begin{definition}[Bounded-Monotone Min-Plus Convolution]
    Given two sequences $A, B \in \Z^{[0, n]}$ where $B$ is monotonically nondecreasing and has entries in $[O(n)]$, compute $A * B$.
\end{definition}

\begin{definition}[Min-Max Convolution]
    Given two sequences $A, B \in \Z^{[0, n]}$, compute a sequence $C = \MinMax(A, B) \in \Z^{[0, 2n]}$, where $C[k] = \min_{i + j = k} \max\{A[i], B[j]\}$. 
\end{definition}

\begin{definition}[Min-Equality Convolution]
    Given two sequences $A, B \in \Z^{[0, n]}$, compute a sequence $C = \MinEqual(A, B) \in \Z^{[0, 2n]}$, where $C[k] = \min_{i + j = k, A[i] = B[j]} A[i]$. 
\end{definition}

\begin{definition}[Min-Witness Convolution]
    Given two Boolean sequences $A, B \in \{0, 1\}^{[0, n]}$, compute a sequence $C = \MinWitness(A, B) \in (\Z\cup \{\bot\})^{[0, 2n]}$, where $C[k] = \min\{i \in [0, n]: k - i \in [0, n] \wedge A[i] = B[k - i] = 1\}$ (if no such $i$ exists, define $C[k]=\bot$).  
\end{definition}

We now define the index set of a sequence. We will frequently use this notion to convert Min-Plus Convolution instances to 3SUM instances.
\begin{definition}[Index set $\mathcal{I}(\cdot)$]
\label{defn:indexset}
   For a sequence $A\in \Z^{[0,n]}$ indexed by integers in $[0,n]$, the \emph{index set} of $A$ is $\mathcal{I}(A) = \{(i,A[i]): i\in [0,n]\}$. 

   More generally,  for a sequence $A\in (\Z \cup \{\bot \})^{[0,n]}$, the index set of $A$ is $\mathcal{I}(A) = \{(i,A[i]): i\in [0,n], A[i]\neq \bot\}$. 
\end{definition}

\subsection{Reduction from Listing to Non-Listing Problems}

For most problems considered in this paper, one can use an algorithm for the non-listing version of a problem to list witnesses,  via a standard approach based on sampling and parallel binary search (e.g., \cite{AlonGMN92,Seidel95}). This is summarized in the following lemma. 

\begin{lemma}[Listing Witnesses for Convolution Problems] 
\label{lem:min-plus-conv-listing}
Any Min-Plus Convolution Listing  instance $(A, B)$ with threshold $t$ can be reduced to $t (\log n)^{O(1)}$ Min-Plus Convolution (non-Listing) subinstances. The reduction runs in time $\tO(n t)$. In particular,
\begin{equation*}
    \MinPlusConvListing(n, t \mid u) \leq \tO(t \cdot \MinPlusConv(n \mid u)).
\end{equation*}
Moreover, for any $t_1, t_2 \geq 1$,
\begin{equation*}
    \MinPlusConvListing(n, t_1 t_2 \mid u) \leq \tO(t_1 \cdot \MinPlusConvListing(n, t_2 \mid u)).
\end{equation*}
The same statements also hold with ``Min-Plus'' replaced by ``Min-Max'' and ``Min-Equality''.
\end{lemma}
We omit the proof of \cref{lem:min-plus-conv-listing} here since it is standard. We refer readers to \cite[Section~5.1]{Fischer26} for a proof in the analogous setting of Min-Plus Matrix Product, which can be adapted to Min-Plus Convolution in a straightforward manner. The same proof strategy also applies to Min-Max Convolution and Min-Equality Convolution.\footnote{
En route we rely on a sampling-based approach, which requires setting some input values to $\perp$, so that they cannot be used in any solution (see \cite[Section 5.1]{Fischer26}). In the problems we consider, $\perp$ can be replaced with actual numbers without heavily affecting the problem sizes. In Min-Plus Convolution with universe $[u]$, we can simply set~$\perp$ to $2u + 1$, so they can never participate in a solution with value less than or equal to $2u$, while all the solutions obtained from non-$\perp$ values have value at most $2u$.  
Similarly, in Min-Max Convolution with universe $[u]$, we can simply set $\perp$ to $u+1$.
In Min-Equality Convolution, we can set $\perp$ in $A$ and in $B$ to be two different special values that do not occur in the non-$\perp$ values of $A,B$.
}

The same sampling strategy together with FFT yields the following more well-known lemma for listing witnesses of Boolean Convolution:
\begin{lemma}
\label{lem:listingboolean}
   Given two sequences $A,B\in \{0,1\}^{[0,n]}$ and a threshold $t\ge 1$, one can compute their Boolean Convolution $C\in \{0,1\}^{[0,2n]}$ defined as $C[k]=\bigvee_{i+j=k}(A[i]\land B[j])$, and additionally, for each $k$, list up to $t$ witnesses $i$ such that $A[i]=B[k-i]=1$, in $\tO(tn)$ time.
\end{lemma}

\section{Popular \texorpdfstring{\boldmath$k$}{k}-Sum Structure} \label{sec:pop-sums}
The goal of this section is to prove \cref{thm:k-sum-pop}. We start with some basic definitions and facts.

\begin{definition}[Popular $k$-Sums]
Let $A_1, \dots, A_k \subseteq \Int$, and $p>0$. The set $\Pop(A_1, \dots, A_k; p)$ consists of all integers~$x$ that can be expressed in at least $p$ distinct ways as
\begin{equation*}
    x = a_1 + \dots + a_k
\end{equation*}
for $(a_1, \dots, a_k) \in A_1 \times \dots \times A_k$; we call these the \emph{$p$-popular $k$-sums}. When $A = A_1 = \dots = A_k$, we simply write $\Pop^k(A; p)$.
\end{definition}

The following simple observation will be useful later.

\begin{observation} \label{obs:pop-remove}
If $\Pop(A_1, \dots, A_k; p) \neq \emptyset$ then $\Pop(A_2, \dots, A_k; p / |A_1|) \neq \emptyset$.
\end{observation}

\thmksumpop*

The following lemma is based on simple random sampling. Its version for $k = 2$ has been considered before (e.g., \cite{JinX23}). 

\begin{lemma}[Approximating Popular Sums] \label{lem:approx-pop-sums}
There is a randomized algorithm that, given sets $A_1, \dots, A_k \subseteq \Int$ and an integer threshold $p \geq 1$, computes a set $P$ such that
\begin{equation*}
    \Pop(A_1, \dots, A_k; 2p) \subseteq P \subseteq \Pop(A_1, \dots, A_k; p)
\end{equation*}
in randomized time
\begin{equation*}
    \tO\parens*{\frac{|A_1| \, \dots \, |A_k|}{p}}.
\end{equation*}
\end{lemma}
\begin{proof}
Denote $A = A_1\times \dots \times A_k$.
For a large constant $C$, we uniformly independently sample $m = \lceil \frac{C|A|\log |A|}{p}\rceil$  $k$-tuples $t_1,t_2,\dots,t_m\in A$. Let $s_i$ denote the sum of the $k$ components in the tuple $t_i$.
    Let $\tilde f_x$ denote the number of times the integer $x$ occurs in the list of sums $(s_1,\dots,s_m)$.
    Let $P$ be the set of all integers $x$ with  $\tilde f_x \ge 1.5p\cdot \frac{m}{|A|}$. We compute and return $P$ in $\tO(m) = \tO(\frac{|A_1|\, \dots \, |A_k|}{p})$ time. 
    
    We now show that $P$ satisfies the claimed property with high probability. For every integer $x$, if $x$ can be expressed in $f_x$ ways as $x=a_1+\dots+a_k$ for $(a_1,\dots,a_k)\in A$, then $\Pr[s_i=x]=f_x/|A|$ for every $1\le i\le m$. 
    Hence, if $f_x\ge 2p$, then by Chernoff bound, $x\notin P$ happens with probability $\Pr[\tilde f_x<1.5p\cdot \frac{m}{|A|}]\le \exp(-\frac{pm}{16|A|})\le |A|^{-C/16}$.
    If $f_x< p$, then by Chernoff bound,  $x\in P$ happens with probability
    $\Pr[\tilde f_x\ge 1.5p\cdot \frac{m}{|A|}]\le \exp(-\frac{pm}{12|A|})\le |A|^{-C/12}$. By a union bound over all $x\in A_1+ \dots + A_k$, we conclude that $P$ is correct with high probability in $|A|$.
\end{proof}

\subsection{Qualitative Optimality}
Let us briefly comment that \cref{thm:k-sum-pop} is (qualitatively) optimal. For every $c > 0$ and every integer $k \ge 2$, there is a small enough $\eps \in (0, 1)$ that would violate the theorem statement. More specifically, there is a family of size-$n$ sets $A$ such that $\Pop^k(A; p) \neq \emptyset$ for some $p = \Omega(|A|^{k-1-c})$, and so that any set $H$ with $|A \cap H| \ge n^{1-\eps}$ and $|H| \le n^{1+c+\eps}$ would have doubling constant much larger than $n^\eps$. The same proof would also show for every $\eps \in (0, 1)$ and every $k \ge 2$, there exists $c > 0$ that violates the theorem statement. 
This shows that $k$ in the theorem necessarily depends on $c$ and $\epsilon$ (though the dependence could possibly be improved to linear in $c/\epsilon$). 

Indeed, take $A = I + R$ where $I$ is an interval (or any other low-doubling set) of size $n^{1-c/(k-1)}$, and where $R$ is a random set of size $n^{c/(k-1)}$ that is not correlated with $H$. By taking $R$ from a large enough universe, we can assume that, for every $\ell \le (c+2) \log n$, there is no $\ell$-term additive relation 
\begin{equation*}
    (i_1 + r_1) + \cdots + (i_\ell + r_\ell) = (i_{\ell+1} + r_{\ell+1}) + \cdots + (i_{2\ell} + r_{2\ell}),
\end{equation*}
where $r_1, \ldots, r_{2\ell}$ are distinct, while $i_1, \ldots, i_{2\ell}$ might not. 
Then $A$ has size roughly $|I| \, |R| = n$. Moreover, $\Pop^k(A; |A|^{k-1-c})$ is non-empty as each sum $r_1 + \cdots + r_k \in k R$ can be extended in about $|I|^{k-1} = n^{k-1-c}$ ways to a common popular sum in $k I + kR$. 
Take $H$ to be any set where $|H| \le n^{1+c+\eps}$ and $|A \cap H| \ge n^{1-\eps}$. By pigeonhole principle, $A \cap H$ must contain $p_t = i_t + r_t$ for $i_t \in I, r_t \in R$ where the values $r_t$ are distinct, for some $t \ge n^{1-\eps}/n^{1-c/(k-1)} = n^{c/(k-1)-\eps}$. By our assumption on $R$, there is no $\ell$-term additive relation for numbers $P = \{p_1, \ldots, p_t\}$ for every $\ell \le (c+2) \log n$.  

Now consider the additive energy $E(P, H)$. For every nontrivial solution, $p_1 + h_1 = p_2 + h_2$ for $p_1, p_2 \in P$ and $h_1, h_2 \in H$, we consider adding an edge between $h_1$ and $h_2$. Note  that an $\ell$-cycle in this graph induces an $\ell$-term additive relation for numbers $P$. Hence, the girth of this graph is $>(c+2) \log n$, which implies that the number of edges is $O(|H|)$ by the graph Moore bound \cite{AlonHL02}. Hence, $E(P, H)$ is dominated by the trivial solutions, so $E(P, H) = O(|P||H|)$. By Cauchy--Schwarz inequality, 
\[
|P+H| = \sum_{z} [r_{P+H}(z) > 0] \ge \frac{(\sum_z r_{P+H}(z))^2}{\sum_z r_{P+H}(z)^2} \ge \frac{|P|^2|H|^2}{O(|P||H|)} \ge \Omega(|P||H|). 
\]
As $|H + H| \ge |P+H|$,  the doubling constant of $H$ is at least $\Omega(|P|) = \Omega(n^{c/(k-1)})$, which is greater than $n^\eps$ when $k \ll \frac{c}{\eps}$. 

\subsection{Existential Proof}
It is more convenient to re-express having popular sums in terms of the higher-order energy $T_k(A)$. 
\cite{Shkredov14}.
The proof relies on the following lemma:

\begin{lemma} \label{lem:higher-energy}
For every $\epsilon, c > 0$, there is some $k = k(c, \epsilon)$ such that the following holds. Let $A \subseteq \Int$ be such that $T_k(A) \geq \Omega(|A|^{2k-1-c})$. Then there exists a set $H \subseteq \Int$ and some $K = O_{c, \epsilon}(|A|^\epsilon)$ such that:
\begin{itemize}
    \item $|H| \leq K |A|^{1+c}$,
    \item $|H + H| \leq K |H|$,
    \item $|A \cap H| \geq |A| / K$.
\end{itemize}
\end{lemma}

Let us postpone the proof of \cref{lem:higher-energy} for the moment, and first see how \cref{thm:k-sum-pop} can be concluded from the lemma.

\begin{proof}[Proof of \cref{thm:k-sum-pop}]
Let $k = k(\epsilon, c)$ be chosen as in \cref{lem:higher-energy}. Let $A$ be a set with a $p$\=/popular $2k$-sum $z$ for some $p = \Omega(|A|^{2k-1-c})$. Let $A^* = \pm (A \cup \set{z})$, then $T_{k}(A^*) \geq p$ (as each solution $a_1 + \dots + a_{2k} = z$ in $A$ can be mapped to a unique solution $a_1 + \dots + a_k = b_1 + \dots + b_k$ in~$A^*$). Apply \cref{lem:higher-energy} to $A^*$. This yields a set $H^*$ satisfying the first two properties for some $K = O_{c, \epsilon}(|A|^\epsilon)$ and that $|A^* \cap H^*| \geq |A^*| / K$. As $A^*$ can be seen of the union of four copies of $A$, possibly translated or negated, one such copy is covered by $H^*$ in at least $|A| / 4K$ elements. Thus, by possibly translating and negating $H^*$ (leaving the first two properties unaffected), one obtains the claimed set $H$.
\end{proof}

\begin{proof}[Proof of \cref{lem:higher-energy}]
We determine $k = k(c, \epsilon)$ later. If $T_{k-1}(A) \geq |A|^{2(k-1)-1-c+\delta}$ for $\delta := 2c/k$, then we instead apply this lemma inductively for $k-1$. This recursive process continues for less than $k$ steps before reaching a contradiction, and thus never reaches a base case. Instead, at some inductive call we must be in the opposite case that $T_{k-1}(A) \leq |A|^{2(k-1)-1-c+\delta}$.

For $D$ a power of $2$, let $P_D \subseteq (k - 1) A$ be the subset of $(k-1)$-sums with multiplicity in the range $[D, 2D)$, i.e.,
\begin{equation*}
    P_D = \set{ x \in \Int : D \leq \abs{\set{(a_1, \dots, a_k) \in A^k : a_1 + \dots + a_k = z}} < 2D}.
\end{equation*}
Then:
\begin{align*}
    T_k(A)
    &= \sum_{z \in \Int} \abs{\set{(a_1, \dots, a_k) \in A^k : a_1 + \dots + a_k = z}}^2 \\
    &\leq \sum_{z \in \Int} \parens*{\sum_D 2D \cdot \abs{\set{(a, x) \in A \times P_D : a + x = z}}}^2 \\
    &= \sum_{D, D'} \sum_{z \in \Int} 2D \cdot \abs{\set{(a, x) \in A \times P_D : a + x = z}} \cdot 2D' \cdot \abs{\set{(a', x') \in A \times P_{D'} : a' + x' = z}} \vphantom{\Bigg(^2}\\
    &\leq \sum_{D, D'} \parens*{\sum_{z \in \Int} 4D^2 \cdot \abs{\set{(a, x) \in A \times P_D : a + x = z}}^2}^{1/2} \\
    &\qquad\qquad\cdot \parens*{\sum_{z \in \Int} 4(D')^2 \cdot \abs{\set{(a', x') \in A \times P_{D'} : a' + x' = z}}^2}^{1/2} \\
    &= \sum_{D, D'} \parens*{4D^2 \cdot E(A, P_D)}^{1/2} \parens*{4(D')^2 \cdot E(A, P_{D'})}^{1/2}.
\end{align*}
We apply the dyadic pigeonhole principle: let $D$ denote the power of 2 maximizing $D^2 \cdot E(A, P_D)$. Then:
\begin{equation*}
    T_k(A) \leq O((k \log|A|)^2) \cdot D^2 \cdot E(A, P_D).
\end{equation*}
Recall that $T_k(A) \geq \Omega(|A|^{2k-1-c})$ by our initial assumption, and thus
\begin{equation} \label{lem:higher-energy:eq:lower}
    D^2 \cdot E(A, P_D) \geq |A|^{2k-1-c-o(1)}.
\end{equation}
Combined with the trivial bounds $E(A, P_D) \leq |A|^2 |P_D|$ and $|P_D| \leq |A|^{k-1} / D$, it further follows that $D \geq |A|^{k-2-c-o(1)}$ and thus $|P_D| \leq |A|^{1+c+o(1)}$. Moreover, we have that
\begin{equation} \label{lem:higher-energy:eq:upper}
    |P_D| \cdot D^2 \leq T_{k-1}(A) \leq |A|^{2(k-1)-1-c+\delta}.
\end{equation}
Combining~\eqref{lem:higher-energy:eq:lower} and~\eqref{lem:higher-energy:eq:upper}, it follows that
\begin{equation*}
    E(A, P_D) \geq \frac{|A|^2 |P_D|}{|A|^{\delta+o(1)}}. 
\end{equation*}

In this situation the asymmetric BSG theorem applies. Specifically, we apply \cref{thm:asym-bsg} to the sets $A$ and~$P_D$ with parameters $\epsilon' = \epsilon/(4c)$ and~$\alpha = |A|^{-\delta-o(1)}$. It yields sets $H, \Lambda \subseteq \Int$ such that
\begin{equation} \label{lem:higher-energy:eq:H-doubling}
    |H + H| \leq K |H|,
\end{equation}
and
\begin{equation} \label{lem:higher-energy:eq:H-size}
    |H| \leq |H| \, |\Lambda| \leq K |P_D| \leq (K |A|^{o(1)}) \cdot |A|^{1+c},
\end{equation}
and
\begin{equation} \label{lem:higher-energy:eq:H-cover}
    |A \cap H| \geq \frac{|A|}{K},
\end{equation}
for some
\begin{equation*}
    K \leq O_{\epsilon,c}\parens*{\alpha^{-O_{\epsilon}(1)} \parens*{1 + \frac{|P_D|}{|A|}}^{\epsilon/(4c)}}.
\end{equation*}
From the previous bounds on $K$ and~$|P_D| / |A|$, we can bound
\begin{equation*}
    K \leq O_{\epsilon,c}(|A|^{O_{\epsilon,c}(\delta)+\epsilon/4 + o(1)}).
\end{equation*}
Recall that $\delta = 2c / k$. Choose $k = k(c, \epsilon)$ sufficiently large so that this bound becomes $K \leq O_{\epsilon,c}(|A|^{\epsilon/2})$. Then the right-hand side of~\eqref{lem:higher-energy:eq:H-size} can be bounded by $O_{\epsilon,c}(|A|^{\epsilon}) \cdot |A|^{1+c}$. Hence, the set~$H$ is as claimed by \cref{lem:higher-energy:eq:H-size,lem:higher-energy:eq:H-doubling,lem:higher-energy:eq:H-cover}.
\end{proof}

\subsection{Efficient Algorithm}
We finally revisit the previous proof to demonstrate that all steps can be implemented efficiently, leading to the following algorithmic theorem:

\begin{theorem} \label{thm:k-sum-pop-algo}
For every $c > 0$ and $\epsilon > 0$, there is some $k = k(c, \epsilon)$ such that the following holds. There is a randomized algorithm that, given $A \subseteq \Int$ and some parameter $p = \Omega(|A|^{k-1-c})$, either certifies that $\Pop^k(A; p) = \emptyset$, or otherwise computes a set $H \subseteq \Int$ such that the following properties hold for some $K \leq O(|A|^\epsilon)$:
\begin{itemize}
    \item $|H| \leq K |A|^{1+c}$,
    \item $|H + H| \leq K |H|$,
    \item $|A \cap H| \geq |A| / K$.
\end{itemize}
The algorithm runs in time $\tO(|A|^{1+c+\epsilon})$.
\end{theorem}

The following lemma is analogous to \cref{lem:higher-energy}.

\begin{lemma} \label{lem:higher-energy-algo}
For every $\epsilon, c > 0$, there is some $k = k(c, \epsilon)$ such that the following holds. There is a randomized algorithm that, given $A \subseteq \Int$ such that $T_k(A) \geq \Omega(|A|^{2k-1-c})$ computes a set $H \subseteq \Int$ such that the following properties hold for some $K \leq O(|A|^\epsilon)$:
\begin{itemize}
    \item $|H| \leq K |A|^{1+c}$,
    \item $|H + H| \leq K |H|$,
    \item $|A \cap H| \geq |A| / K$.
\end{itemize}
The algorithm runs in time $\tO(|A|^{1+c+\epsilon})$.
\end{lemma}

\begin{proof}[Proof of \cref{thm:k-sum-pop-algo}]
We choose $k = k(c, \epsilon)$ as in \cref{lem:higher-energy-algo}. Given $A$, we first apply \cref{lem:approx-pop-sums} to compute a set $P$ satisfying (with high probability) that
\begin{equation*}
    \Pop^k(A; p) \subseteq P \subseteq \Pop^k(A; p/2).
\end{equation*}
If $P = \emptyset$, then we have certified that $\Pop^k(A; p) = \emptyset$. So suppose otherwise that $P$ contains some element $z$. We are guaranteed that $z$ is a $p/2$-popular $k$-sum in $A$. We call \cref{lem:higher-energy-algo} on the set $A^* = \pm (A \cup \set{z})$ with parameter $p/2$. As before, it holds that $T_k(A^*) \geq p/2$ and thus \cref{lem:higher-energy-algo} computes a set $H \subseteq \Int$ which is as claimed, except that instead of the third condition we have $|A^* \cap H| \geq |A| / K$. However, as $A^*$ is a union of four copies of $A$, possibly translated and negated, one of the corresponding translations and/or negations of $H$ must satisfy the claim. We test all four and verify each in linear time. The total time is dominated by the calls to \cref{lem:approx-pop-sums,lem:higher-energy-algo}, running in time $\tilde O(|A|^k / p) = \tilde O(|A|^{1+c})$ and $O(|A|^{1+c+\epsilon})$, respectively.
\end{proof}

\begin{proof}[Proof of \cref{lem:higher-energy-algo}]
Following the proof of \cref{lem:higher-energy}, we design a recursive algorithm. If $T_k(A)$ is sufficiently large as in the lemma statement, then the algorithm returns $H$ satisfying the claimed properties with high probability. If the bound on $T_k(A)$ does not hold, then the algorithm may simply fail.

In the first case of the proof of \cref{lem:higher-energy}, the set $A$ satisfies that $T_{k-1}(A) \geq |A|^{2(k-1)-1-c+\delta}$, for some constant $\delta$. Without knowing if this previous condition is actually satisfied, we regardless call the algorithm recursively for $k-1$. If it reports $H$ as claimed, then we are done. Otherwise, we can infer that we instead in the second case.

In that second case, we construct sets $Q_D$ approximating the sets $P_D$ from the proof of \cref{lem:higher-energy} as follows. Here, $D$ ranges of all powers of $2$ at least $|A|^{k-1-c-o(1)}$. Each set $Q_D$ is required to satisfy that (i) each $x \in Q_D$ has multiplicity at least $D/2$ and less than $4D$ as a $k$-sum in $A$ and (ii) each $k$-sum in $A$ with multiplicity at least $D$ and less than $2D$ is contained in $Q$. Each such set~$Q_D$ can be computed in time $\tilde O(|A|^{1+c})$ by a constant number of calls to \cref{lem:approx-pop-sums}. Then, by the same proof as in \cref{lem:higher-energy} we have that
\begin{equation*}
    T_k(A) \leq O((k \log|A|)^2) \cdot \max_D(D^2 \cdot E(A, Q_D)),
\end{equation*}
and thus
\begin{equation*}
    E(A, Q_D) \geq \frac{|A|^2 |Q_D|}{|A|^{\delta+o(1)}}
\end{equation*}
for some choice of $D$. In this case we apply the asymmetric BSG theorem on $A$ and $Q_D$. In the algorithm, we do not know the optimal value of $D$, so we simply try all $O(\log |A|)$ options. The algorithmic version of the asymmetric BSG theorem (\cref{thm:asym-bsg}) then returns the claimed set $H$ in time $\tilde O((|A| + |Q_D|) K) = O(|A|^{1+c+\epsilon+o(1)})$. (The $o(1)$ overhead can be removed by reducing $\epsilon$ further.)
\end{proof}
\section{Subquadratic-Time Algorithm for Low-Rank 3SUM} \label{sec:3sum-low-rank}
In this section we design a subquadratic-time algorithm for low-rank 3SUM.

\begin{theorem}[Low-Rank 3SUM] \label{thm:3sum-low-rank}
For every $0<\rho \le 1/2$ there is some $\delta = \delta(\rho) > 0$ such that the following holds. There is a randomized algorithm that solves a given All-Integer 3SUM instance~$(A, B, C)$, when additionally provided with a set $X \subseteq \Int$ of size $|X| \leq n^{1-\rho}$ satisfying $A \subseteq X + X$, and runs in time
\begin{equation*}
    \ThreeSUM(n \mid r \leq n^{1-\rho}) = O(n^{2-\delta}).
\end{equation*}
\end{theorem}

We will prove \cref{thm:3sum-low-rank} by a sequence of two reduction steps to reduce from any given low-rank instance to few instances exhibiting even more structure. That structure is that the sets~$A, B, C$ are all contained in some larger set $H$ with small doubling. For ease of notation we will typically refer to $H$ as a \emph{low-doubling container}, formally defined as follows.

\begin{definition}[Low-Doubling Container]
We say that $H$ is a \emph{$(K, N)$-container} of some set $A$ if
\begin{itemize}
    \item $A \subseteq H$,
    \item $|H + H| \leq K |H|$, and
    \item $|H| \leq N$.
\end{itemize}
\end{definition}

Our reduction can then be phrased as follows:

\begin{lemma}[From Low Rank to Low-Doubling Containment] \label{lem:3sum-low-rank-to-low-doubling-contained}
For every $ \epsilon > 0$ and $0<\rho\le 1/2$ there is some $\delta = \delta(\rho, \epsilon) > 0$ such that the following holds. There is a randomized algorithm that, given a 3SUM instance $(A, B, C)$ and a set $X$ of size $|X|\le n^{1-\rho}$ such that $C \subseteq X + X$, computes
\begin{itemize}
    \item $O(n^\epsilon)$ subinstances $(A' \subseteq A, B' \subseteq B, C' \subseteq C)$, each provided with a $(O(n^\epsilon), O(n^{2-2\rho}))$\=/container of $A' \cup B' \cup C'$, and
    \item a set of 3SUM solutions $\mathcal T \subseteq A \times B \times C$,
\end{itemize}
so that each solution is either part of exactly one subinstance or part of $\mathcal T$. The algorithm runs in time $O(n^{2-\delta})$.
\end{lemma}

We devote the following two subsections to the proof of \cref{lem:3sum-low-rank-to-low-doubling-contained}. Once established, the proof of \cref{thm:3sum-low-rank} is almost immediate by the following simple lemma.

\begin{lemma}[3SUM for Small Sumsets] \label{lem:3sum-small-sumset}
Any All-Integer 3SUM instance $(A, B, C)$ with universe $[u]$ can be solved in deterministic time $(|A + B| + |B - C| + |A - C|) (\log u)^{O(1)}$.
\end{lemma}
\begin{proof}
To detect all elements in $C$ that are part of a solution, we compute the sumset $A + B$ by a sparse convolution in time $|A + B| (\log u)^{O(1)}$ by \cref{lem:sparse-conv}, and report all elements in $C \cap (A + B)$. To detect all elements in $A$ and $B$ that are part of a solution, we proceed similarly with $A$ and $B$ in place of $C$.
\end{proof}

\begin{proof}[Proof of \cref{thm:3sum-low-rank}]
Let $\epsilon > 0$ be a parameter to be determined. We apply \cref{lem:3sum-low-rank-to-low-doubling-contained} to the given instance $(A, B, C)$ (by relabeling the input sets, we may assume $C$ has a low-doubling container as required by \cref{lem:3sum-low-rank-to-low-doubling-contained}), and obtain $O(n^\epsilon)$ subinstances $(A', B', C')$ and some set of 3SUM solutions~$\mathcal T$, with the promise that each solution in $(A, B, C)$ either appears in a subinstance or in~$\mathcal T$. We report all items present in $\mathcal T$. Then we solve all subinstances by \cref{lem:3sum-small-sumset}, and further report all items that appear in a solution in at least one subinstance. The correctness is clear.

It remains to analyze the running time. The call to \cref{lem:3sum-low-rank-to-low-doubling-contained} takes time $O(n^{2-\delta})$ for some constant $\delta = \delta(\rho, \epsilon) > 0$. We can also charge the scan through $\mathcal T$ to that running time. By \cref{lem:3sum-low-rank-to-low-doubling-contained}, each constructed subinstance $(A', B', C')$ admits an $(K =O(n^\epsilon), N = O(n^{2-2\rho}))$\=/container $H$ of~\makebox{$A' \cup B' \cup C'$}. It follows that
\begin{equation*}
    |A' + B'| \leq |H + H| \leq K |H| \leq O(n^{2-2\rho+\epsilon}).
\end{equation*}
Similarly, by an application of the Plünnecke--Ruzsa inequality (\cref{lem:plunnecke-ruzsa}),
\begin{equation*}
    |B' - C'| \leq |H - H| \leq K^2 |H| \leq O(n^{2-2\rho+2\epsilon}).
\end{equation*}
Therefore, the $O(n^\epsilon)$ calls to \cref{lem:3sum-small-sumset} run in total time $O(n^{2-2\rho+3\epsilon})$. The claimed running time follows, picking any constant $0 < \epsilon < 2\rho/3$.
\end{proof}

\subsection{Low-Doubling Containment for One Set} \label{sec:3sum-low-rank:sec:low-doubling-contained-one}
The proof of \cref{lem:3sum-low-rank-to-low-doubling-contained} involves two steps: (1) to reduce from any low-rank instance to subinstances where only one of the three sets admits a low-doubling container, and (2) to extend that one low-doubling container to all three sets. We start with step (1):

\begin{lemma}[Low-Doubling Containment for One Set] \label{lem:3sum-low-rank-to-low-doubling-contained-one}
For every $\epsilon > 0$ and $0<\rho\le 1/2$ there is some $\delta = \delta(\rho, \epsilon) > 0$ such that the following holds. There is a randomized algorithm that, given a 3SUM instance $(A, B, C)$ and a set $X$ of size $|X|\le n^{1-\rho}$ such that $C \subseteq X + X$, computes
\begin{itemize}
    \item $O(n^\epsilon)$ subinstances $(A' \subseteq A, B, C)$ with $O(n^\epsilon)$-doubling $O(n^{2-2\rho})$-containers of $A'$, and
    \item a set of 3SUM solutions $\mathcal T \subseteq A \times B \times C$,
\end{itemize}
so that each solution is either part of exactly one subinstance or part of $\mathcal T$. The algorithm runs in time $O(n^{2-\delta})$.
\end{lemma}

The proof of \cref{lem:3sum-low-rank-to-low-doubling-contained-one} requires some setup. We start with the following simple lemma that allows to trade off the size of the container with its doubling.

\begin{lemma} \label{lem:3sum-shrink-container}
Let $\epsilon, \gamma > 0$ and $0<\rho\le 1/2$. There is a deterministic algorithm to reduce any 3SUM instance $(A, B, C)$ provided with a $(O(n^\epsilon), O(n^{2-2\rho+\gamma}))$-container of $A$ to $O(n^{\gamma})$ subinstances $(A' \subseteq A, B, C)$ provided with a $(O(n^{\epsilon+\gamma}), O(n^{2-2\rho}))$-container of $A'$. It runs in linear time $O(n^{1+\gamma})$.
\end{lemma}
\begin{proof}
Let $H$ denote the $K$-doubling container of $A$ with $K=O(n^{\epsilon})$. Arbitrarily partition $H$ into $t = O(n^\gamma)$ pieces $H_1, \dots, H_t$ of size at most $\ceil{|H| / t} = O(n^{2-2\rho})$, and let $A_1, \dots, A_t$ denote the subsets $A_i = A \cap H_i$. Then the 3SUM instance $(A, B, C)$ is equivalent to the instances $(A_1, B, C), \dots, (A_t, B, C)$ (in the sense that their solutions stand in one-to-one correspondence). Moreover, each $A_i$ is contained in $H_i$, and as
\begin{equation*}
    |H_i + H_i| \leq |H + H| \leq K |H| \leq O(K t) |H_i| = O(n^{\epsilon+\gamma}) |H_i|,
\end{equation*}
$H_i$ is a $(O(n^{\epsilon+\gamma}), O(n^{2-2\rho}))$-container of $A_i$.
\end{proof}

The more interesting ingredient to \cref{lem:3sum-low-rank-to-low-doubling-contained-one} is the following lemma. Intuitively, it states that certain structure can be \emph{amplified}. Namely, any 3SUM instance $(A, B, C)$ where $C$ is a subset of $(k+1)$-popular sums can be replaced by a subinstance that is a subset of $(k+2)$-popular sums. Curiously, the amplified structure does not show up in $C$, but in $B$.

\begin{lemma} \label{lem:3sum-low-rank-pop-cover}
There is a randomized algorithm that, given a 3SUM instance $(A, B, C)$ and integer parameters $p, s \geq 1$ and sets $X, Y_1, \dots, Y_k \subseteq \Int$ satisfying $C \subseteq \Pop(X, Y_1, \dots, Y_k; p)$, computes
\begin{itemize}
    \item a subinstance $(A, B' \subseteq B, C)$ so that $B' \subseteq \Pop(X, Y_1, \dots, Y_k, -A; p')$ for $p' = p |A| / s$, and
    \item a set of 3SUM solutions $\mathcal T \subseteq A \times B \times C$,
\end{itemize}
so that each solution is either part of the subinstance or part of $\mathcal T$. The algorithm runs in time
\begin{equation*}
    \tO\parens*{\frac{s |X| \, |Y_1| \dots |Y_k|}{p} + |A| \, |X| + \frac{|B| \, |Y_1| \dots |Y_k|}{p} + \frac{|A| \, |B|}{s}}.
\end{equation*}
\end{lemma}

\begin{proof}
Let $p' = p |A| / s$. We run \cref{lem:approx-pop-sums} to compute a set $P$ satisfying
\begin{equation*}
    \Pop(X, Y_1, \dots, Y_k, -A; 2p') \subseteq P \subseteq \Pop(X, Y_1, \dots, Y_k, -A; p').
\end{equation*}
We call an item $b \in B$ \emph{popular} if $b \in P$ and \emph{unpopular} otherwise, and deal with the unpopular and popular items separately.
\begin{itemize}
    \setlength\parskip{0pt plus 1pt}
    \item \emph{(Unpopular Items)} Our strategy is to list all solutions $a + b = c$ involving an unpopular item~$b$. Let $Y \subseteq Y_1 + \dots + Y_k$ be the set obtained by sampling some tuples $(y_1, \dots, y_k) \in Y_1 \times \dots \times Y_k$ uniformly at random with rate $\min(1,R / p)$ for some $R = \Theta(\log n)$; for each such tuple we insert the sum $y_1 + \dots + y_k$ into~$Y$. Then we run a meet-in-the-middle algorithm to explicitly list all solutions
    \begin{equation*}
        a + b = x + y
    \end{equation*}
    with $(a, b, x, y) \in A \times (B \setminus P) \times X \times Y$. Note that this condition is equivalent to
    \begin{equation*}
        a - x = y - b,
    \end{equation*}
    an application of ``Fredman's trick''~\cite{Fredman76}. More specifically, enumerate all pairs $(a, x) \in A \times X$, and store each pair in a hash table (with collision lists) indexed by~\makebox{$a - x$}. Then we enumerate all unpopular $b \in B \setminus P$ and all $y \in Y$, and further list all entries $(a, x)$ in the hash table at index $y - b$. For each 4-tuple $(a, b, x, y)$ listed in this way we test if $a + b \in C$. We will show that this algorithm indeed lists any 3SUM solution involving an unpopular item.
    
    \item \emph{(Popular Items)} The only 3SUM solutions not listed in the previous step are those involving a popular item $b$, i.e., the solutions of the subinstance $(A, B', C)$ for $B' = B \cap P$. Recall that also $B' \subseteq \Pop(X, Y_1, \dots, Y_k, -A; p')$ by construction.
\end{itemize}

\paragraph{Correctness}
We show that the first case of the algorithm indeed lists all solutions $a + b = c$ involving an unpopular item $b$. From the initial assumption $C \subseteq \Pop(X, Y_1, \dots, Y_k; p)$ we get that~$c$ can be expressed in at least $p$ ways as
\begin{equation*}
    c = x + y_1 + \dots + y_k,
\end{equation*}
for $x \in X$ and $y_i \in Y_i$. Each such representation corresponds to a unique tuple $(y_1, \dots, y_k)$ (as each such tuple uniquely determines $x$ via $x = c - y_1 - \dots - y_k$). By construction of $Y$, with high probability $1 - \exp(-O(R)) = 1 - 1/n^{\Omega(1)}$, at least one such tuple $(y_1, \dots, y_k)$ is present in $Y$ in the form $y = y_1 + \dots + y_k$. In particular, $c$ can be expressed as $c = x + y$ for some $x \in X$ and $y \in Y$, and therefore there is a solution $a + b = x + y$. This solution is detected by the algorithm, and the correctness follows.

\paragraph{Running Time}
The initial call to \cref{lem:approx-pop-sums} takes time $\tO(|X| \, |Y_1| \dots |Y_k| \, |A| / p')$. In the unpopular case, the construction of $Y$ takes time $O(|Y|) = O(|Y_1| \dots |Y_k| \cdot R / p)$ (which will turn out to be negligible). The meet-in-the-middle algorithm takes time $O(|A| \, |X|)$ to prepare the hash table, plus time $O(|B| \, |Y|)$ to enumerate all pairs $(a, x)$. To bound the total time of listing all 4-tuples $(a, b, x, y)$, note that each unpopular item $b$ can be represented in at most~$2p'$ ways as
\begin{equation*}
    b = x + y_1 + \dots + y_k - a.
\end{equation*}
For each such representation, the tuple $(y_1, \dots, y_k)$ is part of $Y$ only with probability $R / p$, and thus the total number of solutions $a = x + y - b$ is at most $O(p' \cdot R / p) = O(|A| / s \cdot R)$ (in expectation and with high probability) for any unpopular $b$. In particular, the algorithm lists at most $O(|B| \, |A| / s \cdot R)$ such solutions overall. The popular case runs in linear time.
\end{proof}

We are finally ready to prove
\cref{lem:3sum-low-rank-to-low-doubling-contained-one}.

\begin{proof}[Proof of
\cref{lem:3sum-low-rank-to-low-doubling-contained-one}]
Let $c = 1 - 2\rho + 3\epsilon$ and let $k = k(c, \epsilon)$ be the constant from \cref{thm:k-sum-pop-algo}. Note that the statement of \cref{lem:3sum-low-rank-to-low-doubling-contained} allows us to arbitrarily decrease $\epsilon$; all parts of the lemma benefit from smaller $\epsilon$, except for the non-explicit dependence of $\delta$. For convenience we will therefore decrease $\epsilon$ whenever convenient.

\paragraph{Step 1: Removing Popular \boldmath$k$-Sums}
We will first repeatedly remove elements from $A$ until it contains no more popular $k$-sums/differences. More precisely, we will ensure that $\Pop^k(\pm A; p) = \emptyset$ for $p := |A|^{k-1-c}$.

Run \cref{thm:k-sum-pop-algo} on $\pm A$ with parameter $p$, leading to one of two possible outcomes. If it certifies that $\Pop^k(\pm A; p) = \emptyset$, then step 1 is complete and we move on to step 2. Otherwise, \cref{thm:k-sum-pop-algo} identifies a subset $A' \subseteq A$ of size $|A'| \geq |A| / K$ along with a $K$-doubling $K |A|^{2-2\rho+3\epsilon}$-container of~$A'$, for some $K = O(|A|^\epsilon)$. In this event we take $(A', B, C)$ as one of the constructed subinstances (the container is not quite as claimed; we deal with this issue in a moment). We then remove~$A'$ from~$A$, and repeat the execution of step 1.

All in all, there can be at most $\tO(|A|^\epsilon)$ repetitions as each repetition removes $\Omega(|A|^{1-\epsilon})$ elements from $A$. In particular, the number of constructed subinstances is $\tO(|A|^\epsilon)$, and each subinstance is $O(|A|^\epsilon)$-doubling $O(|A|^{2-2\rho+4\epsilon})$-contained. By \cref{lem:3sum-shrink-container} (applied with $t = |A|^{4\epsilon}$) we can further refine these subinstances to $O(|A|^{5\epsilon})$ instances that are $\tO(|A|^{5\epsilon})$-doubling $O(|A|^{2-2\rho})$-contained. This is as claimed when reducing $\epsilon$ by a constant.

\paragraph{Step 2: Small Sets}
Next, we easily assert that $n \leq |A|^{1+\epsilon}$. Otherwise we can solve the whole instance trivially in time $n^{2-\Omega(\epsilon)}$, and report the set of solutions $\mathcal T$.

\paragraph{Step 3: Popular \boldmath$k$-Sum Covers}
Once we reach this step, we will exhaustively list all 3SUM solutions. Let $\mathcal T$ be initially empty. We will successively remove elements from $B$ and $C$, while maintaining the invariant that each 3SUM solution that is eliminated in this way has already been listed in $\mathcal T$. The algorithm runs in iterations $\ell = 1, \dots, k$. We also maintain the second invariant that, before the $\ell$-th iteration, 
\begin{equation*}
    \begin{cases}
        B \subseteq \Pop(X, Y_1, \dots, Y_\ell; p_\ell) &\text{if $\ell$ is even,} \\
        C \subseteq \Pop(X, Y_1, \dots, Y_\ell; p_\ell) &\text{if $\ell$ is odd.}
    \end{cases}
\end{equation*}
To satisfy this invariant initially (for $\ell = 1$) we pick $Y_1 = X$ and $p_1 = 1$; indeed, then the invariant follows from the assumption $C \subseteq X + X$. Let $s = |A|^{\epsilon/k}$. Each iteration depends on the parity of $\ell$:
\begin{itemize}
    \item \emph{If $\ell$ is odd:} Apply \cref{lem:3sum-low-rank-pop-cover} on the 3SUM instance $(A, B, C)$, with parameters $p_\ell$ and $s$ and sets $X, Y_1, \dots, Y_\ell$; the invariant makes sure that the premise $C \subseteq \Pop(X, Y_1, \dots, Y_\ell; p_\ell)$ is fulfilled. This yields a set of solutions $\mathcal T'$ and a subinstance $(A, B' \subseteq B, C)$. We update:
    \begin{align*}
        \mathcal T &:= \mathcal T \cup \mathcal T', \\[1ex]
        B &:= B', \\[1ex]
        Y_{\ell+1} &:= -A, \\
        p_{\ell+1} &:= \frac{p_\ell |A|}{s}.
    \end{align*}
    \cref{lem:3sum-low-rank-pop-cover} states that each 3SUM solution is either in $\mathcal T'$ or part of the instance $(A, B', C)$. Thus, after the update each solution is either in $\mathcal T$ or part of the instance $(A, B, C)$, and therefore the first invariant is maintained.
    
    \cref{lem:3sum-low-rank-pop-cover} also guarantees that $B' \subseteq \Pop(X, Y_1, \dots, Y_\ell, -A; p_\ell |A| / s)$. After the update, this precisely maintains the second invariant.
    \item \emph{If $\ell$ is even:} Observe that our 3SUM instance $(A, B, C)$ is equivalent to the appropriately transposed instance $(A, -C, -B)$, i.e., their solutions stand in one-to-one correspondence. Apply \cref{lem:3sum-low-rank-pop-cover} on $(A, -C, -B)$ with parameters $p_\ell$ and $s$ and sets $-X, -Y_1, \dots, -Y_\ell$; the premise $-B \subseteq \Pop(-X, -Y_1, \dots, -Y_\ell; p_\ell)$ is satisfied by our invariant. This yields a set of solutions $\mathcal T'$ and a subinstance \makebox{$(A, C' \subseteq -C, -B)$}. We update:
    \begin{align*}
        \mathcal T &:= \mathcal T \cup \set{(a, b, c) : (a, -c, -b) \in \mathcal T'}, \\[1ex]
        C &:= -C', \\[1ex]
        Y_{\ell+1} &:= A, \\
        p_{\ell+1} &:= \frac{p_\ell |A|}{s}.
    \end{align*}
    Similarly to before, \cref{lem:3sum-low-rank-pop-cover} states that each transposed 3SUM solution is either in $\mathcal T'$ or part of the instance $(A, C', -B)$. Thus, after the update each original solution is either in $\mathcal T$ or part of the instance $(A, B, C)$, and therefore the first invariant is maintained.

    \cref{lem:3sum-low-rank-pop-cover} also guarantees that $C' \subseteq \Pop(-X, -Y_1, \dots, -Y_\ell, -A; p_\ell |A| / s)$. After the update, this precisely maintains the second invariant.
\end{itemize}
Once the final $k$-th iteration is complete, we report the set $\mathcal T$ and the subinstances constructed in step 1. This completes the description of the algorithm.

\paragraph{Correctness}
The correctness of steps 1 and 2 is fairly obvious; the interesting part is to argue that in step 3 we exhaustively list all solutions of $(A, B, C)$. The key correctness claim is that once the $k$-th iteration terminates, $B$ or $C$ must be empty. Suppose otherwise that both are non-empty. Then our invariant implies that
\begin{equation*}
    \Pop(X, Y_1, \dots, Y_{k+1}; p_{k+1}) \neq \emptyset.
\end{equation*}
Recall that we picked $p_1 = 1$ and $p_{\ell+1} = p_\ell |A| / s$, and thus
\begin{equation*}
    p_{k+1} = \frac{|A|^k}{s^k}.
\end{equation*}
Moreover, we picked $Y_1 = X$ and $Y_2, \dots, Y_{k+1} \subseteq \pm A$. In combination, we get
\begin{equation*}
    \Pop\bigg(X, X, \underbrace{\pm A, \dots, \pm A}_{\text{$k$ times}}; \frac{|A|^k}{s^k}\bigg) \neq \emptyset.
\end{equation*}
By two applications of \cref{obs:pop-remove} it further follows that
\begin{equation*}
    \Pop^k\parens*{\pm A; \frac{|A|^k}{|X|^2 s^k}} \neq \emptyset.    
\end{equation*}
But note that
\begin{equation*}
    \frac{|A|^k}{|X|^2 s^k} \geq \frac{|A|^k}{n^{2-2\rho} |A|^{\epsilon}} \geq \frac{|A|^k}{|A|^{2-2\rho+2\epsilon} |A|^\epsilon} = |A|^{k-2+2\rho-3\epsilon} = |A|^{k-1-c} = p,
\end{equation*}
where we used that $|X| \leq n^{1-\rho}$ and that $n \leq |A|^{1+\epsilon}$ as asserted in step 2. Hence,
\begin{equation*}
    \Pop^k(\pm A; p) \neq \emptyset,
\end{equation*}
which contradicts step 1.

\paragraph{Running Time}
The running time of step 1 is dominated by the calls to \cref{thm:k-sum-pop-algo}, running in time $\tO(|A|^{1+c+\epsilon}) = \tO(|A|^{2-2\rho+4\epsilon})$ for each call, and in total $\tO(|A|^{2-2\rho+5\epsilon})$ over all $\tO(|A|^{\eps})$ calls. Step 2 runs in time $O(n^{2-\delta_0})$ for some constant $\delta_0 = \delta_0(\epsilon) > 0$. Finally, focus on step 3. For $1 \le \ell \le k$, the $\ell$-th iteration is dominated by the call to \cref{lem:3sum-low-rank-pop-cover}. In the odd case, this takes time
\begin{align*}
    &\tO\parens*{\frac{s |X| \, |Y_1| \dots |Y_\ell|}{p_\ell} + |A| \, |X| + \frac{|B| \, |Y_1| \dots |Y_\ell|}{p_\ell} + \frac{|A| \, |B|}{s}} \\
    &\qquad{} = \tO\parens*{\frac{s^\ell |X|^2 |A|^{\ell-1}}{|A|^{\ell-1}} + |A| \, |X| + \frac{s^{\ell-1} |B| \, |X| \, |A|^{\ell-1}}{|A|^{\ell-1}} + \frac{|A| \, |B|}{s}} \\[1ex]
    &\qquad{} = \tO\parens*{n^{2-2\rho+\epsilon \ell/k} + n^{2-\rho} + n^{2-\rho+\epsilon(\ell-1)/k} + n^{2-\epsilon/k}},
\end{align*}
and the running time of the even case can be bounded analogously. Note that the fourth term can be bounded by $O(n^{2-\delta_1})$ for some constant $\delta_1 = \delta_1(\rho, \epsilon) > 0$. This leaves the first three terms, and the time complexity of step 1, which are crudely bounded by $\tO(n^{2-\rho+5\epsilon})$. If necessary, decrease $\epsilon$ further so that $\epsilon < \rho/10$. Then let $\delta = \min(\delta_0, \delta_1, \rho/2)$, and the proof is complete.
\end{proof}

\subsection{Low-Doubling Containment for All Sets} \label{sec:3sum-low-rank:sec:low-doubling-contained-all}
Next, we demonstrate that low-doubling containment for one set can be extended to low-doubling containment for all three sets:

\begin{lemma}[Low-Doubling Containment for All Sets] \label{lem:3sum-low-doubling-contained-one-to-all}
Let $\epsilon > 0$. There is a randomized algorithm that, given a 3SUM instance $(A, B, C)$ along with a $(K, N)$-container $H$ of $A$, computes
\begin{itemize}
    \item $O(n^\epsilon)$ subinstances $(A, B' \subseteq B, C' \subseteq C)$ so that $B'$ is contained in a translate of $H$ and $C'$ is contained in a translate of $H + H$, and
    \item a set of 3SUM solutions $\mathcal T \subseteq A \times B \times C$,
\end{itemize}
so that each solution is either part of exactly one subinstance or part of $\mathcal T$. The algorithm runs in time $\tO(K N n^{\epsilon} + K n^{2-\epsilon})$.
\end{lemma}

The proof closely follows Jin and Xu's algorithm for low-doubling 3SUM~\cite[Section~3.3]{JinX23}. The key ingredient is the following covering-type lemma~\cite{JinX23}.

\begin{lemma}[{\cite[Lemma~3.6]{JinX23}}] \label{lem:3sum-cover-low-doubling}
There is a randomized algorithm that, given integer sets $H, B, C$ of size at most~$N$ with $|H + H| \leq K |H|$, computes subsets \makebox{$B_1, \dots, B_m \subseteq B$} and $C_1, \dots, C_m \subseteq C$ such that:
\begin{enumerate}[label=(\roman*)]
    \item\label{lem:3sum-cover-low-doubling:itm:B-partition} $B_1, \dots, B_m$ form a partition of $B$, 
    \item\label{lem:3sum-cover-low-doubling:itm:C-almost-partition} $\sum_i |C_i| \leq O(|C| K \log N)$, 
    \item\label{lem:3sum-cover-low-doubling:itm:translates} for all $1 \leq i \leq m$, there is a translate $s_i$ so that $B_i \subseteq (s_i + H) \cap B$ and $C_i = (s_i + H + H) \cap C$. 
\end{enumerate}
The algorithm runs in time $\tO(N K \log u)$ whenever $H, B, C \subseteq [-u, u]$.
\end{lemma}

\begin{proof}[Proof of \cref{lem:3sum-low-doubling-contained-one-to-all}]
The sets $H, B, C$ satisfy the premise of \cref{lem:3sum-cover-low-doubling}, yielding sets $B_1, \dots, B_m$ and $C_1, \dots, C_m$. By Properties~\ref{lem:3sum-cover-low-doubling:itm:B-partition} and~\ref{lem:3sum-cover-low-doubling:itm:translates} and the fact that $A\subseteq H$, each 3SUM solution from $(A, B, C)$ must be contained in a subinstance $(A, B_i, C_i)$.

Consider all such subinstances, distinguishing two cases. If $|B_i| \leq |B|^{1-\epsilon}$, we simply solve the subinstance by brute-force in time $O(|B_i| \, |C_i|) = O(|B|^{1-\epsilon} \, |C_i|)$. Otherwise, we keep it as one of the designated subinstances. Property~\ref{lem:3sum-cover-low-doubling:itm:translates} guarantees that $B_i$ is contained in a translate of $H$ and~$C_i$ is a translate of $H + H$, as claimed. Moreover, we designate at most $|B|^\epsilon$ subinstances in total, recalling that the sets $B_1, \dots, B_m$ form a partition of $B$ by Property~\ref{lem:3sum-cover-low-doubling:itm:B-partition}. This completes the description of the algorithm. The correctness is immediate.

We focus on the running time. The initial call to \cref{lem:3sum-cover-low-doubling} takes time $\tO(N K \log u)$ where $N$ is the maximum size of $H, B, C$. Solving the subinstance $(A, B_i, C_i)$ takes time $\tO(|B|^{1-\epsilon} \, |C_i|)$ by brute-force, or linear time to copy the instance to the set of reported instances. Therefore, the total time across all subinstances is
\begin{equation*}
    O\parens*{\sum_{i=1}^m |B|^{1-\epsilon} \, |C_i| + n^{1+\epsilon}} = \tO\parens{|B|^{1-\epsilon} \cdot K |C| + n^{1+\epsilon}} = \tO(K n^{2-\epsilon} + n^{1+\epsilon}),
\end{equation*}
by Property~\ref{lem:3sum-cover-low-doubling:itm:C-almost-partition}.
Hence, the total time of the algorithm is 
\[\tO(NK + K n^{2-\epsilon} + n^{1+\epsilon}) = \tO(KNn^{\epsilon}+Kn^{2-\epsilon}) \]
as claimed.
\end{proof}

Finally, the proof of \cref{lem:3sum-low-rank-to-low-doubling-contained} is now a simple combination of \cref{lem:3sum-low-rank-to-low-doubling-contained-one,lem:3sum-low-doubling-contained-one-to-all}.

\begin{proof}[Proof of \cref{lem:3sum-low-rank-to-low-doubling-contained}]
We start from a 3SUM instance $(A, B, C)$ where $C$ has rank at most $n^{1-\rho}$. Let~$\epsilon_1 = \frac{\epsilon}{4}$ and $\epsilon_2 = \frac{3\epsilon}{4}$. \cref{lem:3sum-low-rank-to-low-doubling-contained-one} reduces this instance to~$O(n^{\epsilon_1})$ subinstances~\makebox{$(A' \subseteq A, B, C)$}, each provided with a $(K = O(n^{\epsilon_1}), N = O(n^{2-2\rho}))$-container~$H$ of~$A'$. We pass each such instance to \cref{lem:3sum-low-doubling-contained-one-to-all}, obtaining $O(n^{\epsilon_2})$ subinstances $(A', B' \subseteq B, C' \subseteq C)$ where $B'$ is covered by a translate of $H$ and $C'$ is covered by a translate of $H + H$. Pick
\begin{equation*}
    H' = H \cup B' \cup C';
\end{equation*}
we claim that this is a $(O(K^4), O(N))$-container of $A' \cup B' \cup C'$. The containment condition is clear, and the size of $H'$ is bounded by $|H| + |B| + |C| \leq N + 2n \leq O(N)$. Regarding the doubling of~$H'$, observe that $H' + H'$ can be covered by $O(1)$ translates of $H + H + H + H$. Thus, by the Plünnecke--Ruzsa inequality (\cref{lem:plunnecke-ruzsa}),
\begin{equation*}
    |H' + H'| \leq O(|H + H + H + H|) \leq O(K^4) |H| \leq O(K^4) |H'|.
\end{equation*}
In summary, it follows that $H'$ is indeed a $(O(n^{4\epsilon_1}) = O(n^\epsilon), O(n^{2-2\rho}))$-container of $A' \cup B' \cup C'$. Moreover, the total number of instances is $O(n^{\epsilon_1 + \epsilon_2}) = O(n^\epsilon)$.

Let us finally analyze the running time.  The reduction in \cref{lem:3sum-low-rank-to-low-doubling-contained-one} takes time $O(n^{2-\delta_1})$ for some constant $\delta_1 = \delta_1(\rho, \epsilon_1) > 0$. Each of the $O(n^{\epsilon_1})$ calls to \cref{lem:3sum-low-doubling-contained-one-to-all} takes time $\tO(K N n^{\epsilon_2} + K n^{2-\epsilon_2})$, amounting to a total time of
\begin{equation*}
    \tO(n^{\epsilon_1} \cdot (K N n^{\epsilon_2} + K n^{2-\epsilon_2})) = \tO(n^{2-2\rho+2\epsilon_1+\epsilon_2} + n^{2+2\epsilon_1-\epsilon_2}) = \tO(n^{2-2\rho+5\epsilon/4} + n^{2-\epsilon/4}).
\end{equation*}
If necessary, decrease $\epsilon$ to $4\rho/5$ so that the first term becomes $O(n^{2-\rho})$. Then the running time is as claimed for any constant $0 < \delta < \min(\delta_1, \rho, \frac{\epsilon}{4})$.
\end{proof}
\section{Universe Reductions for Min-Plus Convolution} \label{sec:univ-reduct}
In this section we design our universe reductions for Min-Plus Convolution. These work somewhat differently in two regimes, and thus we design a ``small-universe reduction'' from $[n]$ to $[n^\gamma]$, for~\makebox{$\gamma < 1$}, and a separate ``large-universe reduction'' from~$[n^c]$ to $[n]$, for $c > 1$. We start with the simpler small-universe reduction in \cref{sec:univ-reduct:sec:sm}. Then, to prepare the more involved large-universe reduction, in \cref{sec:univ-reduct:sec:doubling} we first design a ``doubling reduction'', proving that the index sets of the given sequences can be assumed to have low-doubling containers. In \cref{sec:univ-reduct:sec:mono} we then design a universe reduction for the \emph{Monotone} Min-Plus Convolution problem. Finally, in \cref{sec:univ-reduct:sec:lg} we put these two pieces together to obtain the large-universe reduction, conditioned on \cref{hypo:sum-preserving}.

\subsection{Small-Universe Reduction} \label{sec:univ-reduct:sec:sm}
The small-universe reduction can be stated as follows:

\begin{lemma}[Small-Universe Reduction]
\label{lem:small-universe-listing}
For any $0 \leq \gamma \leq 1$ and $\epsilon > 0$, 
\begin{equation} \label{eq:small-universe-reduction}
    \MinPlusConv(n \mid u \leq n) = \tO\parens*{n^\epsilon \cdot \MinPlusConvListing(n, n^{1-\gamma} \mid u \le n^\gamma) + \ThreeSUM(n \mid r \le n^{1-\epsilon/2})}. 
\end{equation}
\end{lemma}
\begin{proof}
Let $(A, B)$ be the given Min-Plus Convolution instance with universe $[0, n]$. Let $q = \ceil{n^{1-\gamma}}$, and rewrite these sequences as 
\begin{align*}
    A &= A_1 \cdot q + A_0, \\
    B &= B_1 \cdot q + B_0, 
\end{align*}
where $A_1, B_1$ have entries less than $n / q \leq n^{\gamma}$ (or $\bot$), and $A_0, B_0$ have entries less than $q$ (or~$\bot$). In fact, we may assume without loss of generality that $A_0, B_0$ have entries bounded by $q / 2$ (or~$\bot$). Indeed, simply split $A_0$ into two subsequences, retaining in one the entries less than $q/2$ and in the other the remaining entries. From the second one subtract $q / 2$ globally. This leaves two subproblems for which $A_0$ is as claimed. Proceeding similarly with $B_0$ yields at most four subinstances for which~$A_0, B_0$ have entries less than $q/2$.

We compute $C_1 = A_1 * B_1$ in time $O(\MinPlusConv(n \mid u \le n^\gamma))$ (which is dominated by the first term in \cref{eq:small-universe-reduction}). As $A_0, B_0$ have entries less than $q/2$, we know that
\begin{equation*}
    C[k] = C_1[k] \cdot q + C_0[k], 
\end{equation*}
where
\begin{equation*}
    C_0[k] = \min\set{A_0[i] + B_0[k-i] : i \in [0, n],\, A_1[i] + B_1[k - i] = C_1[k]},
\end{equation*}
so the remaining task is to compute $C_0[k]$.

In fact, we can learn a close approximation of $C_0[k]$ by a simple scaling trick. Let $A_0' = \floor{A_0 / 2}$ and $B_0' = \floor{B_0 / 2}$ (applied coordinate-wise), and \emph{recursively} compute
\begin{equation*}
    C_0'[k] = \min\set{A_0'[i] + B_0'[k - i] : i \in [0, n],\, A_1[i] + B_1[k - i] = C_1[k]}.
\end{equation*}
It is easy to check that $2C_0'[k] \leq C_0[k] \leq 2 C_0'[k] + 2$, leaving only three candidates for each $C[k]$. In what follows we will determine the correct choice among these three candidates.

We call an index $i$ satisfying $A_1[i] + B_1[k - i] = C_1[k]$ a \emph{pseudo-witness} of $k$. We say that $k$ is \emph{popular} if it has at least $n^{1-\gamma+\epsilon}$ pseudo-witnesses, and \emph{unpopular} otherwise. We will compute the value of~$C[k]$ by one of the two following algorithms, depending on whether $k$ is popular or unpopular. 
\begin{itemize}
    \item \emph{(Unpopular)} We first show how to compute $(A * B)[k]$ for all unpopular indices $k$. For any such~$k$, we can afford to list all $n^{1-\gamma+\epsilon}$ pseudo-witnesses, then compute $C_0[k]$ and $C[k]$ via the previous identities. This takes time
    \begin{equation*}
        O\parens{\MinPlusConvListing(n, n^{1-\gamma+\epsilon} \mid u \le n^\gamma)} = O\parens{n^\epsilon \cdot \MinPlusConvListing(n, n^{1-\gamma} \mid u \le n^\gamma)},
    \end{equation*}
    by \cref{lem:min-plus-conv-listing}, which constitutes the first term in \cref{eq:small-universe-reduction}.
    
    \item\emph{(Popular)} It remains to compute $C[k]$ for all popular indices $k$. Let $X \subseteq [0, n]$ be a random subset including each index independently with probability $p = \sqrt{(100 \log n) / n^{1-\gamma+\epsilon}}$. Then let
    \begin{align*}
        \mathcal X_1 &= \set{(i, A_1[i] \cdot q + a) : i \in X,\, 0 \leq a < \ceil{\sqrt{q}}}, \\
        \mathcal X_2 &= \set{(j, B_1[j] \cdot q + b \cdot \ceil{\sqrt{q}}) : j \in X,\, 0 \leq b < \ceil{\sqrt{q}}}, \\
        \mathcal X &= \mathcal X_1 \cup \mathcal X_2.
    \end{align*}
    Consider the All-Integer 3SUM instance $(\mathcal A, \mathcal B, \mathcal C)$ on the following index sets 
    \begin{align*}
        \mathcal A &= \set{(i, A[i]): i \in [0, n]}, \\
        \mathcal B &= \set{(j, B[j]): j \in [0, n]}, \\
        \mathcal C & = \set{(k, C_1[k] \cdot q + 2 C_0'[k] + z) : k \in [0, 2n],\ z \in \set{0, 1, 2}} \cap (\mathcal X_1 + \mathcal X_2).
    \end{align*}
    By construction, $\mathcal C \subseteq \mathcal X_1 + \mathcal X_2$ and thus $\mathcal C$ has rank at most
    \begin{equation*}
        r := \max\set{|\mathcal X_1|, |\mathcal X_2|} \leq O(|X| \sqrt q) \leq O(p n \cdot \sqrt{q}) = \tO(n^{1-\epsilon/2}),
    \end{equation*}
    where $|X| \leq 2 p n$ follows from Chernoff's bound with high probability. We solve the All-Integer 3SUM instance in time
    \begin{equation*}
        \tO(\ThreeSUM(n \mid r \leq n^{1-\epsilon/2})),
    \end{equation*}
    constituting the second term in the claimed running time. For each $k$, choose the smallest of the three items $(k, c) \in \mathcal C$ that is involved in a 3SUM solution, and set $C[k] = c$.
    
    We argue that we thereby correctly assign $C[k]$ for all popular indices $k$. By definition, a popular index has at least $n^{1-\gamma+\epsilon}$ pseudo-witnesses $i$. Thus, with good probability $X$ contains one such pseudo-witness $i$ and simultaneously $k - i$; the probability of missing all such pairs is
    \begin{equation*}
        (1 - p^2)^{n^{1-\gamma+\epsilon}} \leq \exp(-p^2 n^{1-\gamma+\epsilon}) \leq n^{-100}.
    \end{equation*}
    In the successful event, we have $(i, A_1[i] \cdot q  + a) \in \mathcal X_1$ and $(j, B_1[j] \cdot q + b  \cdot \ceil{\sqrt q}) \in \mathcal X_2$ for all $0 \leq a, b < \sqrt{q}$, and thus $(k, C_1[k] \cdot q + a + b \cdot \ceil{\sqrt q}) \in \mathcal X_1 + \mathcal X_2$. For an appropriate choice of $0 \leq a, b < \sqrt{q}$, this becomes exactly $(k, C_1[k] \cdot q + C_0[k]) = (k, C[k])$. In this event, $(k, C[k])$ is part of the All-Integer 3SUM instance and we correctly identify $C[k]$. \qedhere
\end{itemize}
\end{proof}

\begin{corollary}[Small-Universe Reduction for Min-Plus Convolution Listing]
    \label{cor:small-universe-listing}
    Let $\gamma \in [0, 1]$. Listing $n^{1-\gamma}$ witnesses in Min-Plus Convolution instances with universe $[n^\gamma]$ requires $n^{2-o(1)}$ time, under the strong Min-Plus Convolution Hypothesis.
\end{corollary}
\begin{proof}
Suppose there is such an algorithm with running time $O(n^{2-\eps'})$. We invoke \cref{lem:small-universe-listing} with $\gamma := \gamma$ and $\epsilon = \eps' / 2$. Combined with \cref{thm:3sum-low-rank}, this gives us a truly subquadratic time algorithm for $\MinPlusConv(n \mid u \le n)$, contradicting the strong Min-Plus Convolution Hypothesis. 
\end{proof}

Combined with the listing trick from \cref{lem:min-plus-conv-listing}, we directly obtain \cref{thm:small-universe-min-plus}: 

\smalluniverseminplus*

\subsection{Doubling Reduction} \label{sec:univ-reduct:sec:doubling}
The purpose of this subsection is to prove the following lemma. Intuitively, it states that any Min-Plus Convolution instance $(A, B)$ can be assumed to have a low-doubling container for its joint index set $\mathcal I(A) \cup \mathcal I(B)$ (\cref{defn:indexset}). The proof is close to~\cite[Lemma~5.7]{Fischer26}.

\begin{lemma}[Doubling Reduction] \label{lem:doubling-reduct}
For any $\epsilon > 0$ and $0<\rho\le 1/2$ there is some $\delta = \delta(\rho,\epsilon) > 0$ such that the following holds. There is a randomized reduction from any Min-Plus Convolution instance $(A, B)$ over the universe $[u]$ to $O(n^\epsilon \log u)$ Min-Plus Convolution Listing instances $(A', B')$ with the same size $n$, over the same universe $[u]$, and with listing threshold $O(n^{2\rho})$, each equipped with a $(O(n^\epsilon), O(n^{2-2\rho}))$-container of the joint index set $\mathcal I(A') \cup \mathcal I(B')$. The reduction runs in time $O((n^{1+2\rho+\epsilon} + n^{2-\delta}) \log u)$.
\end{lemma}
\begin{proof}
Let $(A, B)$ denote the given Min-Plus Convolution instance; each entry of $A$ and $B$ is either an integer in $[u]$ or $\bot$ (viewed as $+\infty$). We call $i \in [0, n]$ a \emph{witness} of $(A * B)[k]$ if $A[i] + B[k - i] = (A * B)[k]$, and a \emph{$q$-pseudo-witness} if~\makebox{$A[i] + B[k - i] < (A * B)[k] + q$}. Let $L = \ceil{\log (2u)}$ and let $t = 100 \ceil{n^{2\rho}}$. We classify the indices $k$ into three different kinds, based on their number of pseudo-witnesses:
\begin{enumerate}
    \item $k$ is \emph{exceptionally popular} if $(A * B)[k]$ has at least $t$ $1$-pseudo-witnesses (i.e., witnesses),
    \item $k$ is \emph{exceptionally unpopular} if $(A * B)[k]$ has at most $t$ $2^L$-pseudo-witnesses,
    \item and otherwise $k$ is \emph{ordinary}. In this case, there is some scale $0 \leq \ell < L$ such that $(A * B)[k]$ has at most $t$ $2^\ell$-pseudo-witnesses, and simultaneously at least $t$ $2^{\ell+1}$-pseudo-witnesses. We say that $(A * B)[k]$ is \emph{relevant} at that scale $\ell$.
\end{enumerate}
The reduction deals with the exceptionally popular, exceptionally unpopular, and ordinary pairs in three different cases (though we never have to explicitly determine which index falls into which case). Out of these only the third one is interesting; the others are corner cases. Throughout we maintain a sequence $C$ initialized to $\bot$. We will only ever modify $C$ via updates of the form $C[k] \gets \min\set{C[k], A[i] + B[k - i]}$, ensuring that we never underestimate the convolution.
\begin{enumerate}
    \item\label{lem:doubling-reduct:itm:exc-pop}\emph{(Exceptionally Popular)} Repeat the following $N = \ceil{n^{2-2\rho} \cdot 100 \log n}$ times: Randomly sample $i, j \in [0, n]$, and update $C[k] \gets \min\set{C[k], A[i] + B[j]}$ for $k = i + j$. We argue that, afterwards, $C[k]$ is correctly assigned for each exceptionally popular index $k$. Indeed, for any such index the probability of missing all pairs $(i, k - i)$ where $i$ is a witness of $k$ is at most
    \begin{equation*}
        \parens*{1 - \frac{t}{n^2}}^N \leq \exp\parens*{-\frac{t N}{n^2}} \leq \frac{1}{n^{100}}.
    \end{equation*}
    Thus, with high probability we list one such pair and consequently assign $C[k] = (A * B)[k]$.
    \item\label{lem:doubling-reduct:itm:exc-unpop}\emph{(Exceptionally Unpopular)} By definition, a $2^L$-pseudo-witness $i$ of $(A * B)[k]$ is any index satisfying that $A[i] \neq \bot$ and $B[k - i] \neq \bot$. Thus, to list the at most $t$ $2^L$-pseudo-witnesses of the exceptionally unpopular indices $k$, we can simply construct the sequences $A', B' \in \set{0, \bot}^{[0, n]}$ obtained from $A, B$ by replacing any non-$\bot$ entry by $0$,
   and then for each exceptionally unpopular index we can directly list $t$ witnesses of a Boolean convolution (\cref{lem:listingboolean}), running in time $\tO(t n) = \tO(n^{1+2\rho})$.
    The output is a set of witnesses, i.e., triples $(i, j, k)$ with $A'[i] + B'[j] = (A' * B')[k]$, or equivalently, $A[i]\neq \bot$ and $B[j]\neq \bot$. For each such triple update $C[k] \gets \min\set{C[k], A[i] + B[j]}$.
    \item\label{lem:doubling-reduct:itm:ord}\emph{(Ordinary)} For all $0 \leq \ell < L$, all $x, y \in \set{0, 1}$ and all $z \in [-4, 4]$, repeat the following steps $100 \log n$ times:
    \begin{enumerate}
        \item\label{lem:doubling-reduct:itm:ord:itm:construct} Construct the sequences $P, Q \in \Int^{[0, n]}$ by
        \begin{align*}
            P[i] &=
            \begin{cases}
                \displaystyle\mathrlap{\floor*{\frac{A[i]}{2^\ell}}}\phantom{\floor*{\frac{B[j]}{2^\ell}}} &\text{if $(A[i] \bmod 2^\ell) \in x2^{\ell-1} + [0, 2^{\ell-1})$,} \\
                \bot &\text{otherwise,}
            \end{cases} \\
            Q[j] &=
            \begin{cases}
                \displaystyle\floor*{\frac{B[j]}{2^\ell}} &\text{if $(B[j] \bmod 2^\ell) \in y2^{\ell-1} + [0, 2^{\ell-1})$,} \\
                \bot &\text{otherwise.}
            \end{cases}
        \end{align*}
        Sample $X \subseteq [0, n]$ uniformly at random with rate $p = n^{-\rho}$. By Chernoff's bound, this set has size $|X| = O(n^{1-\rho})$ with high probability. Compute the sequence~$R \in \Int^{[0, 2n]}$ defined by
        \begin{equation*}
            R[k] = \min_{\substack{i, j \in X\\i+j=k}} \floor*{\frac{A[i]}{2^\ell}} + \floor*{\frac{B[j]}{2^\ell}} + z
        \end{equation*}
        by brute-force. Moreover, construct the set
        \begin{equation*}
            \mathcal X = \set*{\parens*{i, \floor*{\frac{A[i]}{2^\ell}}} : i \in X} \cup \set*{\parens*{j, \floor*{\frac{B[j]}{2^\ell}} + z} : j \in X},
        \end{equation*}
        and observe that $\mathcal R := \mathcal I(R) \subseteq \mathcal X + \mathcal X$. In particular, the rank of $\mathcal R$ is at most $r = |\mathcal X| \leq 2|X| = O(n^{1-\rho})$.
        
        \item\label{lem:doubling-reduct:itm:ord:itm:3sum} For convenience, write $\mathcal P := \mathcal I(P)$ and $\mathcal Q := \mathcal I(Q)$. Apply \cref{lem:3sum-low-rank-to-low-doubling-contained} to the 3SUM instance $(\mathcal P, \mathcal Q, \mathcal R)$ and the set $\mathcal X$ (which satisfies $\mathcal R \subseteq \mathcal X + \mathcal X$, as required). This is an instance of 3SUM over the group $\Int^2$, which can be easily mapped to plain integers via the standard embedding into disjoint bit-blocks. The output of \cref{lem:3sum-low-rank-to-low-doubling-contained} consists of
        \begin{itemize}
            \item $O(n^\epsilon)$ subinstances $(\mathcal P' \subseteq \mathcal P, \mathcal Q' \subseteq \mathcal Q, \mathcal R' \subseteq \mathcal R)$, each provided with a $(O(n^\epsilon), O(n^{2-2\rho}))$-container $\mathcal H$ of $\mathcal P' \cup \mathcal Q' \cup \mathcal R'$, and
            \item a set of 3SUM solutions $\mathcal T \subseteq \mathcal P \times \mathcal Q \times \mathcal R$. 
        \end{itemize}

        \item\label{lem:doubling-reduct:itm:ord:itm:brute-force} For each triple $((i, a), (j, b), (k, c)) \in \mathcal T$: Update $C[k] \gets \min\set{C[k], A[i] + B[j]}$.

        \item\label{lem:doubling-reduct:itm:ord:itm:reduct} For each subinstance $(\mathcal P', \mathcal Q', \mathcal R')$: Let $P'$ denote the restriction of $P$ so that $\mathcal P' = \mathcal I(P')$. Similarly, let $Q'$ denote the restriction of $Q$ so that $\mathcal Q' = \mathcal I(Q')$. Let $(P', Q')$ be one of the designated Min-Plus Convolution Listing instances with threshold $t$. We solve the instance by the reduction. The output is a set of witnesses, i.e., triples $(i, j, k)$ with $P'[i] + Q'[j] = (P' * Q')[k]$. For each such triple update $C[k] \gets \min\set{C[k], A[i] + B[j]}$.
    \end{enumerate}
\end{enumerate}
In the end, we return the sequence $C$ computed before. This completes the description of the algorithm. It should be clear that cases 1 and 2 correctly assign $C[k] = (A * B)[k]$ for all exceptional indices $k$, but it remains to show that the same holds for the ordinary ones.

Fix any ordinary index $k^*$. Let $0\le \ell<L$ be the scale at which $k^*$ becomes relevant (i.e., $(A * B)[k^*]$ has at most $t$ $2^\ell$-pseudo-witnesses and at least $t$ $2^{\ell+1}$-pseudo-witnesses).  Also, fix any witness $i^*$ of~$(A * B)[k^*]$ and let $j^* = k^* - i^*$. Let $x, y \in \set{0, 1}$ be such that~$(A[i^*] \bmod 2^\ell) \in x2^{\ell-1} + [0, 2^{\ell-1})$ and~$(B[j^*] \bmod 2^\ell) \in y2^{\ell-1} + [0, 2^{\ell-1})$. Now focus on any iteration corresponding to our choice of~$\ell, x, y$ (and some $z$ to be determined soon).

\begin{claim} \label{lem:doubling-reduct:clm:scaling}
In any iteration corresponding to these choices of $\ell, x, y$, it holds that:
\begin{enumerate}[label=(\roman*)]
    \item\label{lem:doubling-reduct:clm:scaling:itm:one-witness} $i^*$ is a witness of $(P * Q)[k^*]$, and $P[i^*]+Q[j^*] \le \lfloor (A*B)[k^*]/2^\ell\rfloor$, and
    \item\label{lem:doubling-reduct:clm:scaling:itm:all-pseudo-witnesses} every witness of $(P * Q)[k^*]$ is a $2^\ell$-pseudo-witness of $(A * B)[k^*]$.
\end{enumerate}
\end{claim}
\begin{proof}
By definition, $P$ is only supported on the indices $i$ with $(A[i] \bmod 2^\ell) \in x2^{\ell-1} + [0, 2^{\ell-1})$ and analogously, $Q$ is only supported on the indices $j$ with $(B[j] \bmod 2^\ell) \in y2^{\ell-1} + [0, 2^{\ell-1})$. Let~$i, j$ be two indices satisfying these requirements. Then:
\begin{equation*}
    P[i] + Q[j] = \floor*{\frac{A[i]}{2^\ell}} + \floor*{\frac{B[j]}{2^\ell}} = \frac{A[i] + B[j] - (A[i] \bmod 2^\ell) - (B[j] \bmod 2^\ell)}{2^\ell}.
\end{equation*}
This can be upper bounded by
\begin{equation*}
    P[i] + Q[j] \leq \frac{A[i] + B[j] - (x + y)2^{\ell-1}}{2^\ell},
\end{equation*}
and correspondingly lower bounded by
\begin{equation*}
    P[i] + Q[j] > \frac{A[i] + B[j] - (x + y)2^{\ell-1} - 2^{\ell}}{2^\ell} = \frac{A[i] + B[j] - (x + y)2^{\ell-1}}{2^\ell} - 1.
\end{equation*}
But as this is an integer, we must have
\begin{equation*}
    P[i] + Q[j] = \floor*{\frac{A[i] + B[j] - (x + y)2^{\ell-1}}{2^\ell}}.
\end{equation*}
Hence, since $x,y,\ell$ are fixed, it is minimized for any witness $(A * B)[i + j]$. In particular, $i^*$ is a witness of $(P * Q)[k^*]$, and 
\[ P[i^*]+Q[j^*] = \floor*{\frac{(A*B)[k^*] - (x + y)2^{\ell-1}}{2^\ell}} \le \floor*{\frac{(A*B)[k^*]}{2^\ell}},\]
proving~\ref{lem:doubling-reduct:clm:scaling:itm:one-witness}. Moreover, for any witness $i$ of $(P*Q)[k^*]$, $A[i] + B[k^*-i]$ cannot differ from $(A * B)[k^*]$ by $2^\ell$ or more, proving~\ref{lem:doubling-reduct:clm:scaling:itm:all-pseudo-witnesses}. 
\end{proof}

\begin{claim} \label{lem:doubling-reduct:clm:offset}
There is some $z \in [-4, 4]$ such that, in any iteration corresponding to these choices of $\ell, x, y, z$, with probability at least $0.9$ it holds that $P[i^*] + Q[j^*] = R[k^*]$.
\end{claim}
\begin{proof}
We derive lower and upper bounds on
\begin{equation*}
    \min_{\substack{i, j \in X\\i+j=k^*}} \floor*{\frac{A[i]}{2^\ell}} + \floor*{\frac{B[j]}{2^\ell}}.
\end{equation*}
On the one hand, a simple lower bound is $\floor{(A * B)[k^*] / 2^\ell} - 2$. On the other hand, $X$ contains some $2^{\ell+1}$-pseudo-witness $i'$ of $(A * B)[k^*]$ with probability at least $1 - (1 - p^2)^{t} \geq 1 - \exp(-p^2 t) = 
1-\exp(-n^{-2\rho}\cdot 100\lceil n^{2\rho}\rceil)
\geq 0.9$. Conditioned on this event, we have
\begin{align*}
    \min_{\substack{i, j \in X\\i+j=k^*}} \floor*{\frac{A[i]}{2^\ell}} + \floor*{\frac{B[j]}{2^\ell}}
    &\leq \floor*{\frac{A[i']}{2^\ell}} + \floor*{\frac{B[k^* - i']}{2^\ell}} \\
    &\leq \frac{A[i'] + B[k^* - i']}{2^\ell} \\
    &< \frac{(A * B)[k^*] + 2^{\ell+1}}{2^\ell} \\
    &< \floor*{\frac{(A * B)[k^*]}{2^\ell}} + 3.
\end{align*}
By definition of $P,Q$ and Property~\ref{lem:doubling-reduct:clm:scaling:itm:one-witness} of \cref{lem:doubling-reduct:clm:scaling}, it holds that
\begin{equation*}
    \floor*{\frac{(A * B)[k^*]}{2^\ell}} - 2 \leq P[i^*] + Q[j^*] \leq \floor*{\frac{(A * B)[k^*]}{2^\ell}}.
\end{equation*}
Putting all these bounds together, there must be some $z \in [-4, 4]$ such that
\begin{equation*}
    R[k^*] = \min_{\substack{i, j \in X\\i+j=k^*}} \floor*{\frac{A[i]}{2^\ell}} + \floor*{\frac{B[j]}{2^\ell}} + z = P[i^*] + Q[j^*],
\end{equation*}
completing the proof of the claim.
\end{proof}

In light of \cref{lem:doubling-reduct:clm:offset}, with high probability over the $100 \log n$ trials there is some iteration~$\ell, x, y, z$ in which we assign $P[i^*] + Q[j^*] = R[k^*]$. We argue that in this iteration we correctly assign $C[k^*]$. In step~\ref{lem:doubling-reduct:itm:ord:itm:3sum} we construct the 3SUM instance $(\mathcal P, \mathcal Q, \mathcal R)$, and by definition it contains the solution $(i^*, P[i^*]) + (j^*, Q[j^*]) = (k^*, R[k^*])$. \cref{lem:3sum-low-rank-to-low-doubling-contained} guarantees that this solution is either listed in $\mathcal T$, in which case we correctly assign $C[k^*]$ in step~\ref{lem:doubling-reduct:itm:ord:itm:brute-force}, or is part of one of the subinstances $(\mathcal P', \mathcal Q', \mathcal R')$. In the latter case, focus on the corresponding iteration of step~\ref{lem:doubling-reduct:itm:ord:itm:reduct}. By construction, we have $P'[i^*] = P[i^*]$ and $Q'[j^*] = Q[j^*]$. Thus, $(P' * Q')[k^*] = (P * Q)[k^*]$, recalling that $i^*$ is a witness of $(P * Q)[k^*]$ by Property~\ref{lem:doubling-reduct:clm:scaling:itm:one-witness} of \cref{lem:doubling-reduct:clm:scaling}. Moreover, Property~\ref{lem:doubling-reduct:clm:scaling:itm:all-pseudo-witnesses} of \cref{lem:doubling-reduct:clm:scaling} guarantees that every witness $i$ of $(P' * Q')[k^*]$ must be a $2^\ell$-pseudo-witness of $(A * B)[k^*]$. However, since $k^*$ is relevant at level $\ell$, there are at most $t$ such pseudo-witnesses, and consequently, $(P' * Q')[k^*]$ has at most $t$ witnesses. We are therefore guaranteed to find $i^*$ by the listing algorithm, and thus update $C[k^*]$ as intended. This completes the correctness analysis.

Finally, it remains to analyze the running time. Step~\ref{lem:doubling-reduct:itm:exc-pop} runs in time $\tO(n^{2-2\rho})$ and step~\ref{lem:doubling-reduct:itm:exc-unpop} runs in time $\tO(n^{1+2\rho})$. Step~\ref{lem:doubling-reduct:itm:ord} runs in $O(\log u)$ iterations. In each iteration, step~\ref{lem:doubling-reduct:itm:ord:itm:construct} runs in time $O(n + n^{2-2\rho})$, step~\ref{lem:doubling-reduct:itm:ord:itm:3sum} is dominated by the call to \cref{lem:3sum-low-rank-to-low-doubling-contained} running in time $O(n^{2-\delta})$ for some constant $\delta = \delta(\rho, \epsilon) > 0$, step~\ref{lem:doubling-reduct:itm:ord:itm:brute-force} can be charged to the running time of that lemma as well, and step~\ref{lem:doubling-reduct:itm:ord:itm:reduct} constructs $O(n^\epsilon)$ Min-Plus Convolution Listing instances, passes them to the reduction, and processes their $O(n^{1+2\rho})$ reported witnesses per instance. The total time is $O((n^{1+2\rho+\epsilon} + n^{2-2\rho} + n^{2-\delta}) \log u)$. This is as claimed, by possibly reducing $\delta$ to account for the second term.
\end{proof}

\subsection{Universe Reduction for Monotone Min-Plus Convolution} \label{sec:univ-reduct:sec:mono}
Recall that a sequence $A \in \Int^{[0, n]}$ is \emph{monotone} if $A[0] \leq A[1] \leq \dots \leq A[n]$. We refer to the Min-Plus Convolution problem with \emph{both} input sequences $A,B$ being monotone as the \emph{Monotone} Min-Plus Convolution problem.  
We emphasize that this convention applies to \cref{sec:univ-reduct} only; in other parts of this paper (notably, 
\cref{subsec:bdmono}) we often consider the variant where only one of the two input sequences is monotone.

It is easy to see that the general Min-Plus Convolution problem is trivially equivalent to Monotone Min-Plus Convolution: On the one hand, any monotone instance can trivially be regarded as non-monotone. And on the other hand, any non-monotone instance~$(A, B)$ over a universe $[u]$ can be converted to the equivalent monotone instance $(A', B')$ for
\begin{align*}
    A'[i] &= i \cdot 10u + A[i], \\
    B'[j] &= j \cdot 10u + B[j].
\end{align*}
However, this transformation increases the universe size from $u$ to (roughly) $u n$. In this section, we show that a converse transformation exists, from monotone to non-monotone, that reduces the universe size by a factor $n$. In particular, this shows that Min-Plus Convolution over the universe~$[u]$ is \emph{equivalent} to Monotone Min-Plus Convolution over the universe $[n u]$.

\begin{lemma}[Universe Reduction for Monotone Min-Plus Convolution] \label{lem:equiv-mono}
If Min-Plus Convolution over the universe $[u]$ can be solved in time $O(n^{2-\epsilon})$ for some constant $\epsilon > 0$, then Monotone Min-Plus Convolution of the universe $[u n]$ can be solved in time $O(n^{2-\epsilon'})$ for some constant $\epsilon' > 0$. Moreover, the resulting algorithm can provide a witness for each output entry of the Monotone Min-Plus Convolution instance.
\end{lemma}

A major part of the proof of \cref{lem:equiv-mono} follows the ideas of Chan and Lewenstein~\cite{ChanL15}: We impose a coarse-grained grid on the 2-dimensional index set of the given sequences, inducing a partition into cells. For each cell we select some representative point, apply the BSG theorem on these representatives, and may therefore assume that the representatives form a low-doubling set. From then on, our reduction differs, relying on additive hashing rather than sparse convolutions.

Specifically, it is convenient to use on the following covering version of the BSG theorem with an algorithmic proof due to Chan and Lewenstein~\cite{ChanL15}:

\begin{theorem}[BSG Covering~\cite{ChanL15}] \label{thm:bsg-covering}
Let $0<\alpha < 1$, let $t \geq 1$, let $A, B$ be sets of size at most $n$, and let $C$ be a set of size at most $t n$. There are subsets $A_1, \dots, A_m \subseteq A$ and $B_1, \dots, B_m \subseteq B$ and a remainder set $R \subseteq A \times B$ such that:
\begin{enumerate}[label=(\roman*)]
    \item\label{thm:bsg-covering:itm:covering} $\set{(a, b) \in A \times B : a + b \in C} \subseteq R \cup \bigcup_{i=1}^m (A_i \times B_i)$,
    \item\label{thm:bsg-covering:itm:doubling} $m = O(\alpha^{-1})$, and for all $1 \leq i \leq m$, $|A_i + B_i| \leq O(\alpha^{-5} t^3 n)$, and
    \item\label{thm:bsg-covering:itm:remainder} $|R| \leq O(\alpha n^2)$.
\end{enumerate}
Moreover, the sets $A_1, \dots, A_m, B_1, \dots, B_m$ and $R$ can be computed in time $\tO(n^2)$ by a randomized algorithm.
\end{theorem}

\begin{proof}[Proof of \cref{lem:equiv-mono}]
Let $(A, B)$ denote the given Monotone Min-Plus Convolution instance over the universe~$[u n]$ (where both $A,B$ are monotone), and rewrite these two sequences as
\begin{align*}
    A[i] &= A_1[i] \cdot u + A_0[i], \\
    B[j] &= B_1[j] \cdot u + B_0[j],
\end{align*}
where $A_0, B_0$ have entries at most $u$ and $A_1, B_1$ have entries at most $n$. In fact, we may assume without loss of generality that $A_0[i], B_0[j] < u / 2$ (by splitting both sequences into two subsequences, splitting the entries that are less than $u/2$ from the other entries, and computing their four pairwise min-plus convolutions).

Toward computing $A * B$ we maintain a sequence $C$, initialized with $\bot$. Over the course of the algorithm we will modify $C$ by updates of the form $C[k] \gets \min\set{C[k], A[i] + B[j]}$ (so that we never underestimate an output), claiming that, when the algorithm terminates, we have $A * B = C$. This also means that our algorithm can provide a witness $i$ for each entry $C[k]$ in the final output, as desired.

Let $\delta > 0$ be a constant to be determined, let $\Delta_X = \ceil{n^\delta}$ and let $\Delta_Y = \ceil{n^{2\delta}}$. Let $X$ denote the multiples of $\Delta_X$ in the range $[0, n]$, and let $Y$ denote the multiples of $\Delta_Y$ in the range $[0, n]$. We will think of $X \times Y$ as a grid in the plane $[0, n]^2$. The multiples of $\Delta_X$ (i.e., the vertical axes) are called \emph{anchors}, and we say that an integer is \emph{anchored} at the largest anchor smaller than that index. For an anchor $x \in X$, we define
\begin{alignat*}{3}
    L_A[x] &= \min_{x \leq i < x + \Delta_X} A[i], \qquad && L_B[x] &&= \min_{x \leq j < x + \Delta_X} B[j], \\
    U_A[x] &= \max_{x \leq i < x + \Delta_X} A[i], \qquad && U_B[x] &&= \max_{x \leq j < x + \Delta_X} B[j];
\end{alignat*}
i.e., the smallest and largest entries in the respective segments of $A$ and $B$ anchored at $x$. We call an anchor $x$ \emph{$A$-light} if $U_A[x] - L_A[x] < \Delta_Y$ and \emph{$A$-heavy} if $U_A[x] - L_A[x] \geq \Delta_Y$. Similarly define \emph{$B$-light} and \emph{$B$-heavy}. We say that an index $k$ is \emph{heavy} if it has a witness, $A[i] + B[j] = (A * B)[k]$, where $i$ is anchored at an $A$-heavy anchor or $j$ is anchored at a $B$-heavy anchor. Otherwise, we call~$k$ \emph{light}. We will now deal with the heavy and light indices $k$ individually.

\paragraph{Heavy Anchors}
We deal with the heavy indices by brute-force: For all $i \in [0, n]$ anchored at an $A$-heavy anchor, and all $j \in [0, n]$, update $C[i + j] \gets \min\set{C[i + j], A[i] + B[j]}$. Analogously, enumerate all $j \in [0, n]$ anchored at a $B$-heavy anchor and all $i \in [0, n]$ to perform the same update. This clearly correctly assigns $C[k]$ for all heavy $k$.

To bound the running time of this step, recall that $A$ and $B$ (and thus also $A_1$ and $B_1$) are monotone. Each $A$-heavy anchor witnesses a jump from $A[x]$ to $A[x + \Delta_X]$ of at least $\Delta_Y$. Thus, there can be at most $O(n / \Delta_Y)$ $A$-heavy anchors. But each anchor covers only $\Delta_X$ indices, and therefore the running time is $O(n/\Delta_Y \cdot \Delta_X \cdot n) = O(n^{2-\delta})$.

\paragraph{Light Anchors}
To deal with the light indices we rely on the BSG theorem. Let
\begin{align*}
    \mathcal A &= \set{(x, L_A[x]) : \text{$x \in X$ is $A$-light}}, \\
    \mathcal B &= \set{(y, L_B[y]) : \text{$y \in X$ is $B$-light}}.
\intertext{Moreover, compute the following sequence by brute-force, }
L_C[z] & = \min\{L_A[x]+L_B[z-x]: x\in X \text{ is $A$-light and $z-x\in X$ is $B$-light}\}\quad (z\in 2X),
\intertext{and let}
    \mathcal C &= \set{(z, L_C[z] + h) : z \in 2X,\, 0 \leq h < 2\Delta_Y}.
\end{align*}
We apply the BSG Covering Theorem (\cref{thm:bsg-covering}) on $\mathcal A, \mathcal B, \mathcal C$ with parameters $t = 4\Delta_Y$, $\alpha = n^{-\delta}$ and size bound $|X|$,  yielding subsets $\mathcal A_1, \dots, \mathcal A_m \subseteq \mathcal A$ and $\mathcal B_1, \dots, \mathcal B_m \subseteq \mathcal B$ and a remainder set $\mathcal R \subseteq \mathcal A \times \mathcal B$. We enumerate all remainder pairs $((x, L_A[x]), (y, L_B[y])) \in \mathcal R$. For each such pair, list all indices $i$ anchored at $x$, and all indices $j$ anchored at $y$, and update $C[i + j] \gets \min\set{C[i + j], A[i] + B[j]}$.

Then enumerate all pairs $\mathcal A_\ell, \mathcal B_\ell$ for $1 \leq \ell \leq m$. Let $X_\ell$ be the subset of anchors $x$ appearing in~$\mathcal A_\ell$, and similarly let $Y_\ell$ be the subset of anchors $y$ appearing in $\mathcal B_\ell$. Repeat the following steps $100 \log n$ times: Let $\hat p = \Theta(n \alpha^{-5} t^3 \Delta_Y \log u)$ (for a sufficiently large constant), and let $p$ be a random prime in the range $[\hat p, 2\hat p]$. Then construct the cyclic sequences $P, Q \in ([0, u] \cup \set{\bot})^{\Int / p \Int}$ by
\begin{align*}
    P[\hat i] &=
    \begin{cases}
        \mathrlap{A_0[i]}\phantom{B_0[j]} &\text{for some arbitrary $i \in [0, n]$ anchored in $X_\ell$ with $\hat i \equiv i \cdot 2n + A_1[i] \mod p$,} \\
        \bot &\text{if such $i$ does not exist,}
    \end{cases} \\
    Q[\hat j] &=
    \begin{cases}
        B_0[j] &\text{for some arbitrary $j \in [0, n]$ anchored in $Y_\ell$ with $\hat j \equiv j \cdot 2n + B_1[j] \mod p$,} \\
        \bot &\text{if such $j$ does not exist.}
    \end{cases}
\end{align*}
We compute the cyclic min-plus convolution $P * Q$; it is well-known that this reduces to $O(1)$ non-cyclic instances. Moreover, we list some witness of $\hat i$ for each output entry $(P * Q)[\hat k]$ by \cref{lem:min-plus-conv-listing}. For each such pair, letting $\hat j = \hat k - \hat i$, we determine the indices $i$ and $j$ corresponding to $\hat i$ and $\hat j$ that were used in defining $P[\hat i]$ and $Q[\hat j]$ (via $\hat i \equiv i \cdot 2n + A_1[i]$, and analogously for $j$) and update $C[i + j] \gets \min\set{C[i + j], A[i] + B[j]}$.

\paragraph{Correctness}
Let us analyze the correctness of the light step. Our goal is to show that, for each light index $k$, we eventually correctly assign $C[k] = (A * B)[k]$. Fix any such $k$, let $i$ be a witness of $(A * B)[k]$, and let $j = k - i$. Let $x$ be the anchor of $i$, and let $y$ be the anchor of $j$. Since $k$ is light, we know that~$x$ is \emph{$A$-light} and that $y$ is \emph{$B$-light}. Thus, $(x, L_A[x]) \in \mathcal A$ and $(y, L_B[y]) \in \mathcal B$. Let $z = x + y$. We claim that
\begin{equation} \label{lem:equiv-mono:eq:y}
    L_A[x] + L_B[y] < L_C[z] + 2\Delta_Y.
\end{equation}
Suppose otherwise, and let $x' + y' = z$ be such that $L_A[x'] + L_B[y'] = L_C[z]$, $x'\in X$ is $A$-light, and $y'\in X$ is $B$-light. Let $i' = i - x + x'$ and note that $i'$ is anchored at $x'$. Similarly, let $j' = j - y + y'$ and note that $j'$ is anchored at $y'$. Thus,
\begin{align*}
    A[i'] + B[j']
    &\leq U_A[x'] + U_B[y'] \\
    &< L_A[x'] + L_B[y'] + 2\Delta_Y \\
    &= L_C[z] + 2\Delta_Y \\
    &\leq L_A[x] + L_B[y] \\
    &\leq A[i] + B[j],
\end{align*}
which, together with $i' + j' = k$, contradicts the assumption that $i$ is a witness of $(A * B)[k]$. Therefore, \eqref{lem:equiv-mono:eq:y} holds and by definition, $(z, L_A[x] + L_B[y]) \in \mathcal C$.

By Property~\ref{thm:bsg-covering:itm:covering} of \cref{thm:bsg-covering} there are two possible cases: the pair $((x, L_A[x]), (y, L_B[y]))$ appears in the remainder $\mathcal R$ or in some set $\mathcal A_\ell \times \mathcal B_\ell$. In the former case we directly enumerate the pair $(i, j)$ in the algorithm and set $C[k]$ as intended. So focus on the latter case. Focus on the $\ell$-th iteration, and let $p$ be a random prime as in the algorithm. Define
\begin{align*}
    \hat i &= (i \cdot 2n + A_1[i]) \bmod p, \\
    \hat j &= (j \cdot 2n + B_1[j]) \bmod p, \\
    \hat k &= (k \cdot 2n + A_1[i] + B_1[j]) \bmod p.
\end{align*}
We now argue that three success events happen simultaneously with good probability.

First, $P[\hat i] = A_0[i]$, unless there is a collision of the form $i \cdot 2n + A_1[i] \equiv i' \cdot 2n + A_1[i'] \mod p$. By the Prime Number Theorem, for any fixed $i'$ (distinct from $i$) this collision takes place only with probability $O(\hat p^{-1}\log u  )$. This is at most $1/4n$ by choosing the constant in $\hat p$ sufficiently large. Taking a union bound over all $n$ choices of $i'$, it follows that we avoid all collisions with probability at least~$3/4$. Second, we analogously have $Q[\hat j] = B_0[j]$ with probability at least~$3/4$.

The third success event is that the only indices $i', j'$, anchored in $X_\ell$ and $Y_\ell$, respectively, and satisfying
\begin{equation*}
    (i' + j') \cdot 2n + A_1[i'] + B_1[j'] \equiv \hat k \mod p,
\end{equation*}
are those for which $i' + j' = i + j = k$ and $A_1[i'] + B_1[j'] = A_1[i] + B_1[j]$. To analyze this event, note that $(i' + j') \cdot 2n + A_1[i'] + B_1[j']$ takes only $O(|\mathcal A_\ell + \mathcal B_\ell| \Delta_X \Delta_Y)$ distinct values. Each undesired value, i.e., each value other than $(i + j) \cdot 2n + A_1[i] + B_1[j]$, collides with $\hat k$ (modulo $p$) with probability at most $O(\hat p^{-1}\log u )$, and thus the error probability is
\begin{equation*}
    O\parens*{\frac{|\mathcal A_\ell + \mathcal B_\ell| \Delta_X \Delta_Y \log u}{\hat p}} = O\parens*{\frac{\alpha^{-5} t^3 |X| \Delta_X \Delta_Y \log u}{\hat p}} = O\parens*{\frac{ \alpha^{-5} t^3 n \Delta_Y \log u}{\hat p}},
\end{equation*}
by Property~\ref{thm:bsg-covering:itm:doubling} of \cref{thm:bsg-covering}. Choosing the constant in~$\hat p$ sufficiently large, this error probability is at most $1/4$, say.

All three success events happen simultaneously with probability at least $1/4$. In this case we have $(P * Q)[\hat k] \leq P[\hat i] + Q[\hat j] = A_0[i] + B_0[j]$ by the first two events, and
\begin{equation*}
    (P * Q)[\hat k] \geq \min\set{A_0[i'] + B_0[j'] : i' + j' = k,\, A_1[i'] + B_1[j'] = A_1[i] + B_1[j]}
\end{equation*}
by the third event. Put together, it follows that $(P * Q)[\hat k] = A_0[i] + B_0[j]$, and that any witness $\hat i$ corresponds to a witness $i'$ of $k$. Consequently, we correctly set $C[k]$. Over the $100 \log n$ iterations this event happens at least once with high probability.

\paragraph{Running Time}
We finally analyze the running time. Suppose that Min-Plus Convolution over the universe $[u]$ is in time $O(n^{2-\epsilon})$ for some $\epsilon > 0$.

The heavy step runs in time $O(n^{2-\delta})$. The light step first applies the BSG Covering theorem (\cref{thm:bsg-covering}) in time $\tO(|X|^2) = \tO(n^{2-2\delta})$. We then enumerate all pairs in the remainder $\mathcal R$, and for each such pair further enumerate $\Delta_X^2$ pairs of indices amounting to time $O(|\mathcal R| \Delta_X^2) = O(\alpha |X|^2 \Delta_X^2) = O(\alpha n^2) = O(n^{2-\delta})$, by Property~\ref{thm:bsg-covering:itm:remainder} of \cref{thm:bsg-covering}. We then enumerate all $m = O(\alpha^{-1}) = O(n^\delta)$ pairs of sets $\mathcal A_\ell, \mathcal B_\ell$, where the upper bound is due to Property~\ref{thm:bsg-covering:itm:doubling} of \cref{thm:bsg-covering}, and run $O(\log n)$ repetitions of the randomized construction. In each repetition we construct and solve a length-$\Theta(\hat p)$ instance of Min-Plus Convolution. Each such instance can be trivially split into $\ceil{\hat p / n}^2$ subinstances of length at most $n$. The total time to solve these instances is
\begin{equation*}
    \tO\parens*{m \cdot \frac{\hat p^2}{n^2} \cdot n^{2-\epsilon}} = \tO\parens*{\alpha^{-11} t^6 \Delta_Y^2 \cdot n^{2-\epsilon}} = \tO(n^{2+27\delta-\epsilon}).
\end{equation*}
Choosing $\delta = \epsilon/28 > 0$, the total running time becomes $\tO(n^{2-\delta})$. This is truly subquadratic as claimed.
\end{proof}

\subsection{Large-Universe Reduction} \label{sec:univ-reduct:sec:lg}
We are finally in the position to complete the large-universe reduction: 

\begin{lemma}[Large-Universe Reduction for Min-Plus Convolution] \label{lem:univ-reduct-lg}
Assume \cref{hypo:sum-preserving}. If Min-Plus Convolution over the universe $[n]$ can be solved in time $O(n^{2-\epsilon})$ for some constant $\epsilon > 0$, then Min-Plus Convolution over the universe $[u]$ can be solved in time $O(n^{2-\epsilon'} (\log u)^{O(1)})$ for some constant $\epsilon' > 0$.
\end{lemma}

In other words, \cref{lem:univ-reduct-lg} states that the Strong Min-Plus Convolution Hypothesis implies the Min-Plus Convolution Hypothesis assuming \cref{hypo:sum-preserving}. This proves \cref{thm:strong-min-plus-conv-equal}.

For the proof we will also rely on the following simple covering lemma, which essentially constitutes an instance of the classic greedy algorithm for Set Cover. For a proof see, e.g.,~\cite[Lemma~5.11]{Fischer26}.

\begin{lemma}[Greedy Covering] \label{lem:sumset-cover}
There is a deterministic algorithm that, given sets $X, Y \subseteq [-u, u]$, computes a set $S \subseteq [-2u, 2u]$ of size at most $|Y - X| / |Y| \cdot \ln |X|$ such that $X \subseteq Y + S$. It runs in time $|Y - X| \, |S| \cdot (\log u)^{O(1)}$.
\end{lemma}

\begin{proof}[Proof of \cref{lem:univ-reduct-lg}]
Let $\rho, \epsilon > 0$ be two constants to be determined. We apply \cref{lem:doubling-reduct} to reduce any given Min-Plus Convolution instance over the universe $[u]$ to $O(n^\epsilon)$ instances of Min-Plus Convolution Listing with listing threshold $t = O(n^{2\rho})$. Each such instance $(A', B')$ is equipped with a $(K = O(n^\epsilon), N = O(n^{2-2\rho}))$-container $H$ for the joint index set $\mathcal I(A') \cup \mathcal I(B')$. Focus on one such instance~$(A', B')$ for now.

Let $\phi : \Int^2 \to \Int$ be the map defined by $\phi(i, a) = i  \cdot 10M + a$, where $M$ is an upper bound on the largest entry in $H$. Recall that $H \subseteq [2n] \times \Int$, and that $H$ has doubling at most $K$. It follows easily that $\phi$, restricted to $H + H$, is an isomorphism, and thus $\phi(H)$ also has doubling at most $K$. Thus, we may apply \cref{hypo:sum-preserving} on the integer set~$\phi(H)$. This yields a subset $Y \subseteq \phi(H)$ of size at least
\begin{equation*}
    |Y| \geq \Omega\parens*{\frac{|H|}{K^c}},
\end{equation*}
and a sum-order-preserving hash function $h : Y \to \set{0, \dots, N}$, where $c$ is the constant from \cref{hypo:sum-preserving}. Next, apply \cref{lem:sumset-cover} to compute a set of shifts $S$ such that $\phi(H) \subseteq Y + S$ and such that
\begin{equation*}
    |S| \leq \frac{|Y - \phi(H)|}{|Y|} \cdot \ln |H| = O\parens*{K^c \cdot \frac{|\phi(H) - \phi(H)|}{|H|} \cdot \log |H|} = O\parens{K^{c+2} \log |H|},
\end{equation*}
where the last step is due to the Plünnecke--Ruzsa inequality (\cref{lem:plunnecke-ruzsa}).

Next, we apply \cref{lem:min-plus-conv-listing} to reduce the Min-Plus Convolution Listing instance $(A', B')$ with listing threshold $t = O(n^{2\rho})$ to $\tO(n^{2\rho})$ plain Min-Plus Convolution subinstances (i.e., for which we do not have to list multiple witnesses). In what follows, fix any such instance, denoted $(A, B)$; we argue that it reduces to few instances of Monotone Min-Plus Convolution over the universe $[n^2]$. By \cref{lem:equiv-mono}, any such instance further reduces to (non-Monotone) Min-Plus Convolution over the universe $[n]$. Let $C$ be the output sequence, initialized to $\bot$. We will only update $C$ by updates of the form $C[k] \gets \min\set{C[k], A[i] + B[k - i]}$, guaranteeing that no output is underestimated.

Enumerate all shifts $s, t \in S$. For each pair, construct the sequences $A_s, B_t \in ([0, N] \cup \set{\bot})^{[0, n]}$ defined by
    \begin{align*}
        A_s[i] &=
        \begin{cases}
            \mathrlap{h(\phi(i, A[i]) - s)}\phantom{h(\phi(j, B[j]) - t)} &\text{if $\phi(i, A[i]) \in Y + s$,} \\
            \bot &\text{otherwise,}
        \end{cases} \\
        B_t[j] &=
        \begin{cases}
            h(\phi(j, B[j]) - t) &\text{if $\phi(j, B[j]) \in Y + t$,} \\
            \bot &\text{otherwise.}
        \end{cases}
    \end{align*}
Note that both sequences are almost monotone: The sequence $\phi(i, A[i])$ is increasing for $i = 0, \dots, n$ (when restricted to the non-$\bot$ entries), and hence, so are the images of $\phi(i,A[i])-s \in Y$ under the order-preserving function $h$. To remove the $\bot$ entries simply replace both sequences by their suffix-minima, i.e., replace $A_s$ and~$B_t$ by
\begin{align*}
    A_s'[i] &= \min_{i' \geq i} A_s[i'], \\
    B_t'[j] &= \min_{j' \geq j} B_t[j'].
\end{align*}
We compute the Monotone Min-Plus Convolution $A_s' * B_t'$ using the reduction provided by \cref{lem:equiv-mono}. Moreover, by \cref{lem:equiv-mono}, we know one witness $i$ per output $(A_s' * B_t')[k]$; then update $C[k] \gets \min\set{C[k], A[i] + B[k - i]}$. 

\paragraph{Correctness}
We argue that once the algorithm stops, we have correctly assigned $C$. Fix any witness $i^*$ of $(A * B)[k^*]$, and let $j^* = k^* - i^*$. Since $H$ is a container of the index sets $\mathcal I(A)$ and $\mathcal I(B)$, we must have $(i^*, A[i^*]), (j^*, B[j^*]) \in H$.
Thus, by $\phi(H)\subseteq Y+S $, there exist $s, t \in S$ such that $\phi(i^*, A[i^*]) \in Y + s$ and $\phi(j^*, B[j^*]) \in Y + t$. Focus on the iteration associated with that pair~$s, t$.

We first claim that any witness $i$ of $(A_s * B_t)[k^*]$ must also be a witness of $(A * B)[k^*]$. Indeed, take any other pair $i + j = k^*$ and assume that $i$ is not a witness of $(A * B)[k^*]$, i.e.,
\begin{equation*}
    A[i^*] + B[j^*] < A[i] + B[j].
\end{equation*}
Then clearly also
\begin{equation*}
    \phi(i^*, A[i^*]) + \phi(j^*, B[j^*]) < \phi(i, A[i]) + \phi(j, B[j]),
\end{equation*}
and thus,
\begin{equation*}
    (\phi(i^*, A[i^*]) - s) + (\phi(j^*, B[j^*]) - t) < (\phi(i, A[i]) - s) + (\phi(j, B[j]) - t),
\end{equation*}
But $h$ is sum-order-preserving and therefore,
\begin{equation*}
    A_s[i^*] + B_t[j^*] < A_s[i] + B_t[j],
\end{equation*}
for all non-$\bot$ entries, so $i$ cannot be a witness of $(A_s*B_t)[k^*]$. This proves that any witness $i$ of $(A_s * B_t)[k^*]$ is also a witness of $(A * B)[k^*]$ (though the converse is not necessarily true; notably, $i^*$ is not necessarily a witness of $(A_s * B_t)[k^*]$).

We finally argue that any witness of $(A_s' * B_t')[k^*]$ is also a witness of $(A_s * B_t)[k^*]$. Let $i$ be a witness of $(A_s' * B_t')[k^*]$. Let $j = k^* - i. $ By definition, there exist $i' \geq i$ and $j' \geq j$ such that
\begin{equation*}
    (A_s' * B_t')[k^*] = A_s[i'] + B_t[j'].
\end{equation*}
Suppose that $i' + j' > k^*$. Then
\begin{equation*}
    \phi(i^*, A[i^*]) + \phi(j^*, B[j^*]) < \phi(i', A[i']) + \phi(j', B[j']),
\end{equation*}
and thus by the same reasoning as above,
\begin{equation*}
    A_s[i^*] + B_t[j^*] < A_s[i'] + B_t[j'].
\end{equation*}
But then we reach the contradiction that
\begin{equation*}
    (A_s' * B_t')[k^*] \leq A'_s[i^*] + B'_t[j^*]\leq A_s[i^*] + B_t[j^*] < A_s[i'] + B_t[j'] = (A_s' * B_t')[k^*],
\end{equation*}
and must have $i' + j' = k^*$ instead. Combined with $i' \geq i$, $j' \geq j$ and $i + j = k^*$, we conclude that $i = i'$ and $j = j'$. Recall that we initially picked $i$ as an arbitrary witness of $(A_s' * B_t')[k^*]$. We have thus argued that
\begin{equation*}
    (A_s' * B_t')[k^*] = \min_{i+j=k^*} (A_s[i] + B_t[j]),
\end{equation*}
and so the witnesses of $(A_s' * B_t')[k^*]$ and $(A_s * B_t)[k^*]$ are identical.

\paragraph{Running Time}
Initially, we call \cref{lem:doubling-reduct} to reduce from one given instance to $O(n^{\epsilon})$ listing instances, each provided with a low-doubling container. \cref{lem:doubling-reduct} runs in time $O(n^{2-\delta})$ for some constant $\delta = \delta(\rho, \epsilon) > 0$. For each instance we call the algorithm implied by \cref{hypo:sum-preserving} on the container set~$\phi(H)$. This set has size $N = O(n^{2-2\rho})$ and doubling $K = O(n^\epsilon)$, and so this takes time $O(n^{2-2\rho+O(\epsilon)})$. We then apply \cref{lem:sumset-cover} in time
\begin{equation*}
    |Y - \phi(H)| \, |S| (\log u)^{O(1)} = N K^{2c+2} (\log u)^{O(1)} = n^{2-2\rho+O(\epsilon)} (\log u)^{O(1)}.
\end{equation*}

Then we reduce each listing instance (with threshold $t = O(n^{2\rho})$) further to $\tO(n^{2\rho})$ plain Min-Plus Convolution subinstances by \cref{lem:min-plus-conv-listing}. This leaves us with $\tO(n^{2\rho+\epsilon})$ instances that the main part of the algorithm treats. For each pair of shifts $s, t \in S$ we construct the sequences $A_s, B_t$ and~$A_s', B_t'$ in linear time. Assuming the premise of \cref{lem:univ-reduct-lg}, by \cref{lem:equiv-mono} we can compute their convolution in time $O(n^{2-\epsilon_0})$ for some $\epsilon_0 > 0$. The total time is
\begin{equation*}
    \parens*{n^{2-\delta} + n^{2-2\rho+O(\epsilon)} + n^{2\rho+\epsilon+2-\epsilon_0}} (\log u)^{O(1)}.
\end{equation*}
This becomes subquadratic for, say, $\rho = \epsilon_0 / 8$ and by choosing $0 < \epsilon \leq \epsilon_0 / 2$ sufficiently small so that the $O(\epsilon)$ term becomes at most $\rho$.
\end{proof}
\section{Lower Bounds for Intermediate Convolution Problems}

In this section, we show our lower bounds for intermediate problems under the Strong Min-Plus Convolution Hypothesis (also under the standard Min-Plus Convolution Hypothesis if we assume \cref{hypo:sum-preserving}). 

We start with a useful lemma, which separately handles those heavy input entries. 

\begin{lemma}
\label{lem:no-too-frequent-elements}
   Let $(A, B)$ be a given Min-Plus Convolution Listing instance with listing threshold $n^\ell$ for some $\ell \ge 1/2$. Suppose $C = A * B$ is also given and let $\delta > 0$ be a given constant. Then in $\tO(n^{2-\delta/2})$ time, we can reduce it to a Min-Plus Convolution Listing instance $(A, B')$, where all entries in $B$ that appear more than $n^{\ell+\delta}$ times are replaced with $\perp$. 
\end{lemma}
\begin{proof}
Call a value $v$ in $B$ heavy if it appears more than $n^{\ell + \delta}$ times in $B$.
Note that there could be at most $O(n^{1-\ell - \delta})$ heavy value $v$. 
It suffices to list up to $n^\ell$ witnesses $(i, j)$ per output $C[k]$, among those witnesses where $B[j]$ is a heavy element. 

To do so, let $v$ be an arbitrary heavy element in $B$, and define $J_v = \{j \in [0, n]: B[j] = v\}$. For every distinct value $u$ in $A$, define $I_u = \{i \in [0, n]: A[i] = u\}$ and $K_{u+v} = \{k \in [0, 2n]: C[k] = u + v\}$. These sets can be precomputed in near-linear time. Clearly, a pair $(i, j) \in I_u \times J_v$ is a witness for $C[i+j]$ if $i+j \in K_{u+v}$. 

For a fixed pair of $u,v$, using FFT, we can verify whether each $k \in K_{u+v}$ has a witness from $I_u \times J_v$ in $\tO(n)$ total time for all $k$. Moreover, using standard sampling-based approaches (\cref{lem:listingboolean}), we can list up to $\lfloor \frac{|I_u|}{n^{1-\ell - \delta / 2}}\rfloor$ witnesses per $k \in K_{u+v}$ in $\tO(\lfloor \frac{|I_u|}{n^{1-\ell - \delta / 2}}\rfloor n + 1)$ time. The $+1$ term in  the previous running time is needed for the case $\lfloor \frac{|I_u|}{n^{1-\ell - \delta / 2}}\rfloor = 0$. Over all values of $v$ and $u$, the total cost for this case is thus (up to $\polylog(n)$ factor)
\[
\sum_u \sum_v \left( \floor*{ \frac{|I_u|}{n^{1-\ell - \delta / 2}}} n + 1 \right) \le n^{1-\ell -\delta}  \sum_u \left( \frac{|I_u|}{n^{1-\ell - \delta / 2}} \cdot n  + 1\right)\le O( n^{2-\delta / 2}). 
\]

Hence, it suffices to consider some $k \in K_{u+v}$ for  sub-instances $(I_u, J_v, K_{u+v})$ where $k$ has at least $\lfloor \frac{|I_u|}{n^{1-\ell - \delta / 2}}\rfloor$ witnesses (and we have not found up to $n^\ell$ witnesses in the input instance $(A, B)$ for $k$). In this case, we simply spend $O(|I_u|)$ time to list all the witnesses for $k$ in $(I_u, J_v, K_{u+v})$. Let $U_k$ be the set of $u$ where we need to perform this brute-force over $I_u$ for $k$. Note that for a fixed $k$, $C[k] = u + v$ is fixed; in addition, the number of heavy values $v$ is $O(n^{1-\ell -\delta})$, so the number of $u$ where $k \in K_{u + v}$ is also $O(n^{1-\ell -\delta})$. Hence $|U_k| \le O(n^{1-\ell -\delta})$. 

As we only need to output $n^\ell$ total witnesses for $k$, we have 
\[
\sum_{u \in U_k} \floor*{ \frac{|I_u|}{n^{1-\ell - \delta / 2}}} \le n^{\ell}, 
\]
which implies (recall we assume $\ell \ge 1/2$)
\[
\sum_{u \in U_k}  \frac{|I_u|}{n^{1-\ell - \delta / 2}} \le |U_k| + n^{\ell} \le n^{1-\ell-\delta} + n^\ell = O(n^\ell). 
\]
Hence, the cost for a particular value of $k$ is 
\[
O\left(\sum_{u \in U_k} |I_u| \right) \le O(n^\ell \cdot n^{1-\ell - \delta / 2}) = O(n^{1-\delta / 2}). 
\]
Summing over all possible $k$ gives an $O(n^{2-\delta/2})$ running time, as desired. 
\end{proof}

\subsection{Min-Max Convolution}
\thmminmaxconv*
\begin{proof}
Suppose Min-Max Convolution can be solved in $O(n^{3/2-\eps})$ time for some $\eps > 0$ by an algorithm $\cA$. In order to show this contradicts the Strong Min-Plus Convolution Hypothesis, by \cref{cor:small-universe-listing} with $\gamma = \frac{1}{2}$, it suffices to give a truly sub-quadratic time algorithm for Min-Plus Convolution instances $(A, B)$ where the elements are from $[\sqrt{n}]$ while we need to list up to $\sqrt{n}$ witnesses per output. 

Because the universe is $[\sqrt{n}]$, $C = A* B$ can be computed in $\tO(n^{1.5})$ time.  By \cref{lem:no-too-frequent-elements}, we can further assume every non-$\perp$ value in $B$ appears at most $n^{1/2 + \eps / 3}$ times.

For every $b \in [\sqrt{n}]$, we define a Min-Max Convolution instance $(A_b, B_b)$, where $A_b[i] = A[i]$ for every $i \in [0, n]$, and $B_b[j] = -\infty$ if $B[j] = b$ and $B_b[j] = \infty$ otherwise. Let $C_b = \MinMax(A_b, B_b)$. It is not difficult to see that, $(i, j)$ is a witness for $C[k]$ (i.e., $A[i] + B[j] = C[k]$ and $i+j = k$), if and only if for some $b$, $C_b[k] = C[k] - b$ and $(i, j)$ is a witness for $C_b[k]$ in the Min-Max Convolution instance defined for $b$.

Consider the following algorithm. We enumerate $b \in [\sqrt{n}]$, create $A_b$ and $B_b$, and compute $C_b = \MinMax(A_b, B_b)$ using $\cA$ in $O(n^{1.5 - \eps})$ time per $b$.  By standard sampling-based approaches (\cref{lem:min-plus-conv-listing}), we can further list $n^{2\eps/3}$ witnesses for every $C_b[k]$ in $\tO(n^{1.5 - \eps / 3})$ time using $\cA$ as a black-box. This step takes $\tO(\sqrt{n} \cdot n^{1.5 - \eps / 3}) = \tO(n^{2 - \eps / 3})$ time over all $b \in [\sqrt{n}]$.

What remains to consider are those $C_b[k]$ that have more than $n^{2 \eps / 3}$ witnesses and $C_b[k] = C[k] - b$. We call such a pair $(b, k)$ heavy. For a heavy pair, we enumerate all elements in $S_b := \{j \in [n]: B[j] = b\}$ (this can be done in $O(|S_b|)$ time by computing these sets in preprocessing), and check whether each $j \in S_b$ forms a witness for $C_b[k]$. Recall each value in $B$ appears at most $n^{1/2+\eps / 3}$ time, so the cost for each heavy pair is $O(n^{1/2 + \eps / 3})$. Since we only need to list up to $n^{1/2}$ witness per $C[k]$, we only need to resolve $n^{1/2-2\eps/3}$ heavy pairs for a particular value of $k$ before we find enough witnesses for $C[k]$. Hence, we only need to handle $O(n^{3/2-2\eps/3})$ heavy pairs in total, and thus the overall running time for this case is $O(n^{2-\eps / 3})$. 
\end{proof}

\subsection{Min-Equality Convolution}
\thmminequalconv*
\begin{proof}
Suppose Min-Equality Convolution can be solved in $O(n^{3/2-\eps})$ time for some $\eps > 0$ by an algorithm $\cA$. Similar as before, it suffices to give a truly sub-quadratic time algorithm for Min-Plus Convolution Listing instances $(A, B)$ where the elements are from $[\sqrt{n}]$ while we need to list up to $\sqrt{n}$ witnesses per output.

We can compute $C = A* B$ in $\tO(n^{1.5})$ time and we can further assume every non-$\perp$ value in $B$ appears at most $n^{1/2 + \eps / 3}$ time by \cref{lem:no-too-frequent-elements}.

For every $\Delta \in \{-\sqrt{n}, \ldots, \sqrt{n}\}$, we define a Min-Equality Convolution instance $(A_\Delta, B_\Delta)$, where
$A_\Delta[i] = A[i] + \Delta$ for every $i \in [0, n]$ and $B_\Delta = B$. Intuitively, we aim to use $\Delta$ to represent $B[j] - A[i]$. 
Let $C_\Delta = \MinEqual(A_\Delta, B_\Delta)$. 
The following claim shows how the witnesses in the Min-Equality Convolution instance $(A_\Delta, B_\Delta)$ are related to the witnesses in the Min-Plus Convolution instance $(A, B)$. 
\begin{claim}
\label{cl:min-equal-reduction-witness}
    Let $(i_0, j_0, k_0) \in [0, n]^3$ be a triple with $i_0 + j_0 = k_0$, and let $\Delta = B[j_0] - A[i_0]$. Then, $A[i_0] + B[j_0] = C[k_0]$ holds 
    if and only if both $2 C_\Delta[k_0] = C[k_0] + \Delta$ and $A_\Delta[i_0] = B_\Delta[j_0] = C_\Delta[k_0]$ hold. 
\end{claim}
\begin{proof}
    We first show the forward direction. By definition, 
    \begin{align*}
    2C_\Delta[k_0] & = 2\min_{\substack{i + j = k_0 \\ A_\Delta[i] = B_\Delta[j]}} B_\Delta[j] =  \min_{\substack{i + j = k_0\\ A[i] + \Delta = B[j]}} \{A[i] + \Delta + B[j]\}\\
    & \ge \min_{i + j = k_0}  \{A[i] + \Delta + B[j]\} = C[k_0] + \Delta. 
    \end{align*}
    Additionally, since $A_\Delta[i_0] = A[i_0] + \Delta = B[j_0] = B_\Delta[j_0]$, 
    \[
    2C_\Delta[k_0] \le 2 B_\Delta[j_0] = A[i_0]+\Delta + B[j_0] =  C[k_0] + \Delta. 
    \]
    Hence, we have $2C_\Delta[k_0] = C[k_0] + \Delta$. Moreover, we have $C_{\Delta}[k_0]=B_{\Delta}[j_0]$ since the above inequality is tight, so  $A_\Delta[i_0] = B_\Delta[j_0] = C_\Delta[k_0]$ holds as claimed.

    Now we consider the reverse direction. From the assumptions $2 C_\Delta[k_0] = C[k_0] + \Delta$ and $A_\Delta[i_0] = B_\Delta[j_0] = C_\Delta[k_0]$, we have 
    \[
    C[k_0] + \Delta  = 2 C_\Delta[k_0] = A_\Delta[i_0] + B_\Delta[j_0] = A[i_0] + \Delta + B[j_0], 
    \]
   so      $A[i_0] + B[j_0] = C[k_0]$ holds as claimed.
\end{proof}

Consider the following algorithm. We enumerate $\Delta \in \{-\sqrt{n}, \ldots, \sqrt{n}\}$, create $A_\Delta$ and $B_\Delta$, and compute $C_\Delta = \MinEqual(A_\Delta, B_\Delta)$ using $\cA$ in $O(n^{1.5 - \eps})$ time per $\Delta$.  By standard sampling-based approaches (\cref{lem:min-plus-conv-listing}), we can further list $n^{2\eps/3}$ witnesses for every $C_\Delta$ in $\tO(n^{1.5 - \eps / 3})$ time using $\cA$ as a black-box. This step takes $\tO(\sqrt{n} \cdot n^{1.5 - \eps / 3}) = \tO(n^{2 - \eps / 3})$ over all $\Delta$.

What remains to consider are those $C_\Delta[k]$ that have more than $n^{2 \eps / 3}$ witnesses and $2C_\Delta[k] = C[k] + \Delta$. We call such a pair $(\Delta, k)$ heavy. By \cref{cl:min-equal-reduction-witness}, the set of witnesses for $C_\Delta[k]$ is exactly the set of $(i, j)$ where $A[i] + B[j] = C[k]$ and $B[j] - A[i] = \Delta$. Since the values of $C[k]$ and $\Delta$ are both known, we can determine the value $A[i] = u$ and $B[j] = v$ where $u = (C[k] - \Delta)/2, v = (C[k] + \Delta) / 2$. After simple linear time preprocessing, we can enumerate elements in the set $S_v = \{j \in [n]: B[j] = v\}$ in $O(n^{1/2+\eps / 3})$ time (recall each element can appear at most $O(n^{1/2+\eps / 3})$ times in $B$). Once we have this set, we can compute the exact set of witnesses for $C_\Delta[k]$ in $O(n^{1/2+\eps / 3})$ time, simply by  checking if $A_\Delta[k - j] = B_\Delta[j]$ and compare them with $C_\Delta[k]$ for $j \in S_v$. Since we only need to list up to $n^{1/2}$ witness per $C[k]$, we only need to resolve $n^{1/2-2\eps/3}$ heavy pairs for a particular value of $k$ before we find enough witnesses for $C[k]$. Hence, we only need to handle $O(n^{3/2-2\eps/3})$ heavy pairs, and thus the overall running time for this case is $O(n^{2-\eps / 3})$. 
\end{proof}

\subsection{Bounded Monotone Min-Plus Convolution}
\label{subsec:bdmono}

In order to show the hardness for Bounded Monotone Min-Plus Convolution, we first show the hardness of another problem in the following:

\begin{lemma}
\label{lem:bounded-monotone-helper}
Let $A \in [\sqrt{n}]^{[0, n]}$ and $B \in [\sqrt{n}]^{[0, {\sqrt{n}}]}$. Computing the Min-Plus Convolution between $A$ and $B$ requires $n^{1.5 - o(1)}$ time, under the Strong Min-Plus Convolution Hypothesis. 
\end{lemma}
\begin{proof}
    Let $\cA$ be such an $O(n^{3/2-\eps})$ time algorithm. Similar as before, it suffices to give a truly sub-quadratic time algorithm for Min-Plus Convolution instances $(A, B)$ where the elements are from $[\sqrt{n}]$ while we need to list up to $\sqrt{n}$ witnesses per output.  

By trying all sampling rates $p = 1, 2^{-1}, 2^{-2}, \ldots$, and only keeping sampled elements in $B$ as their original values while setting all other elements as $\perp$, then with high probability,  there will be a round where the number of witnesses for $C[k]$ is  $\Theta(\sqrt{n})$ if the number of witnesses for $C[k]$ in the original instance is more than $\sqrt{n}$. Hence, we only need to handle $C[k]$ with $O(\sqrt{n})$ witnesses. By padding the sequences with $O(n)$ additional $\perp$'s, we can further assume that we only need to handle $C[k]$ with $\sqrt{n}$ witnesses.

Let $N$ be a prime number from $[2n, 4n]$, which can be found deterministically in $\tO(n)$ time. We extend $A$ and $B$ to length $N$, where the extra elements are padded as $\perp$. Now the Min-Plus Convolution $C$ between $A$ and $B$ can be equivalently defined by $C[k] = \min_{i \in \F_N} \{A[i] + B[k - i]\}$. 

We repeat the following $\tO(\sqrt{n})$ iterations: sample $x \sim \F_N \setminus \{0\}$ and $y \sim \F_N$. Create a length-$N$ sequence $A'$ where $A'[xi] = A[i]$ for every $i \in \F_N$; create a length-$\lfloor \sqrt{N}\rfloor$ sequence $B'$ where $B'[xj + y] = B[j]$ if $xj + y \in \{0, 1, \ldots, \lfloor\sqrt{N}\rfloor\}$. We can use $\cA$ to compute $C' = A' * B'$ in $O(n^{3/2-\eps})$ time. By the standard sampling-based approach (\cref{lem:min-plus-conv-listing} generalized to input sequences of different lengths), we can also list a witness for each $C'[k]$ in $\tO(n^{3/2-\eps})$ time, by using $\cA$ as a black-box. Observe that, if some $C[k]$ has witnesses $(i, j)$ where $A[i] + B[j] = C[k]$, and if $xj + y \in \{0, 1, \ldots, \lfloor \sqrt{N}\rfloor\}$, then $(xi, xj + y)$ will be a witness for $C'[xk + y]$. Given a witness for $C'[xk + y]$, we can convert it back to a witness for $C[k]$. By repeating $\tO(\sqrt{n})$ iterations, we can list multiple witnesses for $C[k]$, though there might be repeated witnesses from different iterations. It remains to analyze the success probability of collecting all witnesses for each $C[k]$.

Let $J_k$ be the set of $j \in \F_N$ where $j$ is a witness for $C[k]$ (i.e., $A[k - j] + B[j] = C[k]$). Recall we only need to list witnesses for $C[k]$ where $|J_k| \le \sqrt{n}$. Let $j_0 \in J_k$ be an arbitrary witness, and we focus on a particular iteration. Let $X_j$ denote the event that $B[j] = B'[xj + y]$ in this iteration, i.e., $xj + y \in \{0, 1, \ldots, \lfloor \sqrt{N}\rfloor \}$. If only $X_{j_0}$ happens, and $X_j$ does not happen for other $j \in J_k \setminus \{j_0\}$, then the algorithm is able to discover $j_0$ as a witness for $C[k]$ in this iteration. 
Using the near-independence between $X_{j'}$ and $X_j$ for $j \ne j'$ by the way we sample $x$ and $y$, this happens with probability 
\begin{align*}
    \Pr\left[X_{j_0} \wedge \bigwedge_{j \in J_k \setminus \{j_0\} } \neg X_j\right] &= \Pr\left[X_{j_0}\right] -  \Pr\left[X_{j_0} \wedge \bigvee_{j \in J_k \setminus \{j_0\} }  X_j\right]\\
    & \ge \Pr\left[X_{j_0}\right] - \sum_{j \in J_k \setminus \{j_0\}} \Pr\left[X_{j_0} \wedge X_j\right]\\
    & = \frac{\lfloor \sqrt{N}\rfloor + 1 }{N} - (|J_k|-1) \cdot \frac{\lfloor\sqrt{N}\rfloor + 1}{N}\cdot \frac{\lfloor\sqrt{N}\rfloor}{N-1} \\
    & \ge \frac{1}{\sqrt{N}} - \sqrt{n} \cdot \left( \frac{\sqrt{N}+1}{N}\right)^2 \\
    &= \Omega(N^{-1/2}),
\end{align*}
where the last step is because $n \le N / 2$. Hence, in $\tO(N^{1/2})$ iterations, we can list all witnesses for $C[k]$ with high probability. 
\end{proof}

\thmminplusconvmono*

\begin{proof}
By \cref{lem:bounded-monotone-helper}, it suffices to reduce the problem of computing the Min-Plus Convolution between $A \in [\sqrt{n}]^{[0, n]}$ and $B \in [\sqrt{n}]^{[0,\sqrt{n}]}$ to bounded-monotone Min-Plus Convolution. By padding, we can assume $\sqrt{n}$ is an integer for convenience.

We first carefully split $A$ to two halves $A^{\text{odd}}$ and $A^{\text{even}}$, where
\[A^{\text{odd}}[i] = \left\{
  \begin{array}{lr}
    A[i] & \text{if $\lfloor i / \sqrt{n}\rfloor$ is odd,}\\
    \infty & \text{if $\lfloor i / \sqrt{n}\rfloor$ is even,}
  \end{array}
\right. \quad \quad A^{\text{even}}[i] = \left\{
  \begin{array}{lr}
    \infty & \text{if $\lfloor i / \sqrt{n}\rfloor$ is odd,}\\
    A[i] & \text{if $\lfloor i / \sqrt{n}\rfloor$ is even.}
  \end{array}
\right.
\]
It suffices to separately compute $A^{\text{odd}} * B$ and $A^{\text{even}} * B$, and combine the results. As the methods are essentially the same for these two Min-Plus Convolutions, we focus on computing $C^{\text{odd}} = A^{\text{odd}} * B$. 

Create  sequences $A' \in \Z^{[0,n]}$ and $B'\in \Z^{[0,n]}$, where
\[A'[i] = \left\{
  \begin{array}{lr}
    A[i] + \sqrt{n} \cdot (i \bmod \sqrt{n}) & \text{if $\lfloor i / \sqrt{n}\rfloor$ is odd,}\\
    10n & \text{if $\lfloor i / \sqrt{n}\rfloor$ is even,}
  \end{array}
\right. \quad \quad B'[j] = \left\{
  \begin{array}{lr}
    B[j] + \sqrt{n} \cdot j & \text{if $j \le \sqrt{n}$,}\\
    10 n & \text{otherwise.}
  \end{array}
\right.
\]
By definition, all  entries of $A',B'$ are non-negative integers in $O(n)$.
Moreover, observe that $B'$ is monotone. Hence, we can use a bounded-monotone Min-Plus Convolution algorithm to compute $C' = A' * B'$. First, observe that $C^{\text{odd}}[k] = \infty$ if and only if $C'[k] \ge 10n$. Now consider the case where $C'[k] < 10n$. In this case,
\begin{align*}
    C'[k] &= \min_{\substack{0 \le j \le \min\{k, \sqrt{n}\} \\ \lfloor (k - j) / \sqrt{n}\rfloor  \text{ odd}}} \left\{ A'[k - j] + B'[j]\right\} \\
    &= \min_{\substack{0 \le j \le \min\{k, \sqrt{n}\} \\ \lfloor (k - j) / \sqrt{n} \rfloor \text{ odd}}} \left\{ A[k - j] + B[j] + ((k - j) \bmod \sqrt{n} + j) \cdot \sqrt{n}\right\}
\end{align*}
We can simplify the term $(k - j) \bmod \sqrt{n} + j$ via
\[
    (k - j) \bmod \sqrt{n} + j = k - j - \lfloor (k - j) / \sqrt{n}\rfloor \cdot \sqrt{n} + j = k - \lfloor (k - j) / \sqrt{n}\rfloor \cdot \sqrt{n}. 
\]
Because $0 \le j \le \sqrt{n}$, $\lfloor (k - j) / \sqrt{n}\rfloor \in [\lfloor (k - \sqrt{n}) / \sqrt{n}\rfloor, \lfloor k / \sqrt{n}\rfloor]$. Notice that there are exactly two integer values in this range. Combined with the condition that $\lfloor (k - j) / \sqrt{n}\rfloor$ is odd, there can actually be only one valid value, which can be computed in $O(1)$ time. Once $\lfloor (k - j) / \sqrt{n}\rfloor$ is determined, we can compute $k - \lfloor (k - j) / \sqrt{n}\rfloor \cdot \sqrt{n}$, and we denote this value by $f_{\text{odd}}(k)$. Hence, we can further simplify the expression for $C'[k]$ as 
\begin{align*}
    C'[k] 
    & = \min_{\substack{0 \le j \le \min\{k, \sqrt{n}\} \\ \lfloor (k - j) / \sqrt{n} \rfloor \text{ odd}}} \left\{ A[k - j] + B[j] + f_{\text{odd}}(k) \cdot \sqrt{n}\right\}\\
    & = \left(f_{\text{odd}}(k) \cdot \sqrt{n} \right) + \min_{\substack{0 \le j \le \min\{k, \sqrt{n}\} \\ \lfloor (k - j) / \sqrt{n}\rfloor \text{ odd}}} \left\{ A[k - j] + B[j]\right\} \\
    &= \left(f_{\text{odd}}(k) \cdot \sqrt{n} \right) + C^{\text{odd}}[k].
\end{align*}

Hence, we can compute $C^{\text{odd}}$ from $C'$ in linear time, which completes the proof. 
\end{proof}

\subsection{Min-Witness Convolution}

The following lemma is the key step in the reduction to Min-Witness Convolution. 

\begin{lemma}
\label{lem:min-witness-helper}
    Suppose Min-Witness Convolution can be solved in $O(T(n))$ time, and let $1/3 \le \ell \le 1$ be a constant. Then in $\tO(T(n) \cdot n^\ell)$ time we can list up to $n^{\ell}$ witnesses for each output in Min-Plus Convolution instances $(A, B)$ where the universe of all elements in $A$ has size $n^{1/3}$, and the universe of all elements in $B$ is $\{0, \perp\}$. The algorithm succeeds with high probability. 
\end{lemma}
\begin{proof}
First, since $B[j]\in \{0,\perp\}$, we can replace each element in $A$ by its rank in the universe without loss of generality, and hence assume that the universe of elements in $A$ is $[n^{1/3}]$. 
Let $C = A * B$, which can be computed in $\tO(n^{4/3})$ time. 

By trying all sampling rates $p = 1, 2^{-1}, 2^{-2}, \ldots$, and only keeping sampled elements in $B$ as their original values  while setting all other elements as $\perp$, then with high probability,  there will be a round where the number of witnesses for $C[k]$ is  $\Theta(n^\ell)$ if the number of witnesses for $C[k]$ in the original instance is more than $n^\ell$. Hence, we only need to handle $C[k]$ with $O(n^\ell)$ witnesses. 

In addition, it suffices to list $n^{1/3}$ witnesses per $C[k]$ in $\tO(T(n) n^{1/3})$ time: If $\ell > 1/3$, we can randomly partition $B$ to $\Theta(n^{\ell - 1/3})$ random subsequences (by keeping elements not in the corresponding subsequence as $\perp$), so that for every $C[k]$, each subsequence contains at most $n^{1/3}$ witnesses, with high probability. 

Before we give the actual reduction, let us first give some intuition. The difficulty of reducing Min-Plus Convolution to Min-Witness Convolution is that, in Min-Witness Convolution, there is no value in the inputs. Hence, we somehow need to encode the weight information in the indices. One potential way is to take monotone subsequences of $A$, and create Boolean sequence $A'$ where $A'[i] = 1$ if $i$ is in the chosen subsequence and create Boolean sequence $B'$ where $B'[i] = 1$ if $B[i] = 0$. Because the chosen subsequence is monotone, the Min-Witness Convolution naturally favors candidates $(i, k - i)$ where $A[i]$ is smaller, as intended. One potential way to choose monotone subsequences is the following: We take all the $1$'s in the first $n^{2/3}$ entries of $A$, all the $2$'s in the second $n^{2/3}$ entries of $A$, and so on. The downside is that, some $C[k]$ might have many witnesses from the subsequence, while some have few. To overcome this issue, we randomly shuffle the indices and work over the shuffled sequences.  

Let $N$ be a prime number from $[2n, 4n]$, which can be found deterministically in $\tO(n)$ time. We extend $A$ and $B$ to length $N$, where the extra elements are $\perp$. Now Min-Plus Convolution $C$ between $A$ and $B$ can be equivalently defined by $C[k] = \min_{i \in \F_N} \{A[i] + B[k - i]\}$. 

We will repeat the following $\tO(n^{1/3})$ times. Let $x \sim \F_N \setminus \{0\}$ and $y \sim \F_N$ be uniformly random, and let $f_{x, y}: \F_N \rightarrow \F_N$ where $f_{x, y}(i) = i \cdot x + y$. Now define the shuffled sequences $A_{x, y}$ and $B_{x, y}$, where
\[A_{x, y}[f_{x,y}(i)] = A[i] \quad \text{and} \quad B_{x, y}[f_{x, y}(i)] = B[i] \quad \text{for }  i \in \F_N. \]
$A_{x, y}$ and $B_{x, y}$ are well-defined because $f_{x, y}$ is invertible. 
Consider two subsets $S_{x, y} = \{i \in \F_N: A_{x, y}[i] = \lfloor i / N^{2/3}\rfloor\}$ and $T_{x, y} = \{i \in \F_N: B_{x, y}[i] = 0\}$. We then create two Boolean sequences $A'_{x, y}$ and $B'_{x, y}$, which are the indicator sequences for $S_{x, y}$ and $T_{x, y}$ respectively. 

The following claim shows why it is useful to compute the Min-Witness Convolution between $A'_{x, y}$ and $B'_{x, y}$ (even though the indices of them are from $\F_N$ rather than $\Z$, we can still use standard Min-Witness Convolution with some simple postprocessing to compute it). 
\begin{claim}
\label{cl:min-witness-claim1}
Let $I_k$ be the set of witnesses for $C[k]$ in the Min-Plus convolution $C = A * B$. If there is exactly one witness $i \in I_k$ such that $f_{x, y}(i) \in S_{x, y}$, then $\MinWitness(A'_{x, y}, B'_{x, y})[x \cdot k + 2y] = f_{x, y}(i)$.  
\end{claim}
\begin{proof}
    Let $j = k - i$. Since $i$ is a witness for $C[k]$, we must have $B[j]\neq \bot $ and hence $B[j]=0$, and $A[i]=C[k]$. By definition, clearly $A'_{x, y}[f_{x, y}(i)] = 1$, and $B'_{x, y}[(x \cdot k + 2y) - f_{x, y}(i)] = B'_{x, y}[f_{x, y}(j)] = 1$. Hence, $f_{x, y}(i)$ is a witness in the Boolean convolution between $A'_{x, y}$ and $B'_{x, y}$. It remains to show that it is the minimum such index. 

    Suppose for the sake of contradiction that there exists $t < f_{x, y}(i)$ where $A'_{x, y}[t] = B'_{x, y}[(x \cdot k + 2y) - t] = 1$. By definition, these imply
    \[
    A_{x, y}[t] = \lfloor t / N^{2/3}\rfloor \quad \text{and} \quad B_{x, y}[(x \cdot k + 2y) - t] = 0.
    \]
    Because $f_{x, y}$ is invertible, there exists $i' \in \F_N$ such that $f_{x, y}(i') = t$, and correspondingly $f_{x, y}(k - i') = (x \cdot k + 2y) - t$. 
    Hence, we further obtain
    \begin{equation}
    \label{eq:min-witness-eq1}
        A[i'] = \lfloor t / N^{2/3}\rfloor \quad \text{and} \quad B[k - i'] = 0.
    \end{equation}
    
    By the fact $f_{x, y}(i) \in S_{x, y}$, we have $A_{x,y}[f_{x,y}(i)] = \lfloor f_{x,y}(i) / N^{2/3}\rfloor \ge \lfloor t / N^{2/3}\rfloor = A[i']$. Since $A[i] = A_{x,y}[f_{x,y}(i)]$, we have $A[i']\le A[i] =C[k]$.
    By  \cref{eq:min-witness-eq1}, this means $i'\in I_k$ and $A[i']=A[i]=C[k]$.
    Since $f_{x,y}(i')\in S_{x,y}$ holds,
    this contradicts the fact that $i$ is the unique element in $I_k$ where $f_{x, y}(i) \in S_{x, y}$, as $i'$ is also such an element. 
\end{proof}
By \cref{cl:min-witness-claim1}, as long as we can guarantee that there is a unique witness $i \in I_k$ satisfying $f_{x, y}(i) \in S_{x, y}$, we are able to recover it from the result of $\MinWitness(A'_{x, y}, B'_{x, y})$. The next \cref{claim:min-witness-i0} shows this happens with probability $\Omega(1/N^{1/3})$. 
Since we initially assumed $|I_k|\le n^{1/3}$, it thus takes only $\tO(N^{1/3})$ repetitions to recover all elements in $I_k$, in total $\tO(n^{1/3}T(n))$ time as desired. It thus remains to show \cref{claim:min-witness-i0} to finish the proof. 
\begin{claim}
\label{claim:min-witness-i0}
    Let $I_k$ be the set of witnesses for $C[k]$ in the Min-Plus Convolution $C = A * B$ where $|I_k| \le n^{1/3}$, and let $i_0 \in I_k$ be an arbitrary element. Then the probability that $i_0$ is the unique element in $I_k$ that satisfies $f_{x, y}(i) \in S_{x, y}$ is $\Omega(1/N^{1/3})$. 
\end{claim}
\begin{proof}
Note that every $i \in I_k$ has the same value $A[i]$. Let us denote this value by $u$. Then, for $i\in I_k$, the condition $f_{x, y}(i) \in S_{x, y}$ can be equivalently written as 
\[
f_{x, y}(i) \in [u \cdot N^{2/3}, (u+1)\cdot N^{2/3}). 
\]
Let $X_i$ denote the above event.  
Note that $\{f_{x, y}\}_{x, y}$ is a nearly $2$-wise independent hash family. Hence, 
\begin{align*}
    \Pr\left[X_{i_0} \wedge \bigwedge_{i \in I_k \setminus \{i_0\} } \neg X_i\right] &= \Pr\left[X_{i_0}\right] -  \Pr\left[X_{i_0} \wedge \bigvee_{i \in I_k \setminus \{i_0\} }  X_i\right]\\
    & \ge \frac{N^{2/3} - 1}{N} - \sum_{i \in I_k \setminus \{i_0\}} \Pr\left[X_{i_0} \wedge X_i\right]\\
    & \ge \frac{N^{2/3} - 1}{N} - n^{1/3}  \cdot \left(\frac{N^{2/3} + 1}{N} \right)^{2}\\
    &= \Omega(N^{-1/3}),
\end{align*}
where the final step holds because $n \le N / 2$. 
\end{proof}

\end{proof}

\thmminwitnessconv*
\begin{proof}

Suppose Min-Witness Convolution can be solved in $O(n^{4/3-\eps})$ time for some $\eps > 0$ by an algorithm $\cA$. In order to show this contradicts the Strong Min-Plus Convolution Hypothesis, by \cref{cor:small-universe-listing} with $\gamma = \frac{2}{3}$, it suffices to give a truly sub-quadratic time algorithm for $\MinPlusConvListing(n, n^{2/3} \mid u \le n^{1/3})$. Let $(A, B)$ be such a Min-Plus Convolution Listing instance. 
By \cref{lem:no-too-frequent-elements}, we can further assume no elements in $B$ appear more than $n^{2/3 + \eps / 3}$ times. 
Let $C = A * B$, which can be computed in $\tO(n^{4/3})$ time. 

For every $b \in [n^{1/3}]$, we solve a sub-instance between $(A, B_b)$, where $B_b[j] = 0$ if $B[j] = b$, and $B_b[j] = \infty$ otherwise. Using \cref{lem:min-witness-helper} with $\ell = 1/3 + 2 \eps / 3$, we can list up to $n^{1/3+2 \eps/3}$ witnesses per output in $C_b = A * B_b$. The total cost for calling \cref{lem:min-witness-helper} over $b \in [n^{1/3}]$ is $O(n^{5/3+2\eps / 3 - \eps} \cdot n^{1/3}) = O(n^{2-\eps/3})$. 

If $C_b[k] = C[k] - b$, and \cref{lem:min-witness-helper} lists fewer than $n^{1/3 + 2\eps / 3}$ witnesses for $C_b[k]$, then we have listed all witnesses $(i, j)$ satisfying $B[j] = b$ for $C[k]$. Otherwise, if we have not listed $n^{2/3}$ witnesses for $C[k]$, we enumerate all non-$\infty$ indices in $B_b$ (there can be at most $n^{2/3 + \eps / 3}$ such indices), and check whether it forms a witness for $C[k]$. The total number of such enumerations is $O(n \cdot n^{1/3-2\eps / 3})$, as for each $C[k]$, one such enumeration increases the number of witnesses found for it by at least $n^{1/3+2\eps/3}$. Hence, the total cost is $O(n \cdot n^{1/3-2\eps/3} \cdot n^{2/3 + \eps / 3}) = O(n^{2-\eps / 3})$, which is truly subquadratic as desired. 
\end{proof}
\section{Asymmetric BSG Theorem} \label{sec:asym-bsg}

In this section, we show an algorithmic version of  the Asymmetric BSG Theorem (originally proved by Tao and Vu \cite[Section 2.6]{TaoV06}), stated as follows. 
\thmasymbsg*

\subsection{Preliminary Lemmas} \label{sec:asymbsg}
Our algorithmic proof of \cref{thm:asym-bsg} is heavily based on the original proof for the non-algorithmic version by Tao and Vu \cite[Section 2.6]{TaoV06}. We turn each step into an efficient algorithm using mostly basic ideas.  

We start with the following definition from \cite[Definition 2.32]{TaoV06}:

\begin{definition}[Symmetry Sets]
For $A \subseteq \Int$ and $0 < \alpha \leq 1$, we define the \emph{symmetry set}
\begin{equation*}
    \Sym(A; \alpha) = \set{ x \in A - A : r_{A-A}(x) \geq \alpha |A|}.
\end{equation*}
\end{definition}

\noindent
Note that $0 \in \Sym(A; \alpha)$ and $\Sym(A; \alpha)= -\Sym(A; \alpha)$. Moreover, it holds that
\begin{equation*}
    |\Sym(A;\alpha)| \leq \frac{|A|}{\alpha},
\end{equation*}
which follows readily from $\sum_{x \in A - A} r_{A-A}(x) = |A|^2$. In the notation from before one can equivalently express the symmetry set as a set of popular differences, $\Sym(A; \alpha) = \Pop(A, -A; \alpha |A|)$.

We need the following lemma corresponding to~\cite[Lemma~2.34]{TaoV06}.

\begin{lemma}[Dyadic Pigeonhole Principle] \label{lem:dyadic-pigeonhole}
Let $0 < \alpha < 1$ and let $A, D \subseteq \Int$ be such that
\begin{equation*}
    \sum_{x \in D} r_{A-A}(x) \geq \alpha |A|^2.
\end{equation*}
Then there is a subset $D' \subseteq D$ such that, for every $x \in D'$,
\begin{equation*}
    r_{A-A}(x) \geq \Omega\parens*{\frac{\alpha}{1+\log \frac{1}{\alpha} +\log(1 + \frac{|D|}{|A|})} \cdot \frac{|A|^2}{|D'|}}.
\end{equation*}
Moreover, there is a randomized algorithm that, given $A$, $\alpha$ and $D$, computes such a set $D'$ in time $\tO(|A| + |D| / \alpha)$.
\end{lemma}
\begin{proof}
We first call \cref{lem:approx-mult} on the sets $A, -A, D$ with parameter \smash{$\epsilon = \min(\frac{\alpha |A|}{8 |D|}, 1)$} to approximate the multiplicities $r_{A-A}(x)$ for all $x \in D$ by $0 \leq \tilde r(x) \leq |A|$ within additive error
\begin{equation*}
    |r_{A-A}(x) - \tilde r(x)| \leq \epsilon |A| = \frac{\alpha |A|^2}{8 |D|}.
\end{equation*}
Let
\begin{equation*}
    D^* = \set*{ x \in D : \tilde r(x) \geq \frac{\alpha |A|^2}{4 |D|}}.
\end{equation*}
This set is chosen so that for all $x \in D^*$ we have $r_{A-A}(x) \geq \tilde r(x) / 2$. Further observe that
\begin{align*}
    \sum_{x \in D^*} \tilde r(x)
    &\geq \sum_{x \in D} \tilde r(x) - \sum_{x \in D \setminus D^*} \tilde r(x) \\
    &\geq \sum_{x \in D} \tilde r(x) - |D| \cdot \frac{\alpha |A|^2}{4 |D|} \\
    &\geq \sum_{x \in D} r_{A-A}(x) - |D| \cdot \frac{\alpha |A|^2}{8 |D|}- |D| \cdot \frac{\alpha |A|^2}{4 |D|} \\
    &\geq \frac{\alpha |A|^2}{2}.
\end{align*}
Now let $L$ be the smallest integer with \smash{$2^{-L} < \frac{\alpha |A|}{2 |D|}$}, and consider the partition of $D^*$ into $D_1, \dots, D_L$ defined by
\begin{equation*}
    D_\ell = \set*{x \in D^* : \frac{|A|}{2^\ell} < \tilde r(x) \leq \frac{|A|}{2^{\ell-1}}}.
\end{equation*}
By the pigeonhole principle at least one part $1 \leq \ell \leq L$ satisfies
\begin{equation*}
    \sum_{x \in D_\ell} \tilde r(x) \geq \frac{\alpha |A|^2}{2 L};
\end{equation*}
take that part and report $D' = D_\ell$. Indeed, each $x \in D_\ell$ satisfies
\begin{equation*}
    r_{A-A}(x) \geq \frac{\tilde r(x)}{2} \geq \frac{|A|}{2^{\ell+1}} \geq \frac{\sum_{x \in D_\ell} \tilde r(x)}{4 |D_\ell|} \geq \frac{\alpha |A|^2}{8 L |D_\ell|},
\end{equation*}
which is as claimed, recalling that \smash{$L \leq 1 + \log \frac{1}{\alpha} + \log (1 + \frac{|D|}{|A|})$}.

Finally, consider the running time. \cref{lem:approx-mult} takes time $\tO(|A| / \epsilon + |D|) = \tO(|A| + |D| / \alpha)$. All subsequent steps, including the construction of $D^*$, its partitioning, and selecting the relevant part~$D_\ell$, run in linear time $O(|D|)$.
\end{proof}

We also rely on the following lemma, which is a combination of \cite[Lemma~2.33]{TaoV06} and \cite[Lemma~2.34]{TaoV06}.

\begin{lemma} \label{lem:asym-bsg-iter}
Let $0 < \alpha \leq 1$ and let $A, S$ be sets with $S \subseteq \Sym(A; \alpha)$. Then:
\begin{equation*}
    \abs{\set{(s_1, s_2) \in S^2 : s_1 - s_2 \in \Sym(A; \alpha^2 / 2)}} \geq \frac{\alpha^2 |S|^2}{2}.
\end{equation*}
\end{lemma}
\begin{proof}
Supposing otherwise,
\begin{align*}
    E(A, S)
    &= \sum_{s_1, s_2 \in S} r_{A-A}(s_1 - s_2) \\
    &< \frac{\alpha^2 |S|^2}{2} \cdot |A| + |S|^2 \cdot \frac{\alpha^2 |A|}{2} \\[1ex]
    &= \alpha^2 |S|^2 |A|.
\end{align*}
But at the same time, by Cauchy--Schwarz,
\begin{align*}
    E(A, S)
    &\geq \sum_{a \in A} r_{A+S}(a)^2 \\
    &\geq \frac{\parens*{\sum_{a \in A} r_{A+S}(a)}^2}{|A|} \\
    &= \frac{\parens*{\sum_{s \in S} r_{A-A}(s)}^2}{|A|} \\
    &\geq \frac{(|S| \cdot \alpha |A|)^2}{|A|} \\[1ex]
    &= \alpha^2 |S|^2 |A|,
\end{align*}
which shows the desired contradiction.
\end{proof}

\subsection{Existential Proof}
In this subsection we finally prove the existential statement of \cref{thm:asym-bsg}. Recall that we start from integer sets $A, B$ with energy
\begin{equation*}
    E(A, B) \geq \alpha |A|^2 |B|.
\end{equation*}
The proof can conceptually be split into four steps. In the first step we will successively replace~$A$ by ``smoother'' sets $A_1, A_2, \dots$ that admit more additive structure. As an instructive example, consider the set $B = \set{1, \dots, |B|}$ and let $A$ be a random subset of $B$. When $B$ is large compared to $A$, $A$ essentially admits no additive structure at all. We will therefore replace $A$ by smoother sets (think of $A - A, 2A - 2A, \dots$) converging to the highly structured set $B$. Once we reach a set that is sufficiently structured, in the second step we apply the symmetric BSG theorem to identify a large low-doubling subset $H$. In the third step we unroll the iterative smoothing procedure to relate $H$ to the initial set $A$. In the fourth step we relate $H$ to $B$.

Following~\cite{TaoV06}, we introduce a useful piece of notation. We write $f = O_j(g)$ if there exists a positive constant $C(j)$ (that only depends on $j$) such that
\begin{equation*}
    f \leq C(j) \cdot \alpha^{-C(j)} \cdot \parens*{1 + \log \frac{|B|}{|A|}}^{C(j)} \cdot g;
\end{equation*}
that is, the $O_j$ notation hides constant factors, polynomial factors in $1/\alpha$, and polylogarithmic factors in $|B| / |A|$. Similarly define $f = \Omega_j(g)$.

\paragraph{Step 1: Iteratively Smoothing $A$}
Let $J = J(\epsilon)$ be a constant to be fixed later. Let $\alpha_1 = \alpha/2$ and recursively define $\alpha_2, \dots, \alpha_J$ by $\alpha_{j+1} = \alpha_j^2 / 2$. Note that $\alpha_1 > \dots > \alpha_J > 0$ and $\alpha_j = \Omega_j(1)$. We will inductively construct a sequence of sets $A_0, A_1, \dots, A_J$ with the following properties:
\begin{enumerate}[label=(\roman*)]
    \item for all $1 \leq j \leq J$, $A_j \subseteq \Sym(B; \alpha_j)$, \label{thm:asym-bsg:itm:property-sym}
    \item for all $0 \leq j < J$ and all $x \in A_{j+1}$, $r_{A_j-A_j}(x) \geq \Omega_j\parens{|A_j|^2 / |A_{j+1}|}$, \label{thm:asym-bsg:itm:property-pop}
    \item for all $0 \leq j \leq J$, $|A_j| = \Omega_j(|A|)$. \label{thm:asym-bsg:itm:property-size}
\end{enumerate}
In fact, Property~\ref{thm:asym-bsg:itm:property-size} is always implied by Property~\ref{thm:asym-bsg:itm:property-pop}: If $\sum_{x \in A_{j+1}} r_{A_j - A_j}(x) \geq \Omega_j(|A_j|^2)$, then we must have $|A_{j+1}| \geq \Omega_j(|A_j|)$, and then Property~\ref{thm:asym-bsg:itm:property-size} follows by induction. For this reason we will only ensure Properties~\ref{thm:asym-bsg:itm:property-sym} and~\ref{thm:asym-bsg:itm:property-pop} in the following construction.

The base case is concerned with the construction of $A_0$ and $A_1$. We trivially take $A_0 = A$. To construct $A_1$, first observe that
\begin{align*}
    E(A, B)
    &= \sum_{a_1, a_2 \in A} r_{B-B}(a_1 - a_2) \\
    &\leq \sum_{a_1, a_2 \in A} \frac{\alpha |B|}{2} + \abs{\set{(a_1, a_2) \in A^2 : a_1 - a_2 \in \Sym(B; \alpha/2)}} \cdot |B|.
\end{align*}
In particular, recalling that $E(A, B) \geq \alpha |A|^2 |B|$, and letting $D = \Sym(B; \alpha/2)$, we must have
\begin{equation*}
    \sum_{x \in D} r_{A-A}(x) \geq \frac{\alpha |A|^2}{2}.
\end{equation*}
Hence, \cref{lem:dyadic-pigeonhole} yields a set $A_1 \subseteq D = \Sym(B; \alpha/2)$ that directly satisfies both Properties~\ref{thm:asym-bsg:itm:property-sym} and~\ref{thm:asym-bsg:itm:property-pop}.

Now focus on the inductive case, i.e., let $j \geq 1$, assume that $A_j$ has already been constructed, and focus on the construction of $A_{j+1}$. By~\ref{thm:asym-bsg:itm:property-sym} we have $A_j \subseteq \Sym(B; \alpha_j)$. Apply \cref{lem:asym-bsg-iter} with~$\alpha_j$ and sets $B, A_j$, yielding that
\begin{equation*}
    \abs{\set{(a_1, a_2) \in A_j^2 : a_1 - a_2 \in \Sym(B; \alpha_j^2 / 2)}} \geq \frac{\alpha_j^2 |A_j|^2}{2}.
\end{equation*}
Next, apply \cref{lem:dyadic-pigeonhole} with parameter $\alpha_{j+1} = \alpha_j^2 / 2$ on the sets $A_j, \Sym(B; \alpha_{j+1})$ to construct a set $A_{j+1} \subseteq \Sym(B; \alpha_{j+1})$ satisfying, for all $x \in A_{j+1}$,
\begin{equation*}
    r_{A_j-A_j}(x) = \Omega\parens*{\frac{\alpha_j}{1 + \log \frac{1}{\alpha_j} + \log(1 + \frac{|\Sym(B; \alpha_{j+1})|}{|A_j|})} \cdot \frac{|A_j|^2}{|A_{j+1}|}} = \Omega_j\parens*{\frac{|A_j|^2}{|A_{j+1}|}}.
\end{equation*}
In the last step we used that $|A_j| = \Omega_j(|A|)$ (by~\ref{thm:asym-bsg:itm:property-size}) and $|\Sym(B; \alpha_{j+1})| \leq |B| / \alpha_{j+1}$; taken together these imply that \smash{$\log(1 + |\Sym(B; \alpha_{j+1})| / |A_j|)$} can be hidden in the $\Omega_j$ notation. Evidently, the set~$A_{j+1}$ satisfies both Properties~\ref{thm:asym-bsg:itm:property-sym} and~\ref{thm:asym-bsg:itm:property-pop}, completing the construction.

\paragraph{Step 2: Constructing $H$ by the BSG Theorem}
Take $0 \leq j < J$ minimizing $|A_{j+1}| / |A_j|$. Recall that $|A_0| = |A|$ and that $|A_J| \leq |\Sym(B; \alpha_J)| \leq |B| / \alpha_J = O_J(|B|)$ by Property~\ref{thm:asym-bsg:itm:property-sym}. Thus, writing $L = 1 + |B| / |A|$, we must have
\begin{equation*}
    |A_{j+1}| = O_J\parens*{L^{1/J} \cdot |A_j|}.
\end{equation*}
Moreover, Property~\ref{thm:asym-bsg:itm:property-pop} implies that
\begin{align*}
    E(A_j)
    &= \sum_{x \in A_j - A_j} r_{A_j - A_j}(x)^2 \\
    &\geq \sum_{x \in A_{j+1}} r_{A_j - A_j}(x)^2 \\
    &= \Omega_j\parens*{\frac{|A_j|^4}{|A_{j+1}|}} \\
    &= \Omega_J\parens*{\frac{|A_j|^3}{L^{1/J}}}.
\end{align*}
In this situation, the symmetric BSG theorem (\cref{thm:bsg-symm}) applies, yielding some $K_j = O_J(L^{O(1/J)})$ and a set $H$ with
\begin{equation} \label{thm:asym-bsg:eq:H-size}
    |H| \leq K_j |A_j|,
\end{equation}
and
\begin{equation} \label{thm:asym-bsg:eq:H-doubling}
    |H + H| \leq K_j |H|,
\end{equation}
and
\begin{equation} \label{thm:asym-bsg:eq:H-cover-A_j}
    |A_j \cap H| \geq \frac{|A_j|}{K_j}.
\end{equation}
\Cref{thm:asym-bsg:eq:H-doubling} is claimed as in the theorem statement (once we fix the parameter $J$). \Cref{thm:asym-bsg:eq:H-size,thm:asym-bsg:eq:H-cover-A_j} involve $A_j$ rather than $A$; in the following we show that these imply similar statements for~$A$.

\paragraph{Step 3: Relating $H$ to $A$}
To relate $H$ to $A$, we traverse the sets $A = A_0, A_1, \dots, A_j$ \emph{backwards}. More precisely, we deduce from $|A_{j} \cap (H + x_{j})| \geq |A_{j}| / K_{j}$ that $|A_{j-1} \cap (H + x_{j-1})| \geq |A_{j-1}| / K_{j-1}$, for some shifts $x_{j-1}$ and parameters $K_{j-1}$ (we have $x_j = 0$). Indeed, assuming $|A_{j} \cap (H + x_{j})| \geq |A_{j}| / K_{j}$, Property~\ref{thm:asym-bsg:itm:property-pop} implies that
\begin{equation*}
    \abs{\set{(a_1, a_2) \in A_{j-1}^2 : a_1 - a_2 \in (H + x_{j})}} \geq \Omega_{j}\parens*{\frac{|A_{j-1}|^2}{K_{j}}},
\end{equation*}
or equivalently,
\begin{equation*}
    \sum_{a \in A_{j-1}} |A_{j-1} \cap (H + x_{j} + a)| \geq \Omega_{j}\parens*{\frac{|A_{j-1}|^2}{K_{j}}}.
\end{equation*}
In particular, there is some $a$ such that $|A_{j-1} \cap (H + x_{j} + a)| \geq \Omega_{j}(|A_{j-1}| / K_{j})$. Choosing $x_{j-1} = x_{j} + a$ and $K_{j-1} = O_j(K_{j})$, the claim follows. By induction we learn that
\begin{equation} \label{thm:asym-bsg:eq:H-cover-A}
    |A \cap (H + x_0)| \geq \frac{|A|}{K_0},
\end{equation}
for some translate $x_0$ and some parameter
\begin{equation*}
    K_0 = O_j(K_j) = O_J(L^{O(1/J)}). 
\end{equation*}
\cref{thm:asym-bsg:eq:H-cover-A} is exactly as claimed in the theorem statement, when shifting $H$ by $x_0$.

\paragraph{Step 4: Relating $H$ to $B$}
We finally relate $H$ to $B$. Combining Property~\ref{thm:asym-bsg:itm:property-sym} with \cref{thm:asym-bsg:eq:H-size,thm:asym-bsg:eq:H-cover-A_j} implies that
\begin{equation*}
    |H \cap \Sym(B; \alpha_j)| \geq \frac{|H|}{K_j^2}.
\end{equation*}
In particular, by definition of $\Sym(B; \alpha_j)$ we have
\begin{equation} \label{thm:asym-bsg:eq:B+H}
    \abs{\set{(b, h) \in B \times H : b + h \in B}} \geq \frac{\alpha_j}{K_j^2} \cdot |B| \, |H|.
\end{equation}
Let $\Lambda \subseteq B$ be a random subset, sampled with rate $p = \min(1, 4 K_j^2 / (\alpha_j |H|))$. Then, with constant probability,
\begin{equation} \label{thm:asym-bsg:eq:H-cover-B}
    |B \cap (H + \Lambda)| \geq \frac{\alpha_j}{4 K_j^2} \cdot |B|.
\end{equation}
To see this, call an element $b \in B$ \emph{relevant} if
\begin{equation*}
    r_{B+H}(b) \geq \frac{\alpha_j}{2 K_j^2} \cdot |H|.
\end{equation*}
\cref{thm:asym-bsg:eq:B+H} entails that there must be at least $\alpha_j / 2K_j^2 \cdot |B|$ relevant elements. But each relevant item is not covered in $H + \Lambda$ with probability at most $(1 - p)^{\frac{\alpha_j}{2K_j^2} \cdot |H|}$. When $p = 1$, this trivially equals $0$; otherwise, 
\begin{equation*}
    (1 - p)^{\frac{\alpha_j}{2K_j^2} \cdot |H|} \leq \exp(-p \cdot \frac{\alpha_j}{2K_j^2} \cdot |H|) \leq \exp(-2) < \frac{1}{4}.
\end{equation*}
By Markov's inequality, with probability at least $1/2$, at least half the relevant elements are covered by $H + \Lambda$. This proves \cref{thm:asym-bsg:eq:H-cover-B}. At the same time, $\Lambda$ satisfies the following size bound with high probability:
\begin{equation} \label{thm:asym-bsg:eq:Lambda-size}
    |\Lambda| \leq O\parens*{\frac{K_j^2}{\alpha_j} \cdot \frac{|B|}{|H|}}.
\end{equation}

\paragraph{Wrapping Up}
All in all, \cref{thm:asym-bsg:eq:H-doubling,thm:asym-bsg:eq:Lambda-size,thm:asym-bsg:eq:H-cover-A,thm:asym-bsg:eq:H-cover-B} are as claimed in \cref{thm:asym-bsg}; it only remains to choose appropriate parameters. All these equations have an overhead of the form $O_J(L^{O(1/J)})$, i.e.,
\begin{equation*}
    C(J) \cdot \alpha^{-C(J)} \cdot L^{C_0 / J} \cdot (\log L)^{C(J)},
\end{equation*}
for some absolute constant $C_0$ and some constant $C(J)$ depending on $J$. Choosing $J > C_0 / \epsilon$, this overhead indeed takes the form
\begin{equation*}
    K = O_\epsilon(\alpha^{-O_\epsilon(1)} L^\epsilon).
\end{equation*}
The proof of \cref{thm:asym-bsg} is complete.

\subsection{Efficient Algorithm}
In this section we finally argue that the previous proof can be implemented by an efficient algorithm, running in time $\tO((|A| + |B|) K)$ for some $K \leq O_\epsilon(\alpha^{-O_\epsilon(1)} L^\epsilon)$. To this end, we revisit the four steps from before and comment on their implementations.

\paragraph{Step 1: Iteratively Smoothing $A$}
In the first step we explicitly compute the sets $A_1, \dots, A_J$. This step is already almost completely algorithmic. Each set $A_{j+1}$ is constructed by one call to \cref{lem:dyadic-pigeonhole} on input $A_j$ and $D = \Sym(B; \alpha_{j+1})$; the only problem is that we do not have access to $\Sym(B; \alpha_{j+1})$. Instead, we will compute some superset
\begin{equation*}
    \Sym(B; \alpha_{j+1}) \subseteq D \subseteq \Sym(B; \alpha_{j+1}/2)
\end{equation*}
by \cref{lem:approx-pop-sums} in time $\tO(|B| / \alpha_{j+1})$. This leaves the correctness unchanged. The running time of \cref{lem:dyadic-pigeonhole} can then be bounded by
\begin{equation*}
    \tO\parens*{|A_j| + \frac{|D|}{\alpha_{j+1}}} = \tO\parens*{\frac{|B|}{\alpha_{j+1}^2}}.
\end{equation*}
The total time of step 1 is dominated by the final iteration, $\tO(|B| / \alpha_J^2)$, which is within the budget.

\paragraph{Step 2: Constructing $H$ by the BSG Theorem}
In step 2 we first identify the index $j$ minimizing $|A_{j+1}| / |A_j|$. Since we have computed the sets $A_1, \dots, A_J$ explicitly, this can be done in linear time. Then we apply the BSG theorem (\cref{thm:bsg-symm})
to the set $A_j$ to construct $H$. This runs in time $\tO(|A_j| K_j)$, for some $K_j = O_J(L^{O(1/J)})$ as before. Again, recalling that $|A_j| \leq |B| / \alpha_j$ by Property~\ref{thm:asym-bsg:itm:property-size}, this is within the claimed time budget.

\paragraph{Step 3: Relating $H$ to $A$}
The only data that we need to compute in this step is the shift $x_0$ satisfying \cref{thm:asym-bsg:eq:H-cover-A}, i.e.,
\begin{equation*}
    |A \cap (H + x_0)| \geq \frac{|A|}{K_0}.
\end{equation*}
In fact, we settle for a slightly weaker condition. Namely, observe that this shift $x_0$ is equivalently an element of $\Pop(A, -H; |A| / K_0)$. Hence, to find $x_0$ we may run \cref{lem:approx-pop-sums} to compute a set sandwiched between $\Pop(A, -H; |A| / K_0)$ and $\Pop(A, -H; |A| / 2 K_0)$, and pick any element $x_0$ from that set. By definition, that element satisfies
\begin{equation*}
    |A \cap (H + x_0)| \geq \frac{|A|}{2K_0}.
\end{equation*}
We may increase $K_0$ to $2K_0$ without any harm, and then $x_0$ is as claimed in \cref{thm:asym-bsg:eq:H-cover-A}.

The running time is dominated by the call to \cref{lem:approx-pop-sums} in time $\tO(|H| K_0)$. By \cref{thm:asym-bsg:eq:H-size} and Property~\ref{thm:asym-bsg:itm:property-size} this is at most $\tO(|A_j| K_j K_0) = \tO(|B| K_j K_0 / \alpha_j)$, which is within the budget.

\paragraph{Step 4: Relating $H$ to $B$}
The final step only involves the construction of $\Lambda$. Recall that $\Lambda$ is constructed as a random subset of $B$, sampled with an appropriate rate $p$. We have argued that with constant probability, the size $|B \cap (H + \Lambda)|$ is at least as large as claimed in \cref{thm:asym-bsg:eq:H-cover-B}. The only modification in our algorithm is to amplify the success probability. That is, we sample $R = \Theta(\log |B|)$ sets $\Lambda_1, \dots, \Lambda_R$, compute the sumsets $H + \Lambda_1, \dots, H + \Lambda_R$ by brute-force, and report the set $\Lambda_r$ maximizing $|B \cap (H + \Lambda_r)|$. With high probability that choice satisfies \cref{thm:asym-bsg:eq:H-cover-B}. Moreover, with high probability all sets $\Lambda_1, \dots, \Lambda_R$ simultaneously satisfy the size bound from \cref{thm:asym-bsg:eq:Lambda-size},
\begin{equation*}
    |\Lambda_1|, \dots, |\Lambda_R| \leq O\parens*{\frac{K_j^2}{\alpha_j} \cdot \frac{|B|}{|H|}}.
\end{equation*}
The running time of step 4 is therefore bounded by
\begin{equation*}
    O\parens*{\sum_{r=1}^R \parens*{|B| + |H| \, |\Lambda_r|}} = \tO\parens*{\frac{K_j^2}{\alpha_j} \cdot |B|}.
\end{equation*}
This is within the claimed time budget.

\section*{Acknowledgements}
We are very grateful for the detailed feedback of an anonymous reviewer, especially for suggesting a substantial simplification of the proof of \cref{thm:k-sum-pop}.

C.\ J.\ is supported by the Miller Research Fellowship at the Miller Institute for Basic Research in Science, UC Berkeley. Y.\ X.\ is supported by NSF HDR TRIPODS Phase II grant 2217058 (EnCORE Institute).

\bibliographystyle{alphaurl}
\bibliography{main}

\end{document}